\documentclass[reqno]{amsart}

\usepackage{amsmath,amsthm,amsfonts,amssymb,srcltx}
\usepackage{graphicx} 
\usepackage{epstopdf}
\usepackage{enumerate}
\usepackage{url}
\usepackage{bm}

\usepackage{epsfig}
\usepackage{color}

\def\IZ{\mathbb {Z}}

\def\IR{\mathbb {R}}
\def\IC{\mathbb {C}}

\def\IS{\mathbb {S}}

\allowdisplaybreaks

\theoremstyle{plain}
  \newtheorem{prob}{Problem}[section]
  \newtheorem{conj}[prob]{Conjecture}
  
  \newtheorem{thm}[prob]{Theorem}
  \newtheorem{prop}[prob]{Proposition}
   \newtheorem{cor}[prob]{Corollary}
    \newtheorem{lemm}[prob]{Lemma}
    \newtheorem{defi}[prob]{Definition}
\theoremstyle{remark}
  \newtheorem{remark}[prob]{Remark}

\newtheorem{exam}[prob]{Example}

\makeatletter
 
 \@addtoreset{equation}{section}
\makeatother

\begin{document}

\author[H. Fuji]{Hiroyuki Fuji}
\address{
Faculty of Education\\ 
Kagawa University\\
Takamatsu 760-8522\\
Japan\\ AND 
QGM\\
Department of Mathematics\\
Aarhus University\\
DK-8000 Aarhus C\\
Denmark}
\email{fuji{\char'100}ed.kagawa-u.ac.jp}

\author[K. Iwaki]{Kohei Iwaki}
\address{Graduate School of Mathematics\\ Nagoya University\\
Nagoya 464-8602\\ Japan}
\email{iwaki{\char'100}math.nagoya-u.ac.jp}

\author[M. Manabe]{Masahide Manabe}
\address{Faculty of Physics\\
University of Warsaw\\
ul. Pasteura 5, 02-093 Warsaw\\
Poland\\ AND 
Max-Planck-Institut f\"ur Mathematik\\ 
Vivatsgasse 7\\
53111 Bonn\\
Germany}
\email{masahidemanabe{\char'100}gmail.com}

\author[I. Satake]{Ikuo Satake}
\address{
Faculty of Education\\ 
Kagawa University\\
Takamatsu 760-8522\\
Japan}
\email{satakeikuo{\char'100}ed.kagawa-u.ac.jp}

\title [Reconstructing GKZ via topological recursion]{Reconstructing GKZ via topological recursion}

\thanks{Acknowledgments: 
The authors thank Prof.\, Hiroshi Iritani who suggests his idea on 
the equivariant version of the Dubrovin's conjecture.
KI also thanks Dr.\, Fumihiko Sanda for fruitful discussion. 
HF and MM thank Prof. Piotr Su{\l}kowski for stimulating discussions and useful comments.
The research of HF and IS is supported by the
Grant-in-Aid for Challenging Research (Exploratory) [\# 17K18781].
The research of HF is also supported by the
Grant-in-Aid for Scientific Research(C)  [\# 17K05239], and 
Grant-in-Aid for Scientific Research(B)  [\# 16H03927]  from the Japan Ministry of Education, Culture, Sports, Science and Technology, and Fund for Promotion of Academic Research from Department of Education in Kagawa University. 
The research of KI is supported by the
Grant-in-Aid for 
JSPS KAKENHI KIBAN(S) [\# 16H06337],
Young Scientists Grant-in-Aid for (B)  [\# 16K17613]
from the Japan Ministry of Education, Culture,
Sports, Science and Technology.
The work of MM is supported by the ERC Starting Grant no. 335739 ``Quantum fields and knot homologies'' funded by the European Research Council under the European Union's Seventh Framework Programme.
The work of MM is also supported by Max-Planck-Institut f\"ur Mathematik in Bonn.}

\begin{abstract}
In this article, a novel description of the hypergeometric differential equation found from Gel'fand-Kapranov-Zelevinsky's system (referred to as \textit{GKZ equation}) for Givental's $J$-function in the Gromov-Witten theory will be proposed.
The GKZ equation involves a parameter $\hbar$, and we will reconstruct it as 
a quantum curve from the classical limit $\hbar\to 0$ via the topological recursion.
In this analysis, the spectral curve (referred to as \textit{GKZ curve}) plays a central role, and it can be described by the critical point set of the mirror Landau-Ginzburg potential.
Our novel description is derived via the duality relations of the string theories, and various physical interpretations suggest that the GKZ equation is identified with the quantum curve for the brane partition function in the cohomological limit.
As an application of our novel picture for the GKZ equation, we will discuss the Stokes phenomenon for the equivariant $\mathbb{C}\textbf{P}^{1}$ model, and the wall-crossing formula for the total Stokes matrix will be examined. And as a byproduct of this analysis, we will study Dubrovin's conjecture for this equivariant model.
\end{abstract}

\maketitle
\tableofcontents



\section{Introduction}


\subsection{Background}


The {\em Gel'fand-Kapranov-Zelevinsky (GKZ) equations} (or the A-hypergeometric equations) are special class of linear differential equations which can be regarded as a generalization of classical hypergeometric differential equations \cite{GGZ, GKZ}. 
It is well-known that the GKZ equations appear in the context of {\em mirror symmetry} and play crucial roles. 
In this paper, we focus on the mirror symmetry between Fano manifolds and Landau-Ginzburg models. 
On the A-model side, the GKZ equations appear as the {\em quantum differential equations}, which is satisfied by the {\em Givental's $J$-function} \cite{Givental1,Givental_eqv}. The $J$-function is defined by the equivariant Gromov-Witten theory that captures the product structure of the quantum cohomology ring of the target Fano manifold (e.g., \cite{Iritani15,CCIT}). 
On the other hand, the GKZ equations also appear as the {\em Gauss-Manin systems} on the B-model side. Natural solutions of the GKZ equations on the B-model side are given by {\em oscillatory integrals} with a Landau-Ginzburg superpotential and an appropriately chosen volume form called the {\em primitive form} \cite{Saito83}. 
Roughly speaking, the Fano manifold and Landau-Ginzburg model are said to be mirror dual when the associated GKZ equations are identical\footnote{In general, we can identify these differential equations after some coordinate change through the {\em mirror map}. This implies that the $J$-functions are written in terms of the oscillatory integrals. To describe precise relation between these objects, we need the notion of the {\em Gamma class} introduced in \cite{Iritani09,KKP08}; see \cite{GGI14,GI15} for details.}. 
%


On the other hand, the Eynard-Orantin's {\em topological recursion} was originally introduced as a recursive algorithm to compute the $1/N$-expansion of the correlation functions and the partition function of matrix models from its spectral curve, and it is generalized to any algebraic curve which may not arise from a matrix model \cite{Eynard:2007kz}. 
In the topological recursion, we need a spectral curve $(C,x,y)$ as an initial input; where $C$ is a compact Riemann surface, $x, y$ are meromorphic functions on $C$ satisfying some conditions\footnote{Since spectral curves discussed in this paper are of genus $0$, we do not include the choice of $\omega_2^{(0)}$ in the definition of spectral curve.}. 
We may alternatively view the spectral curve $(C,x,y)$ as a meromorphic parametrization of a plane curve $A(x,y) = 0$ defined by a polynomial equation. 
Then we can define the correlators $\omega^{(g)}_n$ which are meromorphic differential forms on $C^{n}$ determined by the topological recursion relation \cite{Eynard:2007kz,Bouchard:2012yg}. 
Topological recursion attracts both mathematicians and physicists since the correlators are expected to encode the information of various enumerative or quantum invariants in mathematical physics; see \cite{Bouchard:2007ys, Eynard:2012nj, Dunin-Barkowski:2013wca, Fang:2016svw} for examples.


The aim of this paper is to give a new construction of a class of GKZ equations arising from the equivariant Gromov-Witten theory of projective spaces via the topological recursion\footnote{
A relationship between Gromov-Witten theory (or Frobenius structures) and the topological recursion has already been discussed in \cite{DOSS, FLZ} etc. However, our viewpoint is different from that of these works. We will focus on the reconstruction of GKZ equation as quantum curves.}. 
For the purpose, we will use the idea of {\em quantum curves} which relates the topological recursion correlators to the {\em WKB (formal) solution} of a certain Schr\"odinger-type differential or difference equation; see \cite{Gukov:2011qp, Mulase:2012tm, Dumitrescu:2013tca, Bouchard:2016obz} for example. A particular claim is given as follows: If we denote by $\omega^{(g)}_n$ the topological recursion correlators defined from a spectral curve $(C,x,y)$, then a generating series $\psi = \psi(x,\hbar)$ (an explicit formula will be given below) satisfies a Schr\"odinger-type differential or difference equation whose classical limit recovers the original spectral curve; thus the generating series $\psi$ is called the {\em wave function}, and the resulting Schr\"odinger-type equation is called the quantum curve. 
We will see that the GKZ equations are reconstructed as quantum curves associated with certain spectral curves, which we call {\em GKZ curves}, that have geometric interpretations in both A-model and B-model. 

We will also discuss two other related topics. The first one is the analytic properties of the wave function,  which is known to be divergent. More precisely, we will investigate properties of {\em Stokes matrices} of the wave function arising from equivariant $\mathbb{C}\textbf{P}^1$ model, with the aid of the {\em exact WKB analysis} \cite{Voros83, KT05}. The other topic is the physical interpretation of our main result (Theorem \ref{thm:reconstruction} in the next subsection). We will see that various physical dualities (geometric engineering, remodeling conjecture, mirror symmetry) allow us to understand GKZ equations as quantum curves.

We will describe some details of our results in the following subsections.


\subsection{Main result: GKZ equations as quantum curves}
\label{subsection:result-reconstruction}
Here we describe our main result. We will focus on the GKZ equations arising from 
the equivariant Gromov-Witten theory with the target space $X$ which is either
\begin{itemize}
\item the projective spaces $\mathbb{C}\textbf{P}^{N-1}_{\bm w}$, or 
\item smooth Fano complete intersections  $X_{\bm{w},\bm{\lambda}}$ of $n$ $(<N)$ degree 1 hypersufaces inside $\mathbb{C}\textbf{P}^{N-1}$. 
\end{itemize}
The subscripts in the above notations indicate the equivariant parameters ${\bm w} = (w_0,\dots,w_{N-1}) \in {\mathbb C}^{N}$ and $\bm{\lambda} = (\lambda_1, \dots, \lambda_n) \in {\mathbb C}^n$ with respect to the torus actions on the target space. 

For these cases, the explicit form of the $J$-functions and the associated GKZ equations (i.e., the quantum differential equations) are obtained in \cite{Givental1, Givental_eqv} (see also Proposition \ref{prop:J_function} and Proposition \ref{prop:GKZ-for-J-functions}). We may observe that the GKZ equations are of the form  
\begin{equation} \label{eq:GKZ-eq-introduction}
\widehat{A}_X(\widehat{x}, \widehat{y}) J_X(x) = 0,
\end{equation}
where $\widehat{x}$ and $\widehat{y}$ act on the $J$-function by 
\begin{equation}\label{eq:190525_1}
\widehat{x}J_X(x)=xJ_X(x),\qquad
\widehat{y}J_X(x)=\hbar x\frac{d}{dx}J_X(x). 
\end{equation}
Here $\hbar$ is a parameter which is contained in the $J$-function, and will play the role of Planck constant in our reconstruction of GKZ equation via the topological recursion. 
On the $B$-model side (i.e., the Landau-Ginzburg model side) of the mirror symmetry, the GKZ equation \eqref{eq:GKZ-eq-introduction} is also regarded as the equation satisfied by the oscillatory integral $\mathcal{I}_X(x)$ (whose definition will be given in Section \ref{subsection:oscillatory-integral-GKZ})
with the mirror superpotential $W_X(\cdot;x) : (\mathbb{C}^{\ast})^k \to \mathbb{C}$ 
(see \cite{Givental1, Givental_eqv}). 
Note that, in the equivariant case, the mirror Landau-Ginzburg potential $W_X$ contains logarithms, and hence, it is a multi-valued function (see Section \ref{subsection:oscillatory-integral-GKZ} for explicit expressions). We will call $\mathcal{I}_X(x)$ the ``\textit{equivariant oscillatory integral}''.

Since the GKZ equation \eqref{eq:GKZ-eq-introduction} is a Schr\"odinger-type equation, we can take the classical limit $A_X(x,y) \in {\mathbb C}[x,y]$ of the differential operator $\widehat{A}_X(\widehat{x}, \widehat{y})$. See Section \ref{subsection:GKZ-curves-critical-set} for the explicit expression of $A_X(x,y)$. 
We call the algebraic curve 
$A_X(x,y) = 0$ the {\em GKZ curve}, and denote it by $\Sigma_X$. 

Our first observation is that all the GKZ curves $\Sigma_X$ for $X = \mathbb{C}\textbf{P}^{N-1}_{\bm w}$ and $X_{\bm{w},\bm{\lambda}}$ are of genus 
$0$; therefore, we may find a pair of rational functions $x=x(z)$ and $y=y(z)$ which parametrize the GKZ curve. For example, the GKZ curve  
\begin{equation}
A_{X_{\bm{w},\bm{\lambda}}}(x,y) = \prod_{i=0}^{N-1}(y-w_i)-x\prod_{a=1}^n(y-\lambda_a) = 0,
\label{gkz_curve_into}
\end{equation}
for the Fano complete intersection $X = X_{\bm{w},\bm{\lambda}}$ is described by the pair 
\begin{equation}
x(z)=\frac{\prod_{i=0}^{N-1}(z-w_i)}{\prod_{a=1}^{n}(z-\lambda_a)},\qquad y(z)=z.
\label{gkz_curve_para}
\end{equation}
Thus, we may regard the GKZ curve as a spectral curve $(C={\mathbb C}{\bf P}^1, x(z), y(z))$ for the topological recursion. The correlators defined through the topological recursion \cite{Eynard:2007kz,Bouchard:2012yg} are denoted by $\omega_n^{(g)}$ (see Section \ref{subsec:T-Rec} for the definition).  

Moreover, we also verify that our GKZ curves satisfy the ``admissibility condition" proposed by Bouchard-Eynard in \cite{Bouchard:2016obz}. 
Therefore, the result of \cite{Bouchard:2016obz} allows us to find a Schr\"odinger-type differential equation, with the GKZ curve being its classical limit (thus the Schr\"odinger-type equation is a quantum curve for $\Sigma_X$), which annihilates the \textit{wave function} $\psi(D)$ defined by the following WKB-type formal series (see Definition \ref{def:rec_wkb}):
\begin{align}
&\psi(D)=\exp\Bigg[\frac{1}{\hbar}\int_D\widehat{\omega}_{1}^{(0)}+\frac{1}{2}\int_D\int_D\widehat{\omega}_{2}^{(0)}+\sum_{\substack{g\ge0,n\ge 1\\g,n\ne (0,1),(0,2)}}\frac{\hbar^{2g+n-2}}{n!}
\int_D\cdots\int_D\omega_n^{(g)}\Bigg].
\label{wave_function0}
\end{align}
Here $\widehat{\omega}_1^{(0)}$ and $\widehat{\omega}_2^{(0)}$ are some modification of correlators $\omega_1^{(0)}$ and $\omega_2^{(0)}$, respectively, and $D$ denotes an integration divisor on $C = {\mathbb C}{\bf P}^1$. The divisor $D$ contains the point $z = z(x)$ (which is a branch of the inverse function of $x=x(z)$ defined away from branch points) as its one of 
the end-points, and we regard $\psi(D)$ as an exponentiated formal series of $\hbar$ whose coefficients are functions of $x$. We denote by $\psi_X(D)$ the wave function for the GKZ curve $\Sigma_X$ and an integration divisor $D$. 

Then, our main result is formulated as follows. 
\begin{thm}[Theorem \ref{theorem:reconstruction-of-GKZ-general}, \ref{theorem:reconstruction-of-GKZ-general_p}]\label{thm:reconstruction}
Let $\Sigma_X$ be the GKZ curve \eqref{gkz_curve_into} 
defined as the classical limit of the GKZ equation \eqref{eq:GKZ-eq-introduction} 
arising from the equivariant Gromov-Witten theory of the smooth Fano complete intersection $X = X_{\bm{w},\bm{\lambda}}$ of $n$ ($<N$) degree $1$ hypersurfaces in $\mathbb{C}\textbf{P}^{N-1}$. (We regard $X = \mathbb{C}\textbf{P}^{N-1}_{\bm w}$ as the case of $n=0$.) 
Let us also denote by $\psi_X(x) = \psi_{X}(D)$ the wave function defined as the generating series \eqref{wave_function0} of topological recursion correlators for the GKZ curve $\Sigma_X$ with the integration divisor $D=[z]-[\infty]$ for the parametrization \eqref{gkz_curve_para}.
Then, the wave function $\psi_X(x)$ satisfies the GKZ equation which is satisfied by the $J$-function $J_X$ and the mirror equivariant oscillatory integral ${\mathcal I}_X$.
\end{thm}

This theorem claims that the GKZ equations \eqref{eq:GKZ-eq-introduction} are reconstructible from the GKZ curve $\Sigma_X$ through the topological recursion. 
This gives a new interpretation of the GKZ equations as quantum curves. 
Note that the result does not directly follow from the one by \cite{Bouchard:2016obz}:  
It is also pointed in \cite{Bouchard:2016obz} that the quantum curve is not unique because it may admit $\hbar$-correction terms depending on the choice of integration divisor $D$.  
Our computation shows that there exists an appropriate choice of the integration divisor $D$ which realizes the GKZ equation as a quantum curve. 

This theorem also implies that the WKB-type formal solution of the GKZ equation (when we regard $\hbar$ as the Planck constant, which is sufficiently small) is computed from the topological recursion. 
The WKB solution has an alternative meaning in the B-model description; it agrees with the saddle point expansion of the oscillatory integral $\mathcal{I}_X(x)$ up to some normalization factor. From this viewpoint, we may also identify the GKZ curve $\Sigma_X$ with the critical set of $W_{X}$ (see Proposition \ref{def:GKZ_curve}). 
For oscillatory integrals associated with critical points satisfying a condition given in Section \ref{subsection:relation-to-oscillatory-integral-section4}, we can specify the factor and obtain a full-order coincidence of the asymptotic expansion of $\mathcal{I}_X(x)$ and the wave function (see Corollary \ref{cor:relation-oscillatory-integral-and-wave-function}). This is an application of our result. 

Actually, the formal solution \eqref{wave_function0} has already been discussed by Dubrovin in his theory of {\em Frobenius structure} \cite{dubrovin} (also called the Saito's {\em flat structure} \cite{Saito83} in the context of singularity theory) that arises from {\em non-equivariant} Gromov-Witten theory\footnote{In view of the GKZ equation \eqref{eq:GKZ-eq-introduction}, this is regarded as the case where all equivariant parameters $w_i$ and $\lambda_a$ are equal to $0$.}. 
In non-equivariant case, it is known that the GKZ equation can be embedded into a holonomic system of partial differential equations in both variables $x$ and $\hbar$, called {\em Dubrovin's first structure connection} associated with the Frobenius structure. 
The differential equation with respect to $\hbar$ has an irregular singular point at $\hbar = 0$, and the above WKB solution coincides with the formal solution around the point up to normalization constants. This also explains the divergence of the WKB solution. 
The celebrated Dubrovin's conjecture \cite{dubrovin-conj} is a statement for the Stokes matrix around $\hbar=0$ (we will come back to this point below). 
Our second result presented in next subsection (and Section \ref{section:Stokes-eqP1} in detail) will be regarded as a generalization of the analysis when the GKZ equation arises from {\em equivariant} Gromov-Witten theory (where the differential equation with respect to $\hbar$ is not available in general). 


\subsection{Exact WKB analysis and Dubrovin's conjecture}

We will also focus on an analytic property of the wave function 
\eqref{wave_function0} defined through the topological recursion. 
More precisely, we will study the Stokes structure (with respect to the variable $\hbar$)
of the GKZ equation for $X = \mathbb{C}\textbf{P}^1_{\bm w}$
\begin{equation} \label{eq:equivariant-P1-intro}
\left[ 
\left( \hbar x \frac{d}{dx} - w_0 \right)
\left( \hbar x \frac{d}{dx} - w_1 \right) - x
\right] \psi = 0
\end{equation}
from the viewpoint of {\em exact WKB analysis}. 

Since the WKB solution \eqref{wave_function0} is usually divergent, we employ the {\em Borel summation method} which allows us to construct an analytic solution $\Psi$ of the quantum curve whose asymptotic expansion when $\hbar \to 0$ in a certain sectorial domain recovers the WKB solution (see \cite{Cos08}). This framework is called the exact WKB analysis, which was developed in \cite{Voros83,DDP93,KT05} etc. 
In general, the Borel sum of a divergent series has an exponentially small ambiguity due to the {\em Stokes phenomenon}, and the ambiguity is described by the {\em Stokes matrices}. For non-equivariant oscillatory integrals, Stokes matrices are computed by the intersection numbers of vanishing cycles via the Picard-Lefschetz formula (e.g., \cite{DH}); in physics side, they are regarded as contributions of {\em 2d BPS solitons} in the Landau-Ginzburg model with the superpotential $W$ studied by Cecotti-Vafa (\cite{CV93}; see also \cite{GMN-SN}).  Since the exact WKB method does not require any differential equation with respect to $\hbar$, it is applicable to study the Stokes phenomenon of \eqref{eq:equivariant-P1-intro} for small $\hbar$. Also, it allows us to describe the properties of Borel sums by a simple combinatorial object, called {\em Stokes graph} (see \cite{KT05}; it has a close relation to the {\em spectral networks} \cite{GMN-SN}). These are advantages of the use of the exact WKB method.

Our first result in Section \ref{section:Stokes-eqP1} is the computation of the {\em ``total" Stokes matrix}, which captures all Stokes matrices of the WKB solution. We observed that the total Stokes matrices contain a contribution from a loop-type Stokes curve which never appears in the non-equivariant situation (i.e., when $w_0-w_1 = 0$). The loop-type Stokes curve was analyzed in \cite{AIT, IN14} with a relation to cluster algebras. Our computation also includes a derivation of an example of {\em 2d-4d wall-crossing formula} established by Gaiotto-Moore-Neitzke \cite{GMN} through the exact WKB analysis of \eqref{eq:equivariant-P1-intro}. It is a consequence of the fact that the total Stokes matrix is locally constant when $x$ varies; the fact is straightforward from properties of the Borel sum of WKB solutions, but the locally constant-ness of the total Stokes matrix involves a non-trivial identity of an infinite product of Stokes matrices.

The second result in Section \ref{section:Stokes-eqP1} is an {\em equivariant version of the Dubrovin's conjecture} (c.f., \cite{dubrovin-conj}) for the equivariant ${\mathbb C}{\bf P}^1$. The Dubrovin's conjecture was originally formulated in non-equivariant situation as the coincidence of the entries of total Stokes matrices with the Euler pairing of coherent sheaves on the target Fano manifold. The original version of the conjecture was proved by Guzzetti \cite{Guz} for projective spaces, Ueda \cite{Ueda,Ueda2} for Grassmannians and smooth cubic surfaces; see also \cite{Iwaki_Takahashi,GGI14,GI15,SS} for related works on the conjecture. Our explicit computation of the total Stokes matrices in equivariant case implies that an entry of the total Stokes matrix of \eqref{eq:equivariant-P1-intro} coincides with the equivariant Euler pairing of the coherent sheaves ${\mathcal O}$ and ${\mathcal O}(1)$ on $\mathbb{C}\textbf{P}^1_{\bm w}$ (regarded as equivariant sheaves in an appropriate manner). Our claim also includes the original Dubrovin's conjecture for non-equivariant ${\mathbb C}{\bf P}^1$ by taking the non-equivariant limit $w_0 - w_1 \to 0$.

Here we emphasize that the holomorphicity of the correlation functions $\omega^{(g)}_{n}$ for $2g-2+n \ge 0$ of the topological recursion are crucially important in the computations of the total Stokes matrices. Our computation seems to suggest that the divergent series defined through the topological recursion have a rich Stokes structure. It would be interesting to investigate other examples arising from topological recursion and quantum curves.

\subsection{Physical perspective and string dualities}

Our result in Theorem \ref{thm:reconstruction} is similar 
and closely related to the  Bouchard-Klemm-Marino-Pasquetti's 
{\em remodeling conjecture} in \cite{Bouchard:2008gu}. 
We explain the background of our research from a mathematical-physics perspective.

From the point of view of the string theory, a quantum structure behind the Gromov-Witten theory has been considered in a different way \cite{Aganagic:2003qj,Iqbal:2003ds,Dijkgraaf:2008fh,Dijkgraaf:2007sw}. The string theoretical quantum structure emerges in the higher genus (open) string free energy in the topological A-model. 

For a special Lagrangian submanifold $L$ with $b_1(L)=1$ 
in a local toric Calabi-Yau 3-fold $Y$, it is argued that 
the \textit{brane partition function} is annihilated by a $q$-difference operator $\widehat{A}^K_Y$ \cite{Aganagic:2003qj,Aganagic:2011mi}:
\begin{align}
\widehat{A}^K_Y(\widehat{\mathsf{x}},\widehat{\mathsf{y}})\,
Z^Y_{\textrm{A-brane}}(\mathsf{x})=0,
\label{quantum_curve}
\end{align}
which is made of non-commutative operators
$\widehat{\mathsf{x}}$ and $\widehat{\mathsf{y}}$:
\begin{align}
\widehat{\mathsf{x}}\,Z^Y_{\textrm{A-brane}}(\mathsf{x})=\mathsf{x}\,Z^Y_{\textrm{A-brane}}(\mathsf{x}),\qquad \widehat{\mathsf{y}}\,Z^Y_{\textrm{A-brane}}(\mathsf{x})=Z^Y_{\textrm{A-brane}}(q\mathsf{x}),
\nonumber
\end{align}
with the $q$-Weyl relation 
$\widehat{\mathsf{y}}\, \widehat{\mathsf{x}}=q\, \widehat{\mathsf{x}}\, \widehat{\mathsf{y}}$, 
where $q=\mathrm{e}^{-g_s}$ and the parameter $g_s\in\mathbb{C}$ is called \textit{string coupling}. 
Here the brane partition function $Z^Y_{\mathrm{A}\textrm{-}\mathrm{brane}}(\mathsf{x})$ is defined 
in terms of the open string free energies as
\begin{align}
Z^Y_{\mathrm{A}\textrm{-}\mathrm{brane}}(\mathsf{x})=\exp\left[
\sum_{g\ge 0, n\ge 1}g_s^{2g+n-2}\, \frac{1}{n!}\, F_{n}^{(g)}(\mathsf{x},\ldots,\mathsf{x})
\right],
\label{a_brane_pf}
\end{align}
and the open string free energy \cite{Marino_book},
$$
F_{n}^{(g)}(\mathsf{x}_1,\ldots,\mathsf{x}_n)=\sum_{\beta\in H_2(Y,L)}\sum_{\bm{p}\in\mathbb{Z}^n}N^{(g)}_{\bm{p},\beta}\,
\mathrm{e}^{-\beta\cdot t}\,
\mathsf{x}_1^{p_1}\cdots \mathsf{x}_n^{p_n},
$$
gives the generating function of open Gromov-Witten 
invariants $N^{(g)}_{\bm{p},\beta}$ 
enumerating holomorphic maps in the topological class labeled by genus $g$, 
the class $\beta\in H_2(Y,L)$, and the winding numbers 
$\bm{p}=(p_1,\ldots,p_n) \in \mathbb{Z}^n$, 
where each integer $p_i$ specifies how many times 
the $i$-th boundary of a world-sheet Riemann surface $C_{g,n}$ wraps 
around the one-cycle in $L$. 
$t^i\in \mathbb{C}$ ($i=1,\ldots, \dim H_2(Y)$) denote the K\"ahler moduli parameters of $Y$. 
The $q$-difference equation (\ref{quantum_curve}) 
is interpreted as a quantum curve (see section \ref{rem:q_curve_brane} as an example) \cite{Dijkgraaf:2007sw,Dijkgraaf:2008fh}.
We then find that the quantum curve arises from a hidden quantum mechanical system behind 
the open topological strings. 

\begin{figure}[h]
\centering%
\includegraphics[width=25cm,height=9cm,keepaspectratio]{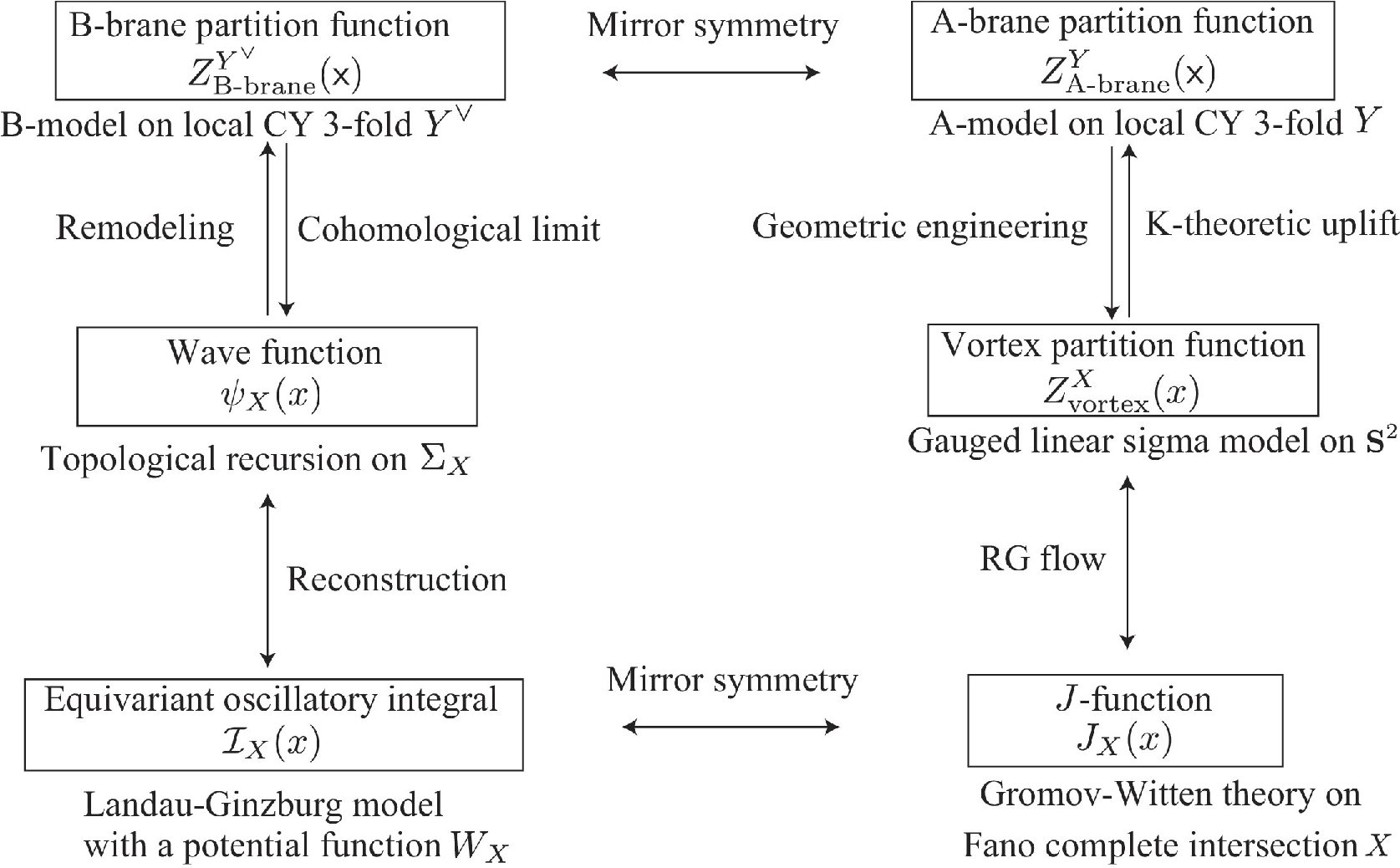}%
\caption{\label{fig:duality}String dualities relate the oscillatory integral with the open topological string partition function.}
\end{figure}

At this stage, we find two kinds of hidden quantum mechanical systems; 
one is behind the Givental's formulation of the Gromov-Witten theory of a compact smooth Fano manifold $X$, 
and the other is behind the open topological A-model on a local toric Calabi-Yau 3-fold $Y$. The quantum structure of these two theories are reflected in the GKZ equation for the $J$-function and the $q$-difference equation for the brane partition function, and the quantum deformation parameter appears as the \textit{Planck's constant} $\hbar$ and string coupling $g_s$, respectively.
Although we find some nice similarities\footnote{Such similarities were also pointed out in \cite{Nawata:2014nca}.} between these two theories, some crucial discrepancies also exist between them.
In the physical perspective, the Givental's formulation is considered essentially as the closed string theory on $X$. In this sense, the $J$-function $J_X(x)$ is the generating function of the genus $0$ correlation functions and $\hbar$ is associated with the degrees of freedom of the gravitational descendants.
On the other hand, the latter theory is considered as the open string theory on $Y$, and the brane partition function $Z^Y_{\textrm{A-brane}}{\color{black}(\mathsf{x})}$ is the 
generating function of {\color{black}all genus open string free energies} and $g_s$ is associated with the coupling to 2-dimensional gravity \cite{Bershadsky:1993cx}. Thus the $J$-function and the brane partition function are essentially different objects in a sense of the  string theory.

To overcome such discrepancies, we will employ the physical idea of \textit{string dualities} depicted in Figure \ref{fig:duality}. The punchline of the string dualities shown in the right hand side of this figure is found as follows. 
Considering the gauged linear sigma model \cite{Witten:1993yc} which describes the Gromov-Witten theory of $X$, we can reinterpret the $J$-function $J_X(x)$ as 
a vortex partition function $Z_{\textrm{vortex}}^{X}(x)$ on 
the 2-sphere $\mathbb{S}^2$ \cite{Dimofte:2010tz,Bonelli:2013mma}.
Furthermore, for the special case $X=\mathbb{C}\textbf{P}^{N-1}$ and Fano complete intersection of hypersurfaces with degrees $l_1=\ldots =l_n=1$, the vortex partition function of the gauged linear sigma model is realized as 
a brane partition function  $Z^Y_{\textrm{A-brane}}{\color{black}(\mathsf{x})}$ on a class of local toric Calabi-Yau 3-fold referred to as \textit{strip geometry} \cite{Iqbal:2004ne} (see Figure \ref{local_A_geom_m} for the toric diagram of a strip geometry) by the geometric engineering \cite{Dimofte:2010tz}. 
More precisely, the brane partition function  $Z^Y_{\textrm{A-brane}}{\color{black}(\mathsf{x})}$ realizes the K-theoretic version of the $J$-function, 
and we need to take the cohomological limit $\beta\to 0$ with reparametrization 
\begin{align}
g_s=\beta \hbar,\qquad
\mathsf{x}=\beta^{N-n}x,\qquad
\widetilde{Q}_{w,i}=\mathrm{e}^{-\beta (w_{0}-w_i)},\qquad
\widetilde{Q}_{\lambda,a}=\mathrm{e}^{-\beta (w_0-\lambda_{a})}.
\nonumber
\end{align}
Here $\widetilde{Q}_{w,i}$ ($i=0,\ldots,N-1$) and $\widetilde{Q}_{\lambda,a}$ $(a=1,\ldots,n)$ denote $\mathrm{e}^{t_{w,i}}$ and $\mathrm{e}^{t_{\lambda,a}}$ with the (reorganized) K\"ahler parameters $t_{w,i}$ and $t_{\lambda,a}$ of the strip geometry $Y$, respectively.
In fact for these examples, by direct computations, the equivariant $J$-function of $X$ and the brane partition function of $Y$ {\color{black}agree} in this limit,
and the $q$-difference equation (quantum curve) for the brane partition function reduces to the GKZ equation for the $J$-function.
Thus, as a result of string theoretical discussions, a novel interpretation of the GKZ equation as a quantum curve is uncovered.

Now we will further proceed with the string dualities by applying the mirror symmetry. Via the local mirror symmetry, the open topological A-model on 
the strip geometry $Y$ turns to  the open topological B-model on a local Calabi-Yau 3-fold
\begin{align}
Y^{\vee}=\big\{\;(\omega_+,\omega_-, \mathsf{x},\mathsf{y})\in{\IC}^2\times(\mathbb{C}^*)^2\;\big|\; \omega_+\omega_-=A_Y^K(\mathsf{x},\mathsf{y})\;\big\},
\nonumber
\end{align}
where $A_Y^K(\mathsf{x},\mathsf{y})\in \mathbb{C}[\mathsf{x},\mathsf{y}]$ is given by the classical limit of the $q$-difference operator $\widehat{A}_Y^K(\widehat{\mathsf{x}},\widehat{\mathsf{y}})$ such that
\begin{align}
q \to 1\ (g_s \to 0),\qquad 
(\widehat{\mathsf{x}},\widehat{\mathsf{y}})\to (\mathsf{x},\mathsf{y}),\qquad 
\widehat{A}_Y^K(\widehat{\mathsf{x}},\widehat{\mathsf{y}})\to A_Y^K(\mathsf{x},\mathsf{y})\in \mathbb{C}[\mathsf{x},\mathsf{y}].
\nonumber
\end{align}
For the topological B-model on this local geometry, the remodeling conjecture \cite{Bouchard:2008gu} proposed by
V.~Bouchard, A.~Klemm, M.~Mari\~no, and S.~Pasquetti
is applicable, and the open topological B-model is studied systematically on the basis of the topological recursion.
The key ingredient of this formalism is the spectral curve 
$\Sigma_{Y}^{K}=\big\{(\mathsf{x},\mathsf{y})\in ({\IC}^*)^2\big|A_{Y}^K(\mathsf{x},\mathsf{y})=0\big\}$, and higher genus open string free energies of the topological B-model are evaluated recursively. 
Then via the local mirror symmetry, one can (re)construct the brane partition function $Z^Y_{\textrm{A-brane}}(\mathsf{x})$ in  (\ref{a_brane_pf}). Therefore, by chasing the web of dualities in Figure \ref{fig:duality}, we perceive that the GKZ equation would be reconstructible as a quantum curve by the topological recursion.

\subsection{Plan of the paper}

The organization of this paper is as follows.
In Section \ref{sec:osc_gkz} we review how the GKZ equations arise in the equivariant Gromov-Witten theory (and the mirror Landau-Ginzburg models) of $\mathbb{C}\textbf{P}^{N-1}$ and smooth Fano complete intersections of hypersurfaces in $\mathbb{C}\textbf{P}^{N-1}$. We also define the GKZ curve as the classical limit of the GKZ equations. 
In Section \ref{sec-q_curve_rec} we will summarize the necessary ingredients of the topological recursion and quantum curve in a nut-shell.
In Section \ref{sec-gkz_reconst} the WKB reconstruction of the GKZ equation will be discussed. 
If the GKZ equation is the second order differential equation,
the (local) topological recursion \cite{Eynard:2007kz} is applicable and the quantum curve is found manifestly by employing the method developed by the work of M.~Mulase and P.~Su{\l}kowski \cite{Mulase:2012tm}.
In more general cases, we need to use the ``\textit{global topological recursion}'' \cite{Bouchard:2012an,Bouchard:2012yg,Bouchard:2016obz}, and a quantum curve is reconstructible only for spectral curves which satisfy the admissibility condition considered in the work of V.~Bouchard and B.~Eynard \cite{Bouchard:2016obz}. 
Among the GKZ curves, we show that the GKZ equation is reconstructible for $\mathbb{C}\textbf{P}^{N-1}$ and the Fano complete intersection of hypersurfaces with degree $l_a=1$ ($a=1,\ldots,n<N$) in $\mathbb{C}\textbf{P}^{N-1}$ (Theorem \ref{theorem:reconstruction-of-GKZ-general}, \ref{theorem:reconstruction-of-GKZ-general_p}). 
In Section \ref{sec-glsm_J} the string dualities behind our proposal will be discussed.
In particular, we focus on 3 different vantage points of the string theories, and the $J$-function is regarded as the vortex partition function and brane partition function {\color{black}in} the open topological A-model and B-model. 
As a result of the string dualities, we find that the GKZ curves considered in Section \ref{sec-gkz_reconst} are obtained as mirror curves in the open topological B-model, and the GKZ equations are found as quantum curves for the brane partition functions. 
In Section \ref{section:Stokes-eqP1} we compute the total Stokes matrices for the quantum curve arising from equivariant ${\mathbb C}{\bf P}^{1}$ by using the exact WKB method. We also examine a wall-crossing formula and equivariant version of the Dubrovin's conjecture in this particular case.

In Appendix \ref{app:J_function} we summarize the $J$-functions for $\mathbb{C}\textbf{P}^{N-1}$ and smooth Fano complete intersections of hypersurfaces in $\mathbb{C}\textbf{P}^{N-1}$. And as a side remark of our proposal, we derive the GKZ curve from the saddle point approximation of the $J$-function in the similar spirit as the generalized volume conjecture 
{\color{black}in knot theory \cite{Gukov:2003na,Hikami:2006cv}}.
In Appendix \ref{app:IGKZ} the GKZ equations for the equivariant oscillatory integrals are given for the mirror Landau-Ginzburg models of 
the equivariant Gromov-Witten theory of $\mathbb{C}\textbf{P}^{N-1}$ 
and smooth Fano complete intersections of hypersurfaces in $\mathbb{C}\textbf{P}^{N-1}$, and we prove Proposition \ref{prop:IGKZ}. 
In Appendix \ref{appendix:asymptotic-behavior-of -coefficients}, 
Proposition \ref{prop:behavior-of-saddle-point-approximation} is proved.
In Appendix \ref{section:computational-results} we will show some explicit computational results on the asymptotic solutions of GKZ equations. In particular, for the equivariant $\mathbb{C}\textbf{P}^{1}$ model, we also compute the free energies via the topological recursion and directly check the agreement with the asymptotic solutions of GKZ equation.

\section{GKZ equations in quantum cohomologies and oscillatory integrals}\label{sec:osc_gkz}


\subsection{$J$-function and GKZ equation}
Let $X$ be a smooth projective variety. 
The (small) quantum cohomology ring $QH^*(X)$ is a generalization of the ordinary cohomology ring $H^*(X)$
arising from a deformation of the cup product referred to as \textit{quantum-cup product}.
The quantum cup product is specified by the intersection indices of holomorphic curves in $X$
with cycles which are Poincar\'e dual to elements in $H^*(X)$, and such indices are known as \textit{Gromov-Witten invariants} of $X$.
In the physics language (see a pedagogical exposition in \cite{Hori_book}), the quantum cup product is realized by correlation functions for the cohomology elements of $X$.
\begin{defi}\label{def:correlation} 
Let $\overline{\mathcal{M}}_{g,n}(X,\beta)$ denote 
the moduli space of stable maps from connected genus $g$ curves $C$ with $n$-marked points $p_1,\ldots,p_n$ to $X$ representing the class $\beta\in H_2(X)$. It carries a virtual fundamental class denoted by $\left[\overline{\mathcal{M}}_{g,n}(X,\beta)\right]^{\mathrm{vir}}$. Given classes $\gamma_1,\ldots,\gamma_n\in H^*(X)$, the correlation function $\langle\gamma_1,\cdots,\gamma_n\rangle_{g,\beta}$
is defined by
\begin{align}
\langle\gamma_1, \cdots, \gamma_n\rangle_{g,\beta}
=\int_{\left[\overline{\mathcal{M}}_{g,n}(X,\beta)\right]^{\mathrm{vir}}}\mathrm{ev}^*_1(\gamma_1)\cup\cdots\cup\mathrm{ev}^*_n(\gamma_n),
\end{align}
where $\mathrm{ev}_i :\overline{\mathcal{M}}_{g,n}(X,\beta) \to X$ ($i=1,\ldots,n$) denotes the evaluation map at the $i$-th marked point such that
$\mathrm{ev}_i(C,p_1,\ldots,p_n,\phi)=\phi(p_i)$. Let $\mathcal{L}_i$ ($i=1,\ldots,n$) be the corresponding tautological line bundles over  $\overline{\mathcal{M}}_{g,n}(X,\beta)$. The correlation function for the gravitational descendants $\tau_k(\gamma)$ ($k\ge 0$) is defined by
\begin{align}
\langle \tau_{k_1}\gamma_1,\cdots ,\tau_{k_n}\gamma_n\rangle_{g,\beta}=\int_{\left[\overline{\mathcal{M}}_{g,n}(X,\beta)\right]^{\mathrm{vir}}}\prod_{i=1}^nc_1(\mathcal{L}_i)^{k_i}\cup\mathrm{ev}_i^*(\gamma_i).
\end{align}
\end{defi}

In celebrated works \cite{Givental1,Givental_eqv,Coates_Givental} by Givental, an elegant framework to uncover profound aspects of the Gromov-Witten theory and mirror symmetry was proposed on the basis of the concept of ``quantization". 
In this framework, a generating function of the genus $g=0$ correlation functions with $n=1$ marked point for the gravitational descendants referred to as \textit{$J$-function}
plays an important role. 
For our purpose, we investigate the restriction of the $J$-function to $H^2(X) \subset H^{\ast}(X)$, called {\em small $J$-function}. Taking generators $\beta_1, \dots, \beta_r \in H_2(X,{\mathbb Z})$, we identify $H_2(X,{\mathbb Z}) \cong {\mathbb Z}^{r}$ and denote its elements by ${\bm d}=(d_1,\dots,d_r)$. 
We also take a basis $T_0=1,T_1, \dots, T_m \in H^*(X)$ such that $T_1,\dots,T_r$ give a basis of $H^2(X)$ 
and $T^0,\dots,T^m \in H^*(X)$ are the dual basis with respect to the Poincar\'e pairing.
\begin{defi}[\cite{Givental1,Givental_eqv,Coates_Givental}]
\label{def:J-function}
The (small) $J$-function of the smooth projective variety $X$ is $H^{\ast}(X) \otimes {\mathbb C}[[\hbar^{-1}]]$-valued formal series defined as the generating function of the correlation function for $g=0$ gravitational descendants:
\begin{align}
J_X(\bm{x})&=\mathrm{e}^{(t_1 T_1 + \cdots t_r T_r)/\hbar}\left(1+\sum_{\bm{d}}\bm{x}^{\bm{d}}
\sum_{i=0}^{m} \sum_{k=0}^{\infty}\hbar^{-k-1}\left\langle \tau_k T_i \right\rangle_{0,\bm{d}} \cdot T^{i}
\right)
\nonumber \\
&=\mathrm{e}^{(t_1 T_1 + \cdots t_r T_r)/\hbar}\left(1+\sum_{\bm{d}}\bm{x}^{\bm{d}}
\mathrm{ev}_{1*}^{{\bm d}}\left(\frac{1}{\hbar-c_1(\mathcal{L}_1)}\right)\right).
\label{J_function}
\end{align}
Here $t_1,\dots,t_r$ denote the linear coordinates of $H^2(X)$ with respect to the basis $T_1, \dots, T_r$, ${\bm d} = (d_1, \dots, d_r)$ runs all $d_i \ge 0$, and $\bm{x^d} = x_1^{d_1} \cdots x_r^{d_r}$ where $x_i = \mathrm{e}^{t_i}$. 
Also $\mathrm{ev}_1^{{\bm d}}:\overline{\mathcal{M}}_{0,2}(X,\sum_{i=1}^{r}d_i\beta_i) \to X$ 
is the evaluation map at the 1st marked point. 
\end{defi}

In the following we consider the smooth Fano complete intersection $X=X_{\bm{l}}$ of hypersurfaces in the projective space $\mathbb{C}\textbf{P}^{N-1}$ given by $n$ equations of degrees $\bm{l}=(l_1, \ldots, l_n)$ with $l_1+\cdots +l_n<N$. According to \cite{Givental1,Givental_eqv}, 
the ($H^*(\mathbb{C}\textbf{P}^{N-1})$-valued) $J$-function 
is defined 
as
\begin{align}
&
J_{X_{\bm{l}}}(x)=\mathrm{e}^{t p/\hbar}\sum_{d=0}^{\infty}x^d\,\mathrm{ev}_{1*}\bigl(S_d(\hbar)\bigr), 
\qquad x=\mathrm{e}^{t}\in \mathbb{C}^*,
\label{J_function2}
\\
&
S_d(\hbar)=\frac{E_d}{\hbar-c_1(\mathcal{L}_1)}\in H^{*}(\overline{\mathcal{M}}_{0,2}(\mathbb{C}\textbf{P}^{N-1},d\beta)).
\nonumber
\end{align}
Here $p \in H^2(\mathbb{C}\textbf{P}^{N-1})$, $\beta \in H_2(\mathbb{C}\textbf{P}^{N-1})$ 
with $\langle \beta,p\rangle=1$, 
$E_d$ denotes the Euler class of the vector bundle over $\overline{\mathcal{M}}_{0,2}(\mathbb{C}\textbf{P}^{N-1},d\beta)$ with fiber $H^{0}(C,\phi^*H^{\otimes l_1}\oplus\cdots\oplus\phi^*H^{\otimes l_n})$, 
where $H^{\otimes l}$ is the $l$-th tensor power of the hyperplane line bundle over $\mathbb{C}\textbf{P}^{N-1}$. 
Also $\mathrm{ev}_1:\overline{\mathcal{M}}_{0,2}(\mathbb{C}\textbf{P}^{N-1},d\beta) \to 
\mathbb{C}\textbf{P}^{N-1}$ is the evaluation map at the 1st marked point. 

As a 
generalization, the equivariant counterpart to the Gromov-Witten theory was also considered in \cite{Givental_eqv}.
For the $N$-dimensional torus $T$, we consider the natural $T$-action 
on $\mathbb{C}\textbf{P}^{N-1}$. Then we have the $T$-equivariant cohomology algebra 
\begin{align}
H^*_T(\mathbb{C}\textbf{P}^{N-1})\cong \mathbb{C}[p,\bm{w}]/((p-w_0)\cdots(p-w_{N-1}))
\end{align}
over $H^*(BT)=\mathbb{C}[\bm{w}=(w_0,\ldots,w_{N-1})]$. 
In addition, the $n$-dimensional torus $T'$ action on the vector bundle $\oplus_{a=1}^nH^{\otimes l_a}$ with the equivariant parameters $\bm{\lambda}=(\lambda_1,\ldots,\lambda_n)$ provides the $T'$-equivariant Euler class $e_{T'}$ such that
\begin{align}
e_{T'}\left(\oplus_{a=1}^n H^{\otimes l_a}\right)=(l_1p-\lambda_1)\cdots(l_np-\lambda_n).
\label{equiv_euler}
\end{align}
By replacing $p$, $E_d$, and $c_1(\mathcal{L}_1)$ in  (\ref{J_function2}) with their equivariant partners,
we find the $J$-function for the equivariant Gromov-Witten theory.

By means of the localization of $S_d(\hbar)$ to the fixed point set of the torus action on the moduli space $\overline{\mathcal{M}}_{0,2}(\mathbb{C}\textbf{P}^{N-1},d\beta)$, the $J$-functions of $\mathbb{C}\textbf{P}^{N-1}$ and the smooth complete intersection of hypersurfaces in $\mathbb{C}\textbf{P}^{N-1}$ are computed manifestly:
\begin{prop}[\cite{Givental_eqv}]\label{prop:J_function}
As the equivariant cohomology valued function with $p\in H^*_T(\mathbb{C}\textbf{P}^{N-1})$ (i.e. $\prod_{i=0}^{N-1}(p-w_i)= 0$), the $J$-function $J_{\mathbb{C}\textbf{P}_{\bm{w}}^{N-1}}(x)$ for the equivariant Gromov-Witten theory of $X=\mathbb{C}\textbf{P}^{N-1}_{\bm{w}}$ is given by
\begin{align}
J_{\mathbb{C}\textbf{P}_{\bm{w}}^{N-1}}(x)=\mathrm{e}^{tp/\hbar}\sum_{d=0}^{\infty}
\frac{x^d}{\prod_{m=1}^d(p-w_0+m\hbar)\cdots\prod_{m=1}^d(p-w_{N-1}+m\hbar)}.
\label{J_CPN}
\end{align}
And the $J$-function $J_{X_{\bm{l};\bm{w},\bm{\lambda}}}(x)$ for the equivariant Gromov-Witten theory of the smooth Fano complete intersection $X_{\bm{l};\bm{w},\bm{\lambda}}$ defined by $n$ equations of degrees $\bm{l}=(l_1, \ldots, l_n)$ with $l_1+\cdots +l_n<N$ in $\mathbb{C}\textbf{P}^{N-1}$ is given by
\begin{align}
J_{X_{\bm{l};\bm{w},\bm{\lambda}}}(x)=\mathrm{e}^{tp/\hbar}\sum_{d=0}^{\infty}x^d
\frac{\prod_{m=0}^{dl_1}(l_1p-\lambda_1+m\hbar)\cdots\prod_{m=0}^{dl_n}(l_np-\lambda_n+m\hbar)}{\prod_{m=1}^d(p-w_0+m\hbar)\cdots\prod_{m=1}^d(p-w_{N-1}+m\hbar)}.
\label{J_comp}
\end{align}
\end{prop}

In the context of quantum $D$-module, 
the action of the Heisenberg algebra 
$D:=\mathbb{C}[[\hbar]][[\widehat{x}]][\widehat{y}]
\ ([\widehat{y},\widehat{x}]=\hbar \widehat{x})$ 
on the $J$-function and its annihilator $I$ are studied, 
here operators $\widehat{x}$ and $\widehat{y}$ act on $J_X(x)$ as (\ref{eq:190525_1}). 

\begin{prop}\label{prop:190519_quantum_differential_operator}
Let $X$ be a Fano manifold. Assume that 
the second cohomology group of $X$ has rank $1$ and 
generates the total cohomology ring. 
Then there exists uniquely the differential operator 
$\widehat{A}_X(\widehat{x},\widehat{y}) \in I$ 
such that 
(i) $\widehat{A}_X(\widehat{x},\widehat{y})$ generates $I$, 
(ii)  $\widehat{A}_X(\widehat{x},\widehat{y})$ 
is a monic polynomial with respect to $\widehat{y}$. 
We also have $\widehat{A}_X(\widehat{x},\widehat{y}) 
\in D^{alg}:=\mathbb{C}[\hbar][\widehat{x}][\widehat{y}]$. 
\end{prop}
\begin{proof}
By the assumptions on the Fano manifold $X$, 
we could take a generator $p_1 \in H^2(X,\mathbb{C})$ 
of the algebra $H^*(X,\mathbb{C})$ 
such that the algebra homomorphism 
$\mathbb{C}[y] \to H^*(X,\mathbb{C})$, $
P(y) \mapsto P(p_1)$ is surjective. 
The kernel $I_0$ of this homomorphism is a principal ideal 
generated by the monic polynomial $f$ of degree $m+1$, 
where $\mathrm{dim}H^*(X,\mathbb{C})=m+1$. 
By division with remainder, we have the decomposition 
$$
\mathbb{C}[y]=I_0 \bigoplus(\bigoplus_{j=0}^{m}\mathbb{C}y^j). 
$$

By this decomposition and using the discussion of the proof of Theorem 2.4 of \cite{Iritani_0410487}, we obtain the decomposition of the Heisenberg algebra $D$:
$$
D
=I \oplus R,\quad 
(R:=\bigoplus_{j=0}^{m}\mathbb{C}[[\hbar]][[\widehat{x}]]\widehat{y}^j). 
$$
By this decomposition, we have 
$
\widehat{y}^{m+1}=f_1+f_2\ 
(f_1 \in I,\ f_2 \in R)
$ 
and put $\widehat{A}_X(\widehat{x},\widehat{y}):=f_1$. 
Then $\widehat{A}_X(\widehat{x},\widehat{y})$ satisfies (ii). 
The left ideal $(\widehat{A}_X(\widehat{x},\widehat{y}))$ satisfies 
$(\widehat{A}_X(\widehat{x},\widehat{y})) \subset I$ and 
$D=(\widehat{A}_X(\widehat{x},\widehat{y}))\oplus R$. Then we have 
$(\widehat{A}_X(\widehat{x},\widehat{y}))=I$. 

By the procedure of the proof of Theorem 2.4 of \cite{Iritani_0410487} 
and the fact that $X$ is a Fano manifold, 
we have the decomposition $\widehat{y}^{m+1}=f_1+f_2$ 
with $\widehat{A}_X(\widehat{x},\widehat{y})=f_1 \in D^{alg}$ and $f_2 \in D^{alg}$. 
\end{proof}
For the equivariant cases, we also define $\widehat{A}_X(\widehat{x},\widehat{y})$ in the same manner. 

We call this differential operator 
$\widehat{A}_X(\widehat{x},\widehat{y})$ the \textit{quantum differential operator} 
and we call the differential equation 
$\widehat{A}_X(\widehat{x},\widehat{y})J_{X}(x)=0$ 
\textit{Gel'fand-Kapranov-Zelevinsky (GKZ) equation}.

Givental also gave the explicit description of the GKZ equation for the $J$-functions in Proposition \ref{prop:J_function} (see Theorem 9.1, Corollary 9.2 and Theorem 9.5 of \cite{Givental_eqv}). 
\begin{prop}[\cite{Givental_eqv}] \label{prop:GKZ-for-J-functions}
The GKZ equation 
\begin{align}
\widehat{A}_{\mathbb{C}\textbf{P}^{N-1}_{\bm{w}}}(\widehat{x},\widehat{y})
J_{\mathbb{C}\textbf{P}^{N-1}_{\bm{w}}}(x)
=0,
\label{GKZ_CPN}
\end{align}
for the $J$-function 
$J_{\mathbb{C}\mathbf{P}^{N-1}_{\bm{w}}}(x)$ 
for the equivariant Gromov-Witten theory of 
$\mathbb{C}\textbf{P}^{N-1}_{\bm{w}}$ is explicitly given by 
\begin{equation} \label{eq:QDE-CPN}
\widehat{A}_{\mathbb{C}\textbf{P}^{N-1}_{\bm{w}}}(\widehat{x},\widehat{y})
=\prod_{i=0}^{N-1}(\widehat{y}-w_i)-\widehat{x}.
\end{equation}
And the GKZ equation 
\begin{align}
\widehat{A}_{X_{\bm{l};\bm{w},\bm{\lambda}}}(\widehat{x},\widehat{y})
J_{X_{\bm{l};\bm{w},\bm{\lambda}}}(x)=0,
\label{GKZ_comp}
\end{align}
for the $J$-function $J_{X_{\bm{l};\bm{w},\bm{\lambda}}}(x)$ 
for the equivariant Gromov-Witten theory of the smooth 
Fano complete intersection $X_{\bm{l};\bm{w},\bm{\lambda}}$ 
is explicitly given by 
\begin{equation}
\widehat{A}_{X_{\bm{l};\bm{w},\bm{\lambda}}}(\widehat{x},\widehat{y})
=\prod_{i=0}^{N-1}(\widehat{y}-w_i)
-\widehat{x}\prod_{m=1}^{l_1}(l_1\widehat{y}-\lambda_1+m\hbar)
\cdots
\prod_{m=1}^{l_n}(l_n \widehat{y}-\lambda_n+m\hbar).
\label{eq:QDE-COMP}
\end{equation}
\end{prop}


\subsection{Oscillatory integral and GKZ equation}
\label{subsection:oscillatory-integral-GKZ}

In the mirror theorem (see  pedagogical expositions in \cite{Hori_book}), a correspondence between the Gromov-Witten theory of smooth Fano manifold $X$ with $\mathrm{dim}H^2(X,\mathbb{C})=1$ and the Landau-Ginzburg model with potential function
\begin{align}
W(\cdot,x): (\widetilde{{\mathbb C}^{\ast}})^k\longrightarrow \mathbb{C}, 
\quad (u_1,\ldots,u_k;x) \mapsto W(u_1,\ldots,u_k;x), 
\label{LG_mirror}
\end{align}
is considered. 
Here $\widetilde{{\mathbb C}^\ast}$ is the universal covering of ${\mathbb C}^\ast$, 
and $x$ is a deformation parameter of the potential function. 
For the equivariant Gromov-Witten theory of $X=\mathbb{C}\textbf{P}^{N-1}$ ($X=\mathbb{C}\textbf{P}^{N-1}_{\bm{w}}$ in short) and the smooth Fano complete intersection of degrees $l_a$ ($a=1,\ldots,n$) hypersurfaces in $\mathbb{C}\textbf{P}^{N-1}$ with $l_1+\cdots +l_n<N$ ($X=X_{\bm{l};\bm{w},\bm{\lambda}}$ in short), the mirror Landau-Ginzburg potential $W_X$ is given as follows.

\begin{defi}[\cite{Coates_Givental,Givental_eqv}]
For $X=\mathbb{C}\textbf{P}^{N-1}_{\bm{w}}$, the 
mirror Landau-Ginzburg potential $W_{\mathbb{C}\textbf{P}^{N-1}_{\bm{w}}}(\cdot;x):(\widetilde{{\mathbb C}^\ast})^{N-1}\to\mathbb{C}$ is defined by 
\begin{align}
W_{\mathbb{C}\textbf{P}^{N-1}_{\bm{w}}}(u_1,\ldots,u_{N-1};x)=
\sum_{i=1}^{N-1} (u_{i}+w_i\log u_i)  + \frac{x}{u_1 \cdots u_{N-1}}
+w_0\log\left(\frac{x}{u_1\cdots u_{N-1}}\right),
\label{LG_CPN}
\end{align}
where  $x=\mathrm{e}^{t}\in\mathbb{C}^*$.
For $X=X_{\bm{l};\bm{w},\bm{\lambda}}$, the mirror Landau-Ginzburg potential $W_{X_{\bm{l};\bm{w},\bm{\lambda}}}(\cdot;x):(\widetilde{{\mathbb C}^\ast})^{N-1}\times (\widetilde{{\mathbb C}^\ast})^n\to\mathbb{C}$ is defined by 
\begin{align}
&W_{X_{\bm{l};\bm{w},\bm{\lambda}}}(u_1,\ldots,u_{N-1},v_1,\ldots,v_n;x)
\nonumber \\
&=\sum_{i=1}^{N-1}(u_{i}+w_i\log u_i)-\sum_{a=1}^n(v_a+\lambda_a\log v_a)+\frac{v_1^{l_1}\cdots v_n^{l_n}}{u_1\cdots u_{N-1}}x
+w_{0}\log\left(\frac{v_1^{l_1}\cdots v_n^{l_n}}{u_1\cdots u_{N-1}}x\right).
\label{LG_comp}
\end{align}
\end{defi}

Here we introduce the notion of critical set.

\begin{defi}
Let ${\rm Crit} \subset (\widetilde{{\mathbb C}^\ast})^k \times \mathbb{C}^*$ 
be the {\em critical set} of $W$ defined by \\
$$
{\rm Crit}=\left\{(u_1^{({\rm c})},\dots,u_k^{({\rm c})};x) \in 
(\widetilde{{\mathbb C}^\ast})^k \times \mathbb{C}^*\;\left|\;
\frac{\partial}{\partial u_i}W(u_1^{\rm (c)},\dots,u_k^{({\rm c})};x)=0 \ 
(i=1,\cdots,k)\right.\right\}. 
$$
For any fixed $x$, the image ${\bm u}^{\rm (c)} = (u_1^{({\rm c})},\dots,u_k^{({\rm c})}) \in (\widetilde{{\mathbb C}^\ast})^k$ of a point $({\bm u}^{\rm (c)};x) \in \mathrm{Crit} \cap ((\widetilde{{\mathbb C}^\ast})^k \times\{x\})$ by the isomorphism
$(\widetilde{{\mathbb C}^\ast})^k \times \{x\} \to (\widetilde{{\mathbb C}^\ast})^k$
is called a critical point of $W(\cdot;x)$.
The value $W({\bm u}^{\rm (c)};x)$ at a critical point ${\bm u}^{\rm (c)}$ is called a {\em critical value}.
\end{defi}
 
In what follows, for any fixed $x$, we assume that 
all critical points of $W$ are non-degenerate.
Then, for each $\hbar$, we associate the critical point ${\bm u}^{(c)}$ with the {\em Lefschetz thimble} $\Gamma$: It is a relative $k$-cycle in $(\widetilde{{\mathbb C}^\ast})^k$ which is 
the real 1-parameter family of $(k-1)$-cycles (called {\em vanishing cycles}) in the Milnor fiber of $W^{-1}(w)$, where $w$ lies on the half-line $\{W({\bm u}^{(c)};x) + r \, {\rm e}^{i(\pi + \arg \hbar)} ~|~ r \ge 0 \}$ emanating from the critical value $W({\bm u}^{(c)};x)$, and the $(k-1)$-cycle tends to a point when $w \to W({\bm u}^{(c)};x)$ along the half line. In this section we assume that $x$ and $\hbar$ are chosen so that the half line $\{W({\bm u}^{(c)};x) + r \, {\rm e}^{i(\pi + \arg \hbar)} ~|~ r \ge 0 \}$ associated with a critical point ${\bm u}^{\rm (c)}$ never hits critical values of $W(\cdot; x)$; then the (equivariant) oscillatory integral defined below has the so-called saddle point expansion of the form \eqref{eq:saddle-pt-expansion-IX} below.

For a Lefschetz thimble $\Gamma$ we consider the (equivariant) 
oscillatory integral $\mathcal{I}(x)$ 
of the type
\begin{align} 
\mathcal{I}(x)=\int_{\Gamma}
 \mathrm{e}^{\frac{1}{\hbar}W(u_1,\ldots,u_k;x)}\,
\zeta(u_1,\ldots,u_k;x), 
\label{eq:equiv-osci-int}
\end{align}
where $\zeta(u_1,\ldots,u_k;x)$ denotes a $k$-form on $({\mathbb C}^\ast)^k$.

\begin{defi}[\cite{Coates_Givental,Givental_eqv}]
The oscillatory integral $\mathcal{I}_{\mathbb{C}\textbf{P}^{N-1}_{\bm{w}}}(x)$ for the projective space $X=\mathbb{C}\textbf{P}^{N-1}_{\bm{w}}$ is defined by 
\begin{align}
\mathcal{I}_{\mathbb{C}\textbf{P}^{N-1}_{\bm{w}}}(x)
&=\int_{\Gamma}\prod_{i=1}^{N-1}\frac{du_i}{u_i}\;
\mathrm{e}^{\frac{1}{\hbar}W_{\mathbb{C}\textbf{P}^{N-1}_{\bm{w}}}(u_1,\ldots,u_{N-1};x)}.
\label{J_CPN_osc}
\end{align}
The oscillatory integral $\mathcal{I}_{X_{\bm{l};\bm{w},\bm{\lambda}}}(x)$ for the Fano complete intersection in the projective space $X=X_{\bm{l};\bm{w},\bm{\lambda}}$ 
is defined by the Laplace transform of the oscillatory integral for $X=\mathbb{C}\textbf{P}^{N-1}_{\bm{w}}$.
\begin{align}
\mathcal{I}_{X_{\bm{l};\bm{w},\bm{\lambda}}}(x)&=
\int_0^{\infty}dv_1\cdots\int_0^{\infty}dv_n\,  \mathrm{e}^{-\frac{\sum_{a=1}^n(v_a+\lambda_a\log v_a)}{\hbar}}\mathcal{I}_{\mathbb{C}\textbf{P}^{N-1}_{\bm{w}}}(v_1^{l_1}\cdots v_n^{l_n}x)
\nonumber \\
&=\int_{\Gamma\times(\mathbb{R}_{\ge 0})^n}\prod_{i=1}^{N-1}\frac{du_i}{u_i}\prod_{a=1}^ndv_a\;
\mathrm{e}^{\frac{1}{\hbar}W_{X_{\bm{l};\bm{w},\bm{\lambda}}}(u_1,\ldots,u_{N-1},v_1,\ldots,v_n;x)}.
\label{J_comp_osc}
\end{align}
\end{defi}

The mirror symmetry between the $J$-function $J_{X_{\bm{l};\bm{w},\bm{\lambda}}}(x)$ for the Gromov-Witten theory and the oscillatory integral $\mathcal{I}_{X_{\bm{l};\bm{w},\bm{\lambda}}}(x)$ for the mirror Landau-Ginzburg model is given by the proposition below:
\begin{prop}\label{prop:IGKZ}
The oscillatory integrals $\mathcal{I}_{X}(x)$ for $X=\mathbb{C}\textbf{P}^{N-1}_{\bm{w}}$ and $X=X_{\bm{l};\bm{w},\bm{\lambda}}$ obey the GKZ equations (\ref{GKZ_CPN}) and (\ref{GKZ_comp}), respectively:
\begin{align} \label{eq:all-GKZ-equation}
\widehat{A}_X(\widehat{x},\widehat{y}) \, \mathcal{I}_{X}(x)=0,
\end{align}
where operators $\widehat{x}$ and $\widehat{y}$ act on $\mathcal{I}_{X}(x)$ as (\ref{eq:190525_1}). 
\end{prop}
The claim of Proposition \ref{prop:IGKZ} is essentially found in \cite{Givental_eqv} for the equivariant $\mathbb{C}\textbf{P}^{N-1}$
and in \cite{Coates_Givental} for non-equivariant models, although the claim for $X=X_{\bm{l};\bm{w},\bm{\lambda}}$ is not mentioned manifestly. 
We will give a proof of this Proposition in Appendix \ref{app:IGKZ}.

\subsection{GKZ curves}
\label{subsection:GKZ-curves-critical-set}

Let $X$ be a Fano manifold with the assumptions 
of Proposition \ref{prop:190519_quantum_differential_operator}.
We first define the classical limit $A_X(x,y)$ 
of the quantum differential operator 
$\widehat{A}_X(\widehat{x},\widehat{y})$ by 
\begin{align}
A_X(x,y):=
\lim_{\hbar \mapsto 0, (\widehat{x},\widehat{y}) \mapsto (x,y)}
\widehat{A}_X(\widehat{x},\widehat{y}) 
\in \mathbb{C}[x,y]. 
\nonumber
\end{align} 
Then we define the {\em GKZ curve} $\Sigma_X$ for $X$ by 
\begin{equation} 
\Sigma_X = \{(x,y) \in {\mathbb C}^2~|~ A_X(x,y) = 0 \}.
\label{eq:GKZ-curve}
\end{equation}

As is mentioned in the introduction, our goal is to show that the GKZ equations (\ref{GKZ_CPN}) for $X=\mathbb{C}\textbf{P}^{N-1}_{\bm{w}}$ and (\ref{GKZ_comp}) for $X=X_{\bm{l};\bm{w},\bm{\lambda}}$ are reconstructed from the GKZ curve \eqref{eq:GKZ-curve} by the topological recursion.

As is proved in Corollary 2.5 of \cite{Iritani_0410487}, the classical limit $A_X(x,y)$ of the quantum differential operator gives the defining ideal of the quantum cohomology ring of $X$. 

A part of Mirror symmetry asserts that the quantum cohomology ring of $X$ 
is isomorphic to the function ring of the critical set ${\rm Crit}$ of the 
corresponding Landau-Ginzburg potential $W$. 
This leads to an alternative description of the GKZ curve:
\begin{prop} \label{def:GKZ_curve}
For the cases of the projective space $X=\mathbb{C}\textbf{P}^{N-1}_{\bm{w}}$ and Fano complete intersection 
$X_{\bm{l};\bm{w},\bm{\lambda}}$, the GKZ curve $\Sigma_X$ coincides with the image $\mathrm{Im}\,\iota$ of critical set, where 
$$\iota: {\rm Crit} \to \mathbb{C}^*\times \mathbb{C},\quad
(u_1^{(\mathrm{c})},\dots,u_k^{(\mathrm{c})};x ) \mapsto 
\left(x,x\frac{d}{d x}W_X(u_1^{(\mathrm{c})},\dots,u_k^{(\mathrm{c})};x)
\right).$$ 
\end{prop}

\begin{proof}
We check them directly for 
the cases of the projective space 
$X=\mathbb{C}\textbf{P}^{N-1}_{\bm{w}}$ 
and Fano complete intersection 
$X_{\bm{l};\bm{w},\bm{\lambda}}$. 

\vspace{0.3em}\noindent
{(1)} Projective space $\mathbb{C}\textbf{P}^{N-1}_{\bm{w}}$:\\
The GKZ curve $A_{\mathbb{C}\textbf{P}^{N-1}_{\bm{w}}}(x,y)=0$ is found from the relations
\begin{align}
\begin{split}
&
y(x)=x\frac{\partial W_{\mathbb{C}\textbf{P}^{N-1}_{\bm{w}}}(u_1,\ldots,u_{N-1};x)}{\partial x}=\frac{x}{u_1\cdots u_{N-1}}+w_0,
\\
&
0=\frac{\partial W_{\mathbb{C}\textbf{P}^{N-1}_{\bm{w}}}(u_1,\ldots,u_{N-1};x)}{\partial u_i}=1+\frac{w_i}{u_i}-\frac{1}{u_i}\frac{x}{u_1\cdots u_{N-1}}-\frac{w_0}{u_i}.
\end{split}
\end{align}
By eliminating $u_i$ ($i=1,\ldots,N-1$) from the above relations, we obtain a polynomial constraint equation
\begin{align}
A_{\mathbb{C}\textbf{P}^{N-1}_{\bm{w}}}(x,y)=\prod_{i=0}^{N-1}(y-w_i)-x=0
\label{A_CPN}
\end{align}
which agrees with the classical limit of the differential operator \eqref{eq:QDE-CPN}.

\noindent{(2)} Fano complete intersection $X_{\bm{l};\bm{w},\bm{\lambda}}$:\\
The GKZ curve $A_{X_{\bm{l};\bm{w},\bm{\lambda}}}(x,y)=0$ is found from the relations
\begin{align}
\begin{split}
&
y(x)=x\frac{\partial W_{X_{\bm{l};\bm{w},\bm{\lambda}}}(u_1,\ldots,u_{N-1},v_1,\ldots,v_n;x)}{\partial x}=x\frac{v_1^{l_1}\cdots v_n^{l_n}}{u_1\cdots u_{N-1}}+w_0,
\\
&
0=\frac{\partial W_{X_{\bm{l};\bm{w},\bm{\lambda}}}(u_1,\ldots,u_{N-1},v_1,\ldots,v_n;x)}{\partial u_i}=1+\frac{w_i}{u_i}-\frac{1}{u_i}\frac{v_1^{l_1}\cdots v_n^{l_n}}{u_1\cdots u_{N-1}}x-\frac{w_0}{u_i},
\\
&
0=\frac{\partial W_{X_{\bm{l};\bm{w},\bm{\lambda}}}(u_1,\ldots,u_{N-1},v_1,\ldots,v_n;x)}{\partial v_a}=-1-\frac{\lambda_a}{v_a}+\frac{l_a}{v_a}\frac{v_1^{l_1}\cdots v_n^{l_n}}{u_1\cdots u_{N-1}}x+l_a\frac{w_0}{v_a}.
\end{split}
\end{align}
Eliminating $u_i$ ($i=1,\ldots,N-1$) and $v_a$ ($a=1,\ldots,n$) from the above relations, we obtain a polynomial constraint equation
\begin{align}
A_{X_{\bm{l};\bm{w},\bm{\lambda}}}(x,y)=
\prod_{i=0}^{N-1}(y-w_i)-x\prod_{a=1}^n(l_ay-\lambda_a)^{l_a}=0
\label{A_comp}
\end{align}
which agrees with the classical limit of the differential operator \eqref{eq:QDE-COMP}.
\end{proof}

The embedding $\iota$ could be explained from the following observation. 
From the view point of Landau-Ginzburg models, the GKZ curve can be obtained from the leading behavior of saddle point approximation of the equivariant oscillatory integral when $\hbar\to 0$:

\begin{equation}
{\mathcal I}_X(x) \sim \exp\left(\frac{1}{\hbar} W_{X}({\bm u}^{\rm (c)};x) \right) \frac{(-2\pi \hbar)^{k/2} \, g({\bm u}^{\rm (c)})}{\sqrt{{\rm Hess}({\bm u}^{\rm (c)})}} \left( 1 + \sum_{m=1}^{\infty} \hbar^m {\mathcal I}_m(x) \right),
\label{eq:saddle-pt-expansion-IX}
\end{equation}
where ${\bm u}^{\rm (c)} = (u_1^{(\mathrm{c})},\ldots,u_k^{(\mathrm{c})})$ is the critical point of $W_X(\cdot; x)$ for the Lefschetz thimble $\Gamma$, ${\rm Hess}({\bm u}^{\rm (c)}) = \det (\partial_{u_i} \partial_{u_j} W_X({\bm u};x) )|_{{\bm u} = {\bm u}^{\rm (c)}}$ is the Hessian of $W_X(\cdot;x)$ at ${\bm u}^{\rm (c)}$, and we wrote 
$\zeta = g(u_1,\dots,u_k) du_1 \cdots du_k$. Note that the right hand side of \eqref{eq:saddle-pt-expansion-IX} is usually divergent, so this is understood as an asymptotic expansion.  We can also arrange the right hand side of \eqref{eq:saddle-pt-expansion-IX} into a WKB form as
\begin{equation} \label{eq:saddle-pt-expansion-IX-exp}
{\mathcal I}_X(x) \sim \exp \left( \sum_{m=0}^{\infty} \hbar^{m-1} S_m(x) \right)
\end{equation} (where overall factor $(-2\pi \hbar)^{k/2}$ is omitted).
The leading term 
\begin{align}
S_0(x)=W_X({\bm u}^{\rm (c)};x).
\label{IX_c}
\end{align}
is the critical value of $W_X$.
Adopting (\ref{eq:saddle-pt-expansion-IX-exp}) to $\widehat{A}_X(x,\hbar x\partial_x)\mathcal{I}_{X}(x)=0$, then we find the semi-classical limit of the differential operator $\widehat{A}_X(\widehat{x},\widehat{y})$  \cite{Dunin-Barkowski:2013wca}:
\begin{align}
0 = \lim_{\hbar\to 0}
\left(
\mathrm{e}^{-\frac{1}{\hbar}S_0(x)}
\widehat{A}_X(\widehat{x},\widehat{y})
\mathrm{e}^{\frac{1}{\hbar}S_0(x)}
\mathrm{e}^{\sum_{m\ge 1}\hbar^{m-1}S_m(x)}
\right)
=A_X\left(x,x\frac{d}{dx}S_0(x)\right)\mathrm{e}^{S_1(x)}.
\label{classical_A}
\end{align}
For the critical value $S_0(x)=W_X({\bm u}^{\rm (c)};x)$,  we obtain an equation:
\begin{align}
A_X\left(x,y(x)\right)=0, \qquad
y(x)=x\frac{d}{dx} W_X({\bm u}^{\rm (c)};x).
\end{align} 

\subsection{Asymptotics of the coefficients}
In Section \ref{sec-gkz_reconst} we will compare the asymptotic expansion of oscillatory integral to a wave function constructed via topological recursion applied to the GKZ curve $A_X(x,y) = 0$, for $X = \mathbb{C}\textbf{P}^{N-1}_{\bm{w}}$ and $X_{\bm{l};\bm{w},\bm{\lambda}}$ with $l_1 = \cdots = l_n = 1$ ($X_{\bm{w},\bm{\lambda}}$ for short). For the purpose, we investigate the asymptotic behavior of the coefficients ${\mathcal I}_m(x)$ in the expansion \eqref{eq:saddle-pt-expansion-IX} when $x$ tends to $\infty$.

In this subsection we consider the case $X = X_{\bm{w},\bm{\lambda}}$:
\[
W_X = \sum_{i=1}^{N-1} (u_i + w_i \log u_i) - 
\sum_{a=1}^{n} (v_a + \lambda_a \log v_a) + 
\frac{v_1 \cdots v_n}{u_1 \cdots u_{N-1}} x 
+ w_0 \log \left( \frac{v_1 \cdots v_n}{u_1 \cdots u_{N-1}} x \right).
\]
($X = \mathbb{C}\textbf{P}^{N-1}_{\bm{w}}$ is included as the case of $n=0$.)
It is easy to check that, at a critical point $({\bm u}^{\rm (c)}, {\bm v}^{\rm (c)}) = (u^{\rm (c)}_1, \cdots, u^{\rm (c)}_{N-1}, v^{\rm (c)}_1, \dots, v^{\rm (c)}_n)$,
\begin{equation} \label{eq:critical-condition-1}
u^{\rm (c)}_i + w_i - w_0 = v^{\rm (c)}_a + \lambda_a - w_0 = 
\frac{v^{\rm (c)}_1 \cdots v^{\rm (c)}_n}
{u^{\rm (c)}_1 \cdots u^{\rm (c)}_{N-1}} x  
\end{equation}
holds for any $i$ and $a$. Since the right hand side of \eqref{eq:critical-condition-1} is independent of $i$ and $a$, we can write all $u^{\rm (c)}_i$ and $v^{\rm (c)}_a$ in terms of $u_1^{\rm (c)}$:
\[
u^{\rm (c)}_i = u^{\rm (c)}_1 + w_1 - w_i, \quad 
v^{\rm (c)}_a = u^{\rm (c)}_1 + w_1 - \lambda_a \quad
(i=1,\dots,{N-1},~ a = 1,\dots, n). 
\]
Therefore, $u^{\rm (c)}_1$ must be a solution of the algebraic equation 
\begin{equation}
\prod_{i=0}^{N-1}(u^{\rm (c)}_1 + w_1 - w_i) - x 
\prod_{a=1}^{n} (u^{\rm (c)}_1 + w_1 - \lambda_a) = 0.
\label{eq:critical-point-equation-u1}
\end{equation}
Hence, for generic $w_i$ and $\lambda_a$, there are exactly $N$ critical points of $W_X$. These critical points define $N$ Lefschetz thimbles, and hence we have $N$ independent solutions of the GKZ equation. 

\begin{lemm} \label{lemm:asymptotic-critical-pt}
The asymptotic behavior of a critical point of $W_X$ for large $x$ is given by one of the following:
\begin{itemize}
\item[(i)] For any fixed $p=1,\dots, N-n$, there exists a critical point $({\bm u}^{\rm (c)}, {\bm v}^{\rm (c)})$ of $W_X$ behaves as
\begin{equation} \label{eq:critical-bahavior-1}
u_i^{\rm (c)} = \zeta^p x^{\frac{1}{N-n}} (1 + O(x^{-\frac{1}{N-n}})), \quad  
v_a^{\rm (c)} = \zeta^p x^{\frac{1}{N-n}} (1 + O(x^{-\frac{1}{N-n}}))
\end{equation}
when $x \to \infty$. Here $\zeta = \exp(2\pi i/(N-n))$ is the primitive $(N-n)$-th root of unity.
\item[(ii)] For any fixed $b=1,\dots,n$, there exists a critical point $({\bm u}^{\rm (c)}, {\bm v}^{\rm (c)})$ of $W_X$ behaves as
\begin{equation} \label{eq:critical-bahavior-2}
u^{\rm (c)}_i = (\lambda_b - w_i) + c_b x^{-1} + O(x^{-2}),\quad 
v^{\rm (c)}_a = (\lambda_b - \lambda_a) + c_b x^{-1} + O(x^{-2})
\end{equation}
when $x \to \infty$. Here 
\[
c_b = \frac{\prod_{i=0}^{N-1}(\lambda_b - w_i)}
{\prod_{\substack{a=1 \\ a \ne b}}^n(\lambda_b - \lambda_a)}.
\]
\end{itemize}
\end{lemm}
The claim follows easily from the equations \eqref{eq:critical-condition-1} and \eqref{eq:critical-point-equation-u1} for the critical points. 

Let us look at the asymptotic expansion \eqref{eq:saddle-pt-expansion-IX} of ${\mathcal I}_X$ for our case:
\begin{equation} \label{eq:saddle-expansion-oscillatory-integral}
{\mathcal I}_X(x) \sim \exp\left( \frac{1}{\hbar} W_X({\bm u}^{\rm (c)},{\bm v}^{\rm (c)};x) \right) \, \frac{(-2\pi \hbar)^{\frac{N+n-1}{2}}}{u^{\rm (c)}_1 \cdots u^{\rm (c)}_{N-1} \, \sqrt{{\rm Hess}({\bm u}^{\rm (c)},{\bm v}^{\rm (c)})}} \, \left( 1 + \sum_{m = 1}^{\infty} \hbar^m {\mathcal I}_m(x) \right).
\end{equation}
Using Lemma \ref{lemm:asymptotic-critical-pt}, we can prove  

%
%

\begin{prop} \label{prop:behavior-of-saddle-point-approximation} ~
\begin{itemize}
\item[(i)]
The factor $(u^{\rm (c)}_1 \cdots u^{\rm (c)}_{N-1})^{-1} ({\rm Hess}({\bm u}^{\rm (c)},{\bm v}^{\rm (c)}))^{-1/2}$ in \eqref{eq:saddle-expansion-oscillatory-integral} behaves when $x \to \infty$ as
\begin{eqnarray*}
\frac{1}{u^{\rm (c)}_1 \cdots u^{\rm (c)}_{N-1} \, \sqrt{{\rm Hess}({\bm u}^{\rm (c)},{\bm v}^{\rm (c)})}} = 
\begin{cases} 
O(x^{- \frac{N-n-1}{2(N-n)}}) & \text{in the case of \eqref{eq:critical-bahavior-1}}, \\[+.5em]
O(x^{-1}) & \text{in the case of \eqref{eq:critical-bahavior-2}}.
\end{cases}
\end{eqnarray*}
\item[(ii)] For $m \ge 1$, the coefficient ${\mathcal I}_m(x)$ in the asymptotic expansion \eqref{eq:saddle-expansion-oscillatory-integral} behaves as
\begin{eqnarray}
{\mathcal I}_m(x) = 
\begin{cases}
O(x^{-\frac{1}{2(N-n)}}) &  \text{in the case of \eqref{eq:critical-bahavior-1}} \\
O(1) & \text{in the case of \eqref{eq:critical-bahavior-2}}
\end{cases}
\end{eqnarray}
when $x \to \infty$.
\end{itemize}
\end{prop}
We will give a proof of Proposition \ref{prop:behavior-of-saddle-point-approximation} in Appendix \ref{appendix:asymptotic-behavior-of -coefficients}.

Proposition \ref{prop:behavior-of-saddle-point-approximation} (i) implies that, if we arrange the asymptotic expansion \eqref{eq:saddle-expansion-oscillatory-integral} into WKB form as \eqref{eq:saddle-pt-expansion-IX-exp}, then 
\begin{equation} \label{eq:normalization-condition-for-S}
\lim_{x \to \infty} S_m(x) = 0 \quad \text{for $m \ge 2$}
\end{equation} 
if the corresponding critical point behaves as \eqref{eq:critical-bahavior-1} when $x \to \infty$. This property is used to compare the wave function and oscillatory integral.

\section{Quantum curves and topological recursion}\label{sec-q_curve_rec}

Quantization of a spectral curve $\Sigma$ (or a quantum curve, for short) is formulated as an $\hbar$-deformed differential equation, whose semi-classical limit $\hbar\to 0$ yields $\Sigma$ \cite{Dumitrescu:2013tca,Dunin-Barkowski:2013wca}.
In this section, under a general setting, we will briefly summarize how to construct quantum curves by the topological recursion.

\subsection{Spectral curves and quantum curves}

\begin{defi}[Spectral curve \cite{Eynard:2007kz,Bouchard:2016obz}]
\label{def:spectral-curve-for-TR}
A {\em spectral curve} is a triple $\Sigma = (C,x,y)$, where 
$C$ is a Torelli marked compact Riemann surface and
$x, y : C \to \mathbb{C}\textbf{P}^1$ are meromorphic functions, 
such that the zeroes of $dx$ do not coincide with the zeroes of $dy$.
\end{defi}

The meromorphic functions $x$, $y$ must satisfy an absolutely 
irreducible equation of the form 
$A(x,y) = \sum_{i,j} A_{i,j} x^i y^j =0$. 
We will just simply denote a spectral curve by 
\begin{align}
\Sigma = \left\{\, (x,y) \in {\mathbb C}^2 \; |\; A(x,y)=0\, \right\},
\label{eq:spectral_curve}
\end{align}
after the parametrization $(x,y) = (x(z), y(z))$ of $\Sigma$ by 
$z \in C$ is fixed. We will show the meromorphic functions 
$x(z)$ and $y(z)$ which parametrize the GKZ curves 
in Section \ref{sec-gkz_reconst}.

\begin{remark}
From the viewpoint of the quantum curve or the WKB analysis, it is natural to regard that $(x,y)$ is an affine coordinate of the cotangent bundle of ${\mathbb C}$, where $x$ (resp. $y$) represents the coordinate of the base (resp. fiber) of the cotangent bundle (e.g. \cite{Dumitrescu:2013tca}). 
In particular, the spectral curve $\Sigma$ is equipped with the 1-form 
\begin{equation}
\omega(x) = y(x)\,dx, 
\label{eq:ydx}
\end{equation}
which is the restriction of the canonical 1-form in the cotangent bundle $T^\ast{\mathbb C} \cong {\mathbb C}^2$. 
\end{remark}

\begin{defi}[Quantum curve]\label{def:q_curve}
A quantum curve is a triple 
$(\widehat{A}(\widehat{x},\widehat{y}), A(x,y), \psi(x))$ 
where $\psi(x)$ is a function, 
$\widehat{A}(\widehat{x},\widehat{y})$ 
is a differential operator with 
\begin{align}
\widehat{x}\,\psi(x)=x\, \psi(x),\qquad
\widehat{y}\,\psi(x)=\hbar \frac{d}{dx}\,\psi(x),
\label{hatx_y}
\end{align}
for the operators $\widehat{x}$ and $\widehat{y}$,\footnote{
We remark that the operators $\widehat{x}$ and $\widehat{y}$ obey 
the commutation relation
$
\left[\widehat{y}, \widehat{x}\right]=\hbar. 
$
} and 
$A(x,y)$ is an irreducible polynomial of $x,y$ 
satisfying the following conditions:
\begin{itemize}
\item $\widehat{A}$ is a 
(possibly infinite-order)
differential 
operator such that
\begin{align}
&\widehat{A}(\widehat{x},\widehat{y}) \, \psi(x)=0.
 \label{def2:q_curve} 
\end{align}
\item 
$\psi(x)$ has the following expression (WKB solution): 
\begin{align}
\psi(x)=\exp \left(\frac{1}{\hbar}\int^x \omega(x)+O(\hbar^{0})\right). 
\label{wkb_lead}
\end{align}
\item 
By taking the semi-classical limit $\hbar\to 0$, we have $A(x,y)$:
\begin{align}
\hbar \to 0,\qquad 
(\widehat{x},\widehat{y})\to (x,y),\qquad 
\widehat{A}(\widehat{x},\widehat{y})\to A(x,y)\in \mathbb{C}[[x,y]].
\nonumber
\end{align}
\end{itemize}
We call the function $\psi(x)$ a wave function. 
We also call the equation (\ref{def2:q_curve}) a quantum curve for short. 
\end{defi}

In Section \ref{subsec-rec_q_curve}, following \cite{Bouchard:2016obz} we will 
define the wave function associated with a spectral curve $\Sigma$ (Definition \ref{def:rec_wkb}) defined by the topological recursion in Section \ref{subsec:T-Rec}, and give a construction of a quantum curve which annihilates it.

\begin{remark} \label{rem:def-of-omega-x}
For our purpose, we will regard $x$ as a coordinate of ${\mathbb C}^\ast$: 
\begin{align}
\Sigma = \left\{\, (x,y) \in {\mathbb C}^\ast \times {\mathbb C} ~|~ A(x,y)=0\, \right\},
\label{eq:spectral_curve-C-ast}
\end{align}
since we use $x=e^t \in {\mathbb C}^\ast$ as a coordinate when we regard GKZ curve $\Sigma_X$ as a spectral curve. The spectral curve is a subset of $T^\ast {\mathbb C}^\ast \cong {\mathbb C}^\ast \times {\mathbb C}$.
Then we use 
\begin{equation}
\omega(x) = y(x)\frac{dx}{x}
\label{eq:ydx/x}
\end{equation}
as the counterpart of \eqref{eq:ydx}, and the quantum curve is similarly defined
as Definition \ref{def:q_curve} by\footnote{The commutation relation becomes $\left[\widehat{y}, \widehat{x}\right]=x \hbar$ which is sightly modified from the previous one.}
\[
\widehat{x}\,\psi(x)=x\, \psi(x),\qquad
\widehat{y}\,\psi(x)=\hbar x\frac{d}{dx}\,\psi(x).
\] 
\end{remark}

\begin{remark} \label{rem:def-of-omega-x-2}
On the one hand, in the local mirror symmetry discussed in Section \ref{sec-glsm_J}, spectral curves which are algebraic in exponentiated variables $\mathsf{x}, \mathsf{y} \in {\mathbb C}^\ast$ appear. Then the corresponding quantum curve is defined by\footnote{
Using the operators $\widehat{x}$ and $\widehat{y}$ in (\ref{hatx_y})
one can represent operators $\widehat{\mathsf{x}}$ and $\widehat{\mathsf{y}}$ as
 $\widehat{\mathsf{x}}=\mathrm{e}^{\widehat{x}}$and $\widehat{\mathsf{y}}=\mathrm{e}^{\widehat{y}}$.
}
$$
\widehat{\mathsf{x}}\,\psi(\mathsf{x})= \mathsf{x} \, \psi(\mathsf{x}),\qquad 
\widehat{\mathsf{y}}\,\psi(\mathsf{x})= \psi({\rm e}^{\hbar} \mathsf{x}),
$$  
with the commutation relation
$\widehat{\mathsf{y}} \, \widehat{\mathsf{x}} = \mathrm{e}^{\hbar} \, \widehat{\mathsf{x}} \, \widehat{\mathsf{y}}$. 
Here the counterpart of \eqref{eq:ydx} in this case is given by 
\begin{equation}
\omega(\mathsf{x}) = \log \mathsf{y}(\mathsf{x}) \frac{d\mathsf{x}}{\mathsf{x}}
\label{eq:logydx/x}.
\end{equation}
\end{remark}

\subsection{Topological recursion}
\label{subsec:T-Rec}

For a spectral curve $\Sigma = (C,x,y)$, one can (re)construct the wave function as the WKB expansion via the topological recursion. Before describing the reconstruction we will firstly review the (local) topological recursion defined for $\Sigma$ with only simple ramification points \cite{Eynard:2007kz}, and also review the (global) topological recursion defined for $\Sigma$ with arbitrary ramification points \cite{Bouchard:2012an,Bouchard:2012yg,Bouchard:2016obz}.  In the following we use 
\begin{itemize} 
\item 
$\omega (x(z)) = y(z) dx(z)$ if $(x,y)$ is a coordinate of ${\mathbb C}^2$.
\item 
$\omega (x(z)) = y(z) \frac{dx(z)}{x(z)}$ if $(x,y)$ is  a coordinate of ${\mathbb C}^\ast \times {\mathbb C}$ (see Remark \ref{rem:def-of-omega-x}).
\item 
$\omega (x(z)) = \log y(z) \frac{dx(z)}{x(z)}$ if $(x,y)$ is  a coordinate of $({\mathbb C}^\ast)^2$ (see Remark \ref{rem:def-of-omega-x-2}).
\end{itemize}

\subsubsection{Local topological recursion for simple ramified spectral curves}

Let $\Sigma$ be a spectral curve whose all branch points (zeros of $dx=0$) on the $x$-plane are simple. Let $R$ be the set of all ramification points in $C$.
Then near each ramification point $q\in R\subset C$ one can take a local coordinate $z \in C$ and find a unique conjugate point $\overline{z}=\sigma_q(z)\ne z$, where $\sigma_q$ is the local Galois conjugation of $\Sigma$ near $q$.

\begin{defi}[\cite{Eynard:2007kz}]\label{def:top_rec1}
For a simple ramified spectral curve $\Sigma$, the symmetric meromorphic differentials $\omega_n^{(g)}$ ($g\in\mathbb{Z}_{\ge 0}$) on $C^n$ for $(g,n)\neq(0,1)$, $(0,2)$ are recursively defined by the local topological recursion
\begin{align}
\begin{split}
&
\omega^{(g)}_{n+1}(z,\bm{z}_N)=
\sum_{q \in R}\mathop{\mathrm{Res}}_{w=q}\;
\frac{\frac{1}{2}\int^w_{\overline{w}}B(\cdot,z)}{\omega(x(w))-\omega(x(\overline{w}))}\bigg(\omega^{(g-1)}_{n+2}(w,\overline{w},\bm{z}_N)
\\
&
\hspace{12em}
+\sum_{\ell=0}^{g}\sum_{\emptyset=J\subseteq N}
\omega^{(g-\ell)}_{|J|+1}(w,\bm{z}_J)\omega^{(\ell)}_{|N|-|J|+1}(\overline{w},\bm{z}_{N \backslash J})\bigg),
\label{top_recursion}
\end{split}
\end{align}
with initial inputs
$$
\omega^{(0)}_{1}(z)=0,\qquad \omega^{(0)}_{2}(z_1,z_2)=B(z_1,z_2),
$$
in addition to the 1-form $\omega(x(z))$ on $C$, where $\bm{z}_N=\{z_1,\ldots,z_n\}$, $N=\{1,\ldots,n\}\supset J=\{i_1,\ldots,i_j\}$, $N\backslash J=\{i_{j+1},\ldots,i_n\}$. Here $B(z_1,z_2)$ is the Bergman kernel on $C^2$, which is a bi-differential and holomorphic except $z_1=z_2$, defined uniquely by
$$
\bullet\ \ B(z_1,z_2) = B(z_2, z_1), \qquad
\bullet\ \ B(z_1,z_2)\mathop{\sim}_{z_1 \to z_2} \frac{dz_1dz_2}{(z_1-z_2)^2}+\textrm{reg.}\qquad 
\bullet\ \ \oint_{A_i}B(z_1,z_2)=0,
$$
where $A_i$ ($i=1,\ldots, \textrm{genus of}\ C$) are canonical $A$-cycles (recall that $C$ is Torelli marked).
\end{defi}

\begin{exam}
For the case $C = \mathbb{C}\textbf{P}^1$, the Bergman kernel is given by
\begin{align}
B(z_1,z_2)=\frac{dz_1dz_2}{(z_1-z_2)^2}.
\end{align}
\end{exam}

\subsubsection{Global topological recursion for arbitrary ramified spectral curves}

The local topological recursion in Definition \ref{def:top_rec1} is applicable only for simple ramified spectral curves. In \cite{Bouchard:2012an,Bouchard:2012yg,Bouchard:2016obz} it was proposed the global topological recursion which is also applicable for arbitrary ramified spectral curves. Consider a spectral curve with degree $r$ of $x$ defined by
\begin{align}
A(x,y)=\sum_{k=0}^r a_{r-k}(x)y^k =0,\qquad
x,y\in {\IC}\textrm{ or }{\IC}^*,
\label{sp_curve_poly}
\end{align}
where $a_k(x)$ are polynomials of $x$. Let $R$ be the set of all ramification points on the $x$-plane, and $\sigma_q$ be the local Galois conjugation of $\Sigma$ near $q\in R\subset C$. For a local coordinate $z \in C$, one finds a set $\sigma_q(z)$ of $r-1$ points near each ramification point $q$.

\begin{defi}[\cite{Bouchard:2012an,Bouchard:2012yg,Bouchard:2016obz}]\label{def:top_rec2}
For a multi-ramified spectral curve $\Sigma$ defined by (\ref{sp_curve_poly}), the symmetric meromorphic differentials $\omega_n^{(g)}$ ($g\in\mathbb{Z}_{\ge 0}$) on $C^n$ for $(g,n)\neq(0,1)$, $(0,2)$ are recursively defined by the global topological recursion
\begin{align}
\begin{split}
\omega^{(g)}_{n+1}(z,\bm{z}_N)&=\sum_{q \in R}\mathop{\mathrm{Res}}_{w=q}\Bigg(
\sum_{k=1}^{r-1}\sum_{\beta(w) \subseteq_{k} \sigma_q(w)}
\frac{(-1)^{k+1}\int_{w_*}^w B(\cdot,z)}{\prod_{b_w \in \beta(w)} \left(\omega(x(w))-\omega(x(b_w))\right)}
\\
&\hspace{13.5em}\times
\mathcal{R}^{(k+1)} \omega^{(g)}_{n+1}(w,\beta(w); \bm{z}_N)\Bigg),
\label{g_top_recursion}
\end{split}
\end{align}
with initial inputs
$$
\omega^{(0)}_{1}(z)=0,\qquad \omega^{(0)}_{2}(z_1,z_2)=B(z_1,z_2),
$$
in addition to the 1-form $\omega(x(z))$ on $C$, where
$$
\mathcal{R}^{(k)}\omega^{(g)}_{n+1}(\bm{w}_K; \bm{z}_N) = 
\sum_{{\mu} \in \mathcal{P}(K)} \sum_{ \uplus_{i=1}^{\ell(\mu)} J_i =  N} \sum_{\sum_{i=1}^{\ell(\mu)}g_i = g + \ell(\mu) - k } \left(\prod_{i=1}^{\ell(\mu)} \omega_{g_i, |\mu_i| + |J_i|} (\bm{w}_{\mu_i}, \bm{z}_{J_i}) \right).
$$
Here $\bm{z}_N=\{z_1,\ldots,z_n\}$ (resp. $\bm{w}_K=\{w_1,\ldots,w_k\}$) for $N=\{1,\ldots,n\}$ (resp. $K=\{1,\ldots,k\}$), $\beta(w) \subseteq_{k} \sigma_q(w)$ means $\beta(w) \subseteq \sigma_q(w)$ with $|\beta(w)|=k$,
$\mathcal{P}(K)$ is the set of partitions of $K$, $\ell(\mu)$ is the number of subsets in the set partition $\mu$, and the symbol $\uplus$ means the pairwise disjoint union. $w_*$ is a reference point on $\Sigma$ and we see that $\omega_n^{(g)}$'s do not depend on it.
\end{defi}

\begin{remark}
For simple ramified spectral curves, the global topological recursion in Definition \ref{def:top_rec2} is equivalent to the local topological recursion in Definition \ref{def:top_rec1}.
\end{remark}

\subsection{Reconstruction of quantum curves by topological recursion}\label{subsec-rec_q_curve}

By the symmetric meromorphic differentials $\omega_n^{(g)}$ on $C^n$ defined by the topological recursion (\ref{top_recursion}) or (\ref{g_top_recursion}), one can (re)construct the wave function as the WKB expansion. For the reconstruction we define the divisor $D$ for the integration contour of $\omega_{n}^{(g)}$'s as follows.
\begin{defi}
Let $D$ denote the degree $0$ divisor on $C$ with $q_k\in C$ such that
$$
D=\sum_{k}d_k[q_k],\quad \mathrm{deg}D=\sum_{k}d_k=0.
$$
For the degree $0$ divisor $D$, an integration of a meromorphic $1$-form $\alpha$ on $C$ is defined by
$$
\int_D\alpha=\sum_{k}d_k\int_b^{q_k}\alpha,
$$
where $b\in C$ is an arbitrary reference point, and the integration contours are assumed to not intersect with the homology 1-cycles of $C$. In fact, each integral $\int_D\alpha$ does not depend on the choice of the reference point $b$, because the integration divisor $D$ obeys $\mathrm{deg} D=0$.
\end{defi}

On the basis of this notation the wave function is reconstructed as:
\begin{defi}[Reconstructing WKB]\label{def:rec_wkb}
The wave function $\psi(D)$ associated with a spectral curve $\Sigma$ and degree $0$ divisor $D$ is defined by
\begin{align}
\begin{split}
\psi(D)&=\exp\bigg[
\frac{1}{\hbar}\int_{D}\widehat{\omega}_1^{(0)}(z_1)
+\frac{1}{2}\int_{D}\int_{D}\widehat{\omega}_2^{(0)}(z_1,z_2)
\\
&\hspace{3.2em}
+\sum_{(g,n)\ne (0,1), (0,2)}\frac{1}{n!}\hbar^{2g-2+n}
\int_{D}\cdots \int_{D}\omega_n^{(g)}(z_1,\ldots,z_n)\bigg],
\label{wave_function}
\end{split}
\end{align}
where $z \in C$ is away from ramification points.
Here we have defined
\begin{align}
\widehat{\omega}_1^{(0)}(z_1)=\omega(x(z_1)),\qquad
\widehat{\omega}_2^{(0)}(z_1,z_2)=B(z_1,z_2)-\frac{dx(z_1)dx(z_2)}{(x(z_1)-x(z_2))^2}.
\nonumber
\end{align}
In this reconstruction the leading and subleading integrals $\int_{D}\widehat{\omega}_1^{(0)}(z_1)$ and $\int_{D}\int_{D}\widehat{\omega}_2^{(0)}(z_1,z_2)$ should be regularized so as to remove divergence by an overall normalization factor for $\psi(D)$. 
\end{defi}

In \cite{Bouchard:2016obz}, by the global topological recursion (\ref{g_top_recursion}) the quantum curve which annihilates the wave function in (\ref{wave_function}) was reconstructed systematically for a special class of the spectral curve. Here we will review their elegant results of the WKB reconstruction of quantum curves. 
In the following, we consider the spectral curve in ${\IC}^2$ as
\begin{align}
\Sigma=\Big\{\,
(x,y)\in\mathbb{C}^2\;\Big|\; P(x,y)=\sum_{k=0}^{r}p_{k}(x)y^{r-k}=0\,
\Big\},
\label{sp_curve_be}
\end{align}
where $p_{k}(x)$ are polynomials of $x$. The special class of the spectral curve $\Sigma$ discussed in \cite{Bouchard:2016obz} is referred to as \textit{admissible}.
The admissibility condition is defined by the Newton polygon $\Delta$ for the defining polynomial $P(x,y)$ of the spectral curve $\Sigma$:
$$
P(x,y)=\sum_{k=0}^{r}p_{k}(x)y^{r-k}=\sum_{(k,i)\in  \mathbb{Z}^2}p_{k,i}x^ky^i.
$$
A Newton polygon $\Delta$ for $P(x,y)$ is then a convex hull of the set $S_P$ such that
$$
S_P=\big\{\,(k,i)\in\mathbb{Z}^2\;\big|\;p_{k,i}\ne 0\,\big\}.
$$

\begin{figure}[t]
\centering%
\includegraphics[width=7cm,height=7cm,keepaspectratio]{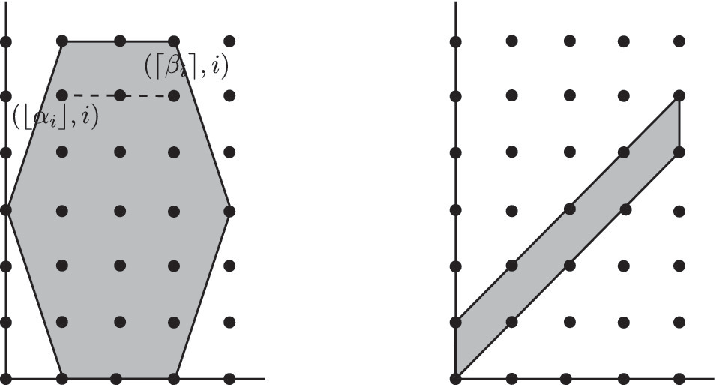}%
\caption{\label{fig:Newton1}The Newton polygon for 2-parameter polynomial $P(x,y)$. The right figure satisfies the admissibility condition.}
\end{figure}

\noindent For each level set labelled by $m\in\mathbb{Z}$ in a Newton polygon $\Delta$, we define
$$
\alpha_i=\mathrm{inf}\bigl\{k\,|\,(k,i)\in\Delta\bigr\},\qquad \beta_i=\mathrm{sup}\bigl\{k\,|\,(k,i)\in\Delta\bigr\}.
$$
The number  of interior integral points of a Newton polygon $\Delta$ is given by $\sum_{i\in\mathbb{Z}}(\lceil\beta_i\rceil-\lfloor\alpha_i\rfloor-1)$, and
the Newton polygon $\Delta$ has no interior point, if $S_P$ satisfies
$$
\lceil\beta_i\rceil-\lfloor\alpha_i\rfloor=1,\ \
\textrm{for all $i$.}
$$
On the basis of the above notions of Newton polygon, the admissibility condition of the spectral curve $\Sigma$ is given as follows. 
\begin{defi}[\cite{Bouchard:2016obz}]\label{def:admissibility}
The spectral curve $\Sigma=\{(x,y)\in\mathbb{C}^2|P(x,y)=0\}$ is admissible if the following two conditions are satisfied:
\begin{enumerate}
\item The Newton polygon $\Delta$ associated with $\Sigma$ has no interior point.
\item If $\Sigma$ contains the origin $(x,y)=(0,0)\in \mathbb{C}^2$, then the curve is smooth at this point.
\end{enumerate}
\end{defi}

The quantum curve which annihilates the wave function $\psi(D)$ with $D=[z]-[z_*]$ is reconstructed manifestly for the admissible spectral curve.
\begin{prop}[Lemma 5.14 in \cite{Bouchard:2016obz}]\label{prop:BE_reconstruction}
Let $D_k$ be the differential operators 
\begin{align}
D_k=\hbar\frac{x^{\lfloor\alpha_k\rfloor}}{x^{\lfloor\alpha_{k-1}\rfloor}}\frac{d}{dx},\quad k=1,\ldots,r.
\label{differential_op}
\end{align}
For the  degree $0$ divisor $D=[z]-[z_*]$ with a simple pole $z_*$ of $x$ as reference point, the wave function $\psi(D)$ satisfies the order $r$ ordinary differential equation
\begin{align}
&
\Biggl[
D_1D_2\cdots D_{r-1}\frac{p_0(x)}{x^{\lfloor\alpha_r\rfloor}}D_r+D_1D_2\cdots D_{r-2}\frac{p_1(x)}{x^{\lfloor\alpha_{r-1}\rfloor}}D_{r-1}+
\cdots +\frac{p_{r-1}(x)}{x^{\lfloor\alpha_{1}\rfloor}}D_{1}+\frac{p_{r}(x)}{x^{\lfloor\alpha_{0}\rfloor}}
\nonumber \\
&
-\hbar\, C_1D_1D_2\cdots D_{r-2}\frac{x^{\lfloor\alpha_{r-1}\rfloor}}{x^{\lfloor\alpha_{r-2}\rfloor}}
-\hbar\, C_2D_1D_2\cdots D_{r-3}\frac{x^{\lfloor\alpha_{r-2}\rfloor}}{x^{\lfloor\alpha_{r-3}\rfloor}}
-\cdots
-\hbar\, C_{r-1}\frac{x^{\lfloor\alpha_{1}\rfloor}}{x^{\lfloor\alpha_{0}\rfloor}}
\Biggr]\psi(D)=0.
\label{WKB_quantum_curve}
\end{align}
Here the coefficients $C_k$'s ($k=1,\ldots r-1$) are given by
\begin{align}
C_k=\lim_{z\to z_*}\frac{P_{k+1}(x,y(x))}{x^{\lfloor\alpha_{r-k}\rfloor+1}},
\qquad
P_{k+1}(x,y)=\sum_{i=1}^{k}p_{k-i}(x)y^i,
\label{Ck}
\end{align}
where $y(x)$ obeys $P(x,y(x))=0$.
\end{prop}

\begin{remark}[Remark 5.12 in \cite{Bouchard:2016obz}]\label{rem:wkb_reconst}
Even when the reference point $z_*$ is a higher order pole of $x$, if the integrals in the WKB reconstruction (\ref{wave_function}) converge, Proposition \ref{prop:BE_reconstruction} is correct.
\end{remark}

\section{GKZ equations as quantum curves}\label{sec-gkz_reconst}

In this section, we will prove Theorem \ref{thm:reconstruction} 
(Theorem \ref{theorem:reconstruction-of-GKZ-general}, \ref{theorem:reconstruction-of-GKZ-general_p}) by 
reconstructing GKZ equations as quantum curves from GKZ curves for  
the equivariant Gromov-Witten theory of the projective space $\mathbb{C}\textbf{P}^{N-1}_{\bm{w}}$ and 
the Fano complete intersection $X_{\bm{w},\bm{\lambda}}=X_{\bm{l}=\bm{1};\bm{w},\bm{\lambda}}$ 
of degree $l_{a=1,\ldots,n}=1$. 
For this purpose we employ two methods developed in works by 
Mulase-Su{\l}kowski \cite{Mulase:2012tm} (in Section \ref{subsec:Mulase_Sulkowski}) 
and Bouchard-Eynard  \cite{Bouchard:2016obz} (in Section \ref{subsec:Bouchard_Eynard}).

The former method uses the local topological recursion \eqref{top_recursion} 
and the recursion relation for $S_m$'s in the WKB expansion \eqref{wkb_x1} 
obtained from the GKZ equation is manifestly reconstructed for two examples: 
(1) $\mathbb{C}\textbf{P}^{1}_{\bm{w}}$ and 
(2) the degree $1$ Fano hypersurface $X_{\bm{w},\lambda}$ in $\mathbb{C}\textbf{P}^{1}_{\bm{w}}$. 
Such manifest results are helpful to find the relation between the wave function and the WKB solution which will be studied 
 in Section \ref{section:Stokes-eqP1} (e.g. Lemma \ref{lem:limit-WKB-normalization} below).
 
The latter method uses the global topological recursion \eqref{g_top_recursion} 
and we can reconstruct more general class of $X$ such that
$X=\mathbb{C}\textbf{P}^{N-1}_{\bm{w}}$ and the Fano complete intersection 
$X=X_{\bm{w},\bm{\lambda}}$ of degree $l_{a=1,\ldots,n}=1$.
For these class of $X$ the admissibility condition in 
Definition \ref{def:admissibility} is satisfied, and then we can show that 
the GKZ equation is reconstructed as a quantum curve 
for the wave function $\psi_X(D)$ associated with 
the GKZ curve $\Sigma_X$ by appropriately choosing 
an integration divisor $D$ in \eqref{wave_function}. 
Here, by Proposition \ref{prop:BE_reconstruction} the WKB reconstruction 
of quantum curves depends on the choice of integration divisor $D$. 
We show that there actually exists the integration 
divisor $D=D^*$ which realizes the GKZ equation.
In Section \ref{subsection:relation-to-oscillatory-integral-section4} we 
also show an explicit relation between the wave function $\psi_X(D^*)$ and 
the oscillatory integral ${\mathcal I}_X(x)$ in \eqref{eq:equiv-osci-int}.

\subsection{The GKZ equation from the local topological recursion {\`a} la  Mulase-Su{\l}kowski}\label{subsec:Mulase_Sulkowski}

In \cite{Mulase:2012tm} the second order ordinary differential equation for the wave function $\psi(x)$ was derived via the local topological recursion (\ref{top_recursion}) in Definition \ref{def:top_rec1}. 
In the following, we will reconstruct the GKZ equation from the data of the GKZ curve.
Among GKZ equations discussed in this article, this derivation is applicable to the following two models.

\vspace{0.2cm}
\noindent{(1) 
Equivariant Gromov-Witten theory of the projective space 
$\mathbb{C}\textbf{P}^1_{\bm{w}}$:
}
\begin{align}
\left(\hbar x\frac{d}{dx}-w_0\right)\left(\hbar x\frac{d}{dx}-w_1\right) \psi(x)
= x \psi(x).
\label{GKZ_eqv_CP1}
\end{align}
\noindent{(2) Equivariant Gromov-Witten theory of the degree $1$ hypersurface  
$X_{\bm{w},\lambda}$
in $\mathbb{C}\textbf{P}^1_{\bm{w}}$:
}
\begin{align}
\left(\hbar x\frac{d}{dx}-w_0\right)\left(\hbar x\frac{d}{dx}-w_1\right) \psi(x)
= x\left(\hbar x\frac{d}{dx}-\lambda+\hbar\right) \psi(x).
\label{GKZ_eqv_deg1_hyp}
\end{align}
\vspace{0.2cm}
The defining equation of the GKZ curve $\Sigma_X$ is directly found by replacements $x\to x$, $\hbar xd/dx\to y$, and $\hbar\to 0$ in the GKZ equation, 
and indeed, one finds the GKZ curves for these models
\begin{align}
&
\Sigma_{\mathbb{C}\textbf{P}^1_{\bm{w}}}=\left\{\,
(x,y)\in {\IC}^*\times {\IC}\;\Big|\; y^2-(w_0+w_1)y+w_0w_1-x=0\,\right\},
\label{GKZ1}
\\
&
\Sigma_{X_{\bm{w},\lambda}}
=\left\{\,
(x,y)\in {\IC}^*\times {\IC}\;\Big|\; y^2-(w_0+w_1+x)y+w_0w_1+
\lambda x=0\,\right\}.
\label{GKZ2}
\end{align}

To apply the topological recursion to the above GKZ curves, we need to use appropriate local coordinates\footnote{ 
In \cite{Mulase:2012tm}, such local coordinates are specified by the Lagrangian singularity of $\Sigma_X$.} 
to pick up residues in the topological recursion (\ref{top_recursion}) systematically. In the following for the case (1) we will introduce a local coordinate $z$ as
\begin{align}
x(z)=z^2-\Lambda,\qquad
y(z)=z+\frac{1}{2}(w_0+w_1),\qquad
\Lambda=\frac{1}{4}(w_0-w_1)^2,
\label{local_coord_CP1}
\end{align}
and for the case (2) we will introduce a local coordinate $z$ via the Zhukovsky coordinate $u$ as
\begin{align}
\begin{split}
&
x(z)=\frac{\alpha+\beta}{2}+\frac{\beta-\alpha}{4}\left(u(z)+\frac{1}{u(z)}\right),\\
&
y(z)=\frac{w_0+w_1+x(z)}{2}+\frac{\beta-\alpha}{8}\left(u(z)-\frac{1}{u(z)}\right), \\
&\alpha+\beta=-2(w_0+w_1-2\lambda),\quad \beta-\alpha=4\sqrt{(\lambda-w_0)(\lambda-w_1)},\quad
u(z)=\frac{z+1}{z-1}.
\label{local_coord_hyp}
\end{split}
\end{align}
We assume that the parameters $w_i$ and $\lambda_a$ are generic so that the $dx$ and $dy$ do not have common zero. 
Then we can identify the GKZ curve $\Sigma_X$ with the spectral curve 
$(C = {\mathbb C}{\bf P}^1, x(z), y(z))$ in the sense of Definition \ref{def:spectral-curve-for-TR}.
For both of these local coordinates, we see that the local Galois conjugation $\sigma$ near the branch point $z=0\in \Sigma_X$ acts as $\sigma(z)=-z$.

The basic building blocks of the topological recursion (\ref{top_recursion}) are given in the above local coordinates as follows:
\begin{align}
&
\omega(x(z))=y(z)\frac{dx(z)}{x(z)}=\frac{y(z)}{x(z)}\frac{dx(z)}{dz}dz, 
\qquad
B(z_1,z_2)=\frac{dz_1dz_2}{(z_1-z_2)^2}.
\nonumber
\end{align}
And we obtain
\begin{align}
\begin{split}
&
\frac{\frac{1}{2}\int^z_{-z}B(\cdot,z_1)}{\omega(x(z))-\omega(x(-z))}
=p(z)\left(\frac{1}{z+z_1}+\frac{1}{z-z_1}\right)\frac{dz_1}{dz},
\\
&
p(z)=-\frac{x(z)}{2\big(y(z)-y(-z)\big)\frac{dx(z)}{dz}},
\nonumber
\end{split}
\end{align}
where for the above two models the function $p(z)$ is given by
\begin{align}
&
p(z)=-\frac{z^2-\Lambda}{8z^2}, \qquad \textrm{for } X=\mathbb{C}\textbf{P}^1_{\bm{w}},
\nonumber\\
&
p(z)=\frac{(z^2-1)^2\left((w_0+w_1-2\lambda)(1-z^2)+2(1+z^2)\sqrt{(\lambda-w_0)(\lambda-w_1)}\right)}{64z^2(\lambda-w_0)(\lambda-w_1)},\quad \textrm{for } X=X_{\bm{w},\lambda}.
\nonumber
\end{align}
Both of these functions do not have poles except for $z=0, \infty$.

Adopting the ingredients one finds that the topological recursion (\ref{top_recursion}) is rewritten as follows:
\begin{align}
&
\omega_n^{(g)}(z_1,\ldots,z_n)
\nonumber \\
&=\frac{1}{2\pi\mathrm{i}}\int_{\gamma}p(z)\left(\frac{1}{z+z_1}+\frac{1}{z-z_1}\right)\frac{dz_1}{dz}
\nonumber \\
&\quad\quad\times\biggl[
\sum_{j=2}^{n}\Bigl(
B(z,z_j)\omega_{n-1}^{(g)}(-z,z_2,\ldots,\widehat{z}_j,\ldots,z_n)
+B(-z,z_j)\omega_{n-1}^{(g)}(z,z_2,\ldots,\widehat{z}_j,\ldots,z_n)
\Bigr)
\nonumber \\
&\quad\quad\quad+\omega_{n+1}^{(g-1)}(z,-z,z_2,\ldots,z_n)
+\sum_{\substack{g_1+g_2=g\\ I\sqcup J=\{2,3,\ldots,n\}}}^{\prime}
\omega_{|I|+1}^{(g_1)}(z,\bm{z}_I)\omega_{|J|+1}^{(g_2)}(-z,\bm{z}_J)
\biggr],
\label{top_rec2}
\end{align}
where the contour $\gamma$ encloses an annulus bounded by two concentric circles centered at the origin encircles $z=\pm z_i$ ($i=1,\ldots,n$) as depicted in Figure \ref{fig:contour0}. The prime in the last summation means that $\omega_{1}^{(0)}$ and $\omega_{2}^{(0)}$ are excluded from the summation. 
Proceeding along the same line of the proof for Theorem 4.1 in \cite{Mulase:2012tm}, we can show the following lemma.
\begin{figure}[t]
\centering
\includegraphics[width=5cm,height=5cm,keepaspectratio]{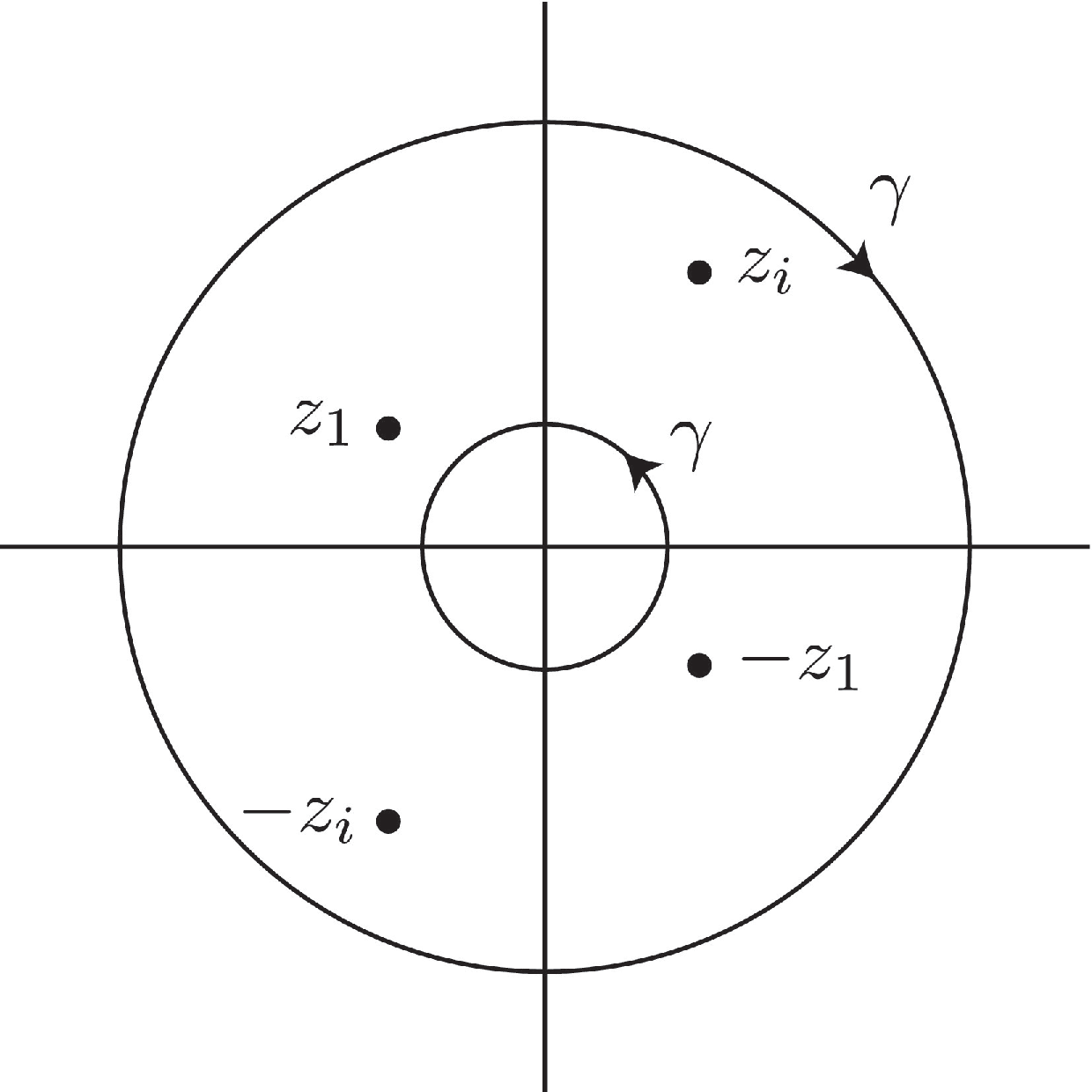}%
\caption{\label{fig:contour0}The contour $\gamma$ in $z$-plane consists of two concentric circles centered at the origin. The  inner circle which encircles the origin has an infinitesimally small radius with the positive orientation.  The outer circle which encircles $z=\pm z_i$ ($i=1,\ldots,n$) is oriented negatively.}
\end{figure}
\begin{lemm}
Let $F_{n}^{(g)}(z_1,\ldots,z_n;z_*)$ ($2g-2+n>0$) denote a function such that
\begin{align}
F_{n}^{(g)}(z_1,\ldots,z_n;z_*)=\int_{z_*}^{z_1}\cdots\int_{z_*}^{z_n}\omega_{n}^{(g)}(z_1',\ldots,z_n'),
\label{free_en_ms}
\end{align}
where $z_*$ denotes a reference point of the integration.
If the meromorphic function $p(z)$ does not have any poles except for $z=0, \infty$, the topological recursion (\ref{top_rec2}) leads to the following recursion relation:
\begin{align}
\frac{\partial}{\partial z_1}F_n^{(g)}(\bm{z}_{[n]};z_*)
&=\sum_{j=2}^{n}\frac{z_j}{z_1^2-z_j^2}\left(4p(z_1)\frac{\partial}{\partial z_1}F_{n-1}^{(g)}(\bm{z}_{[\hat{j}]};z_*)-4p(z_j)\frac{\partial}{\partial z_j}F_{n-1}^{(g)}(\bm{z}_{[\hat{1}]};z_*)\right)
\nonumber \\
&\quad-\sum_{j=2}^{n}\frac{z_*}{z_1^2-z_*^2}4p(z_1)\frac{\partial}{\partial z_1}F_{n-1}^{(g)}(\bm{z}_{[\hat{j}]};z_*)
\nonumber \\
&\quad+2p(z_1)\left[
\frac{\partial^2}{\partial u_1\partial u_2}F_{n+1}^{(g)}(u_1,u_2,\bm{z}_{[\hat{1}]};z_*)
\right]\Bigg|_{u_1=u_2=z_1}
\nonumber \\
&\quad+2p(z_1)\sum_{\substack{g_1+g_2=g\\ I\sqcup J=\{2,3,\ldots,n\}}}^{\prime}
F_{|I|+1}^{(g_1)}(z_1,\bm{z}_I;z_*)F_{|J|+1}^{(g_2)}(z_1,\bm{z}_J;z_*),
\label{rec_free_energy}
\end{align}
where $[n]=\{1,2,\ldots,n\}$ and $[\hat{j}]=\{1,2,\ldots,\hat{j},\ldots,n\}$.
\end{lemm} 

In addition to (\ref{free_en_ms}) we also define
\begin{align}
\begin{split}
&
F_{1}^{(0)}(z_1;z_{*})=\int_{z_{*}}^{z_1}\omega(x(z_1'))
=\int_{z_{*}}^{z_1}y(z_1')\frac{dx(z_1')}{x(z_1')},
\\
&
F_{2}^{(0)}(z_1,z_2;z_{*})=\int_{z_{*}}^{z_1}\int_{z_{*}}^{z_2}
\left(B(z_1',z_2')-\frac{dx(z_1')dx(z_2')}{(x(z_1')-x(z_2'))^2}\right)
=\int_{z_{*}}^{z_1}\int_{z_{*}}^{z_2}\frac{dz_1'dz_2'}{(z_1'+z_2')^2},
\label{free_en_ms_us}
\end{split}
\end{align}
where in the second equality of $F_{2}^{(0)}(z_1,z_2;z_{*})$ we have used the local coordinates (\ref{local_coord_CP1}) and (\ref{local_coord_hyp}). Note that these integrals should be regularized so that the integrals converge by adding certain constants which depend on the reference point $z_{*}$. Following the WKB reconstruction (\ref{wave_function}), we define the wave function $\psi_X(x)$ such that
\begin{align}
\psi_X(x)=\exp\left(
\sum_{m=0}^{\infty}\hbar^{m-1}F_m(x)\right),\qquad
F_m(x(z)) = \sum_{\substack{g\ge 0,\,n\ge 1 \\2g-1+n=m}}\frac{1}{n!}F_{n}^{(g)}(z,\ldots,z;z_*).
\label{phase_ms}
\end{align}
Adopting technical identities developed in Lemma A.1 of \cite{Mulase:2012tm} to the recursion relation (\ref{rec_free_energy}) we arrive at the following lemma.
\begin{lemm}
The functions $F_m$ ($m\ge 2$) obey the following recursion relation:
\begin{align}
\frac{d}{dz}F_{m+1}=2p(z)\Bigg(
\frac{d^2}{dz^2}F_m+\sum_{\substack{a+b=m+1\\a,b\ge 2}}
\frac{dF_a}{dz}\frac{dF_b}{dz}
\Bigg)+\left(
2\frac{dp(z)}{dz}-4\frac{z_*}{z^2-z_*^2}p(z)
\right)\frac{d}{dz}F_m.
\label{top_rec_F_m}
\end{align}
\end{lemm}

Now we will show that the above recursion relation (\ref{top_rec_F_m}) agrees with the recursion relation for $S_m$'s found from the GKZ equation for two models. 
\begin{prop} \label{prop:quantum-curve-equivariant-CP1}
For $X = \mathbb{C}\textbf{P}^1_{\bm{w}}$, the wave function $\psi_X(x)$ defined by (\ref{phase_ms}) with the integration divisor $D=[z]-[\infty]$ in the local coordinate (\ref{local_coord_CP1}) satisfies the GKZ equation (\ref{GKZ_eqv_CP1}) for the equivariant Gromov-Witten theory of $\mathbb{C}\textbf{P}^1_{\bm{w}}$. 
Here we choose the reference point $z_* = \infty$ so that $x(z_*)=\infty$ holds.
\end{prop}

\begin{proof}
Adopt the WKB expansion
\begin{align}
\psi(x) = \exp\left(\sum_{m=0}^{\infty}\hbar^{m-1}S_m(x)\right)
\label{wkb_x1} 
\end{align}
into the GKZ equation (\ref{GKZ_eqv_CP1}), then one finds a hierarchy of differential equations for $S_m$'s:
\begin{align}
x\left(
\frac{d^2}{dx^2}S_m+\sum_{a+b=m+1}\frac{dS_a}{dx}\frac{dS_b}{dx}
\right)-(w_0+w_1)\frac{d}{dx}S_{m+1}+\frac{d}{dx}S_m=0.
\label{hierarchy21}
\end{align}
The remaining $\hbar^{0}$-terms are treated separately as follows:
\begin{align}
\left(x\frac{dS_0}{dx}\right)^2-(w_0+w_1)x\frac{dS_0}{dx}+w_0w_1-x=0.
\label{hierarchy21_sep}
\end{align}
By definitions (\ref{free_en_ms_us}) and (\ref{phase_ms}) we see that $F_0(x(z))=F_1^{(0)}(z;z_*)$ obeys the differential equation (\ref{hierarchy21_sep}) by $S_0=F_0$, and one obtains the GKZ curve (\ref{GKZ1}) by a replacement $xdS_0/dx =y$. 

Next we will consider the relation between $S_1$ and $F_1$. The subleading term $S_1$ of the WKB expansion (\ref{wkb_x1}) is computed from the recursion relation (\ref{hierarchy21}) for $m=0$.
\begin{align}
\frac{dS_1}{dx}=-\frac{1}{4x+(w_0-w_1)^2}=-\frac{1}{4z^2},
\nonumber
\end{align}
where in the second equality we have used the local coordinate (\ref{local_coord_CP1}). 
On the other hand, by definitions (\ref{free_en_ms_us}) and (\ref{phase_ms}) one finds
\begin{align}
\begin{split}
&
F_1(x(z))=\frac{1}{2}F_2^{(0)}(z,z;z_*)=\frac{1}{2}\log\frac{(z+z_*)^2}{4zz_*},
\\
&
\frac{dF_1}{dx}\Bigg|_{x=x(z)}=\frac{1}{2}\left(\frac{dx}{dz}\right)^{-1}\left(\frac{2}{z+z_*}-\frac{1}{z}\right).
\label{S1_1}
\end{split}
\end{align}
By comparison of these results,  we see that for the specialization $z_*=\infty$ with $x(z_*)=\infty$, $dF_1/dx$ in  (\ref{S1_1}) agrees with $dS_1/dx$.

Now rewriting the second term in  (\ref{hierarchy21}):
\begin{align}
\sum_{a+b=m+1}\frac{dS_a}{dx}\frac{dS_b}{dx}=2\frac{dS_0}{dx}\frac{dS_{m+1}}{dx}+2\frac{dS_1}{dx}\frac{dS_m}{dx}+\sum_{\substack{a+b=m+1\\a,b\ge 2}}\frac{dS_a}{dx}\frac{dS_b}{dx},
\label{eq:recursion-for-S-equive-CP1}
\end{align}
one obtains
\begin{align}
-\left(2x\frac{dS_0}{dx}-(w_0+w_1)\right)\frac{d}{dx}S_{m+1}=x\frac{d^2}{dx^2}S_m+\left(2x\frac{dS_1}{dx}+1\right)\frac{d}{dx}S_m
+x\sum_{\substack{a+b=m+1\\a,b\ge 2}}\frac{dS_a}{dx}\frac{dS_b}{dx}.
\label{hierarchy31}
\end{align}
To switch $x$-derivatives in the above recursion to $z$-derivatives, 
one can use
\begin{align}
\frac{d}{dx}=\left(\frac{dx}{dz}\right)^{-1}\frac{d}{dz},\qquad 
\frac{d^2}{dx^2}=\left(\frac{dx}{dz}\right)^{-2}\frac{d^2}{dz^2}-\left(\frac{dx}{dz}\right)^{-3}\frac{d^2x}{dz^2}\frac{d}{dz}.
\label{ddd}
\end{align}
Plugging  $xdS_0/dx =y$ and (\ref{ddd}) into  (\ref{hierarchy31}), one gets the recursion relation for $S_m$'s
\begin{align}
\frac{d}{dz}S_{m+1}=c_1(z)\Bigg(
\frac{d^2}{dz^2}S_m+\sum_{\substack{a+b=m+1\\a,b\ge 2}}
\frac{dS_a}{dz}\frac{dS_b}{dz}
\Bigg)+c_2(z)\frac{d}{dz}S_m, 
\label{rec_x1_result}
\end{align}
where
\begin{align}
&
c_1(z)=\frac{-x(z)}{(2y(z)-(w_0+w_1))\frac{dx}{dz}}, 
\nonumber\\
&
c_2(z)=\frac{-1}{2y(z)-(w_0+w_1)}\left(-x(z)\left(\frac{dx}{dz}\right)^{-2}\frac{d^2x}{dz^2}+2x(z)\frac{dS_1}{dx}\bigg|_{x=x(z)}+1\right).
\nonumber
\end{align}
Using the local coordinate (\ref{local_coord_CP1}), after some short computations, one obtains
\begin{align}
c_1(z)=-\frac{z^2-\Lambda}{4z^2}=2p(z), \qquad c_2(z)=-\frac{\Lambda}{2z^3}=2\frac{dp(z)}{dz}.
\nonumber
\end{align}
As a consequence, it is found that the recursion relation (\ref{rec_x1_result}) agrees with the recursion relation (\ref{top_rec_F_m}) for $F_m$'s of $X = \mathbb{C}\textbf{P}^1_{\bm{w}}$ under the choice $z_*=\infty$.
\end{proof}

\begin{prop} \label{prop:quantum-curve-comp1-in-CP1}
For the degree $1$ hypersurface $X = X_{\bm{w},\lambda}$
in $\mathbb{C}\textbf{P}^1_{\bm{w}}$, the wave function $\psi_X(x)$ defined by (\ref{phase_ms}) with the integration divisor $D=[z]-[-1]$ in the local coordinate (\ref{local_coord_hyp}) satisfies the GKZ equation \eqref{GKZ_eqv_deg1_hyp} for the equivariant Gromov-Witten theory of 
$X_{\bm{w},\lambda}$.
Here we choose the reference point $z_* = -1$ so that $x(z_*)=\infty$ holds.
\end{prop}

\begin{proof}
Adopt the WKB expansion (\ref{wkb_x1}) into the GKZ equation (\ref{GKZ_eqv_deg1_hyp}), then one finds a hierarchy of differential equations for $S_m$'s:
\begin{align}
-\left(2x\frac{dS_0}{dx}-(w_0+w_1+x)\right)\frac{d}{dx}S_{m+1}=x\frac{d^2}{dx^2}S_m+\left(2x\frac{dS_1}{dx}+1\right)\frac{d}{dx}S_m
+x\sum_{\substack{a+b=m+1\\a,b\ge 2}}\frac{dS_a}{dx}\frac{dS_b}{dx}.
\label{hierarchy32}
\end{align}
In particular, 
 $S_0$ obeys
\begin{align}
\left(x\frac{dS_0}{dx}\right)^2-(w_0+w_1+x)x\frac{dS_0}{dx}+w_0w_1+\lambda x=0,
\end{align}
and this differential equation gives the GKZ curve (\ref{GKZ2}) by $y=xdS_0/dx$.

From the recursion relation (\ref{hierarchy32}) for $m=0$, it is found that
\begin{align}
\frac{dS_1}{dx}&=-\frac{w_0+w_1+x-2\lambda+\sqrt{x^2+2(w_0+w_1-2\lambda)x+(w_0-w_1)^2}}{2(x^2+2(w_0+w_1-2\lambda)x+(w_0-w_1)^2)}
\nonumber\\
&=-\frac{(z+1)^3(z-1)}{16z^2\sqrt{(\lambda-w_0)(\lambda-w_1)}},
\nonumber
\end{align}
where in the second equality the local coordinate (\ref{local_coord_hyp}) is adopted. On the other hand, by definitions (\ref{free_en_ms_us}) and (\ref{phase_ms}) one obtains  (\ref{S1_1}). We see that for the specialization $z_*=-1$ which corresponds to $x(-1)=\infty$, $dF_1/dx$ in  (\ref{S1_1}) agrees with $dS_1/dx$.

Switching from $x$-coordinate to the local $z$-coordinate, one finds that the recursion relation (\ref{hierarchy32}) is rewritten as
\begin{align}
&\frac{d}{dz}S_{m+1}=c_1(z)\Bigg(
\frac{d^2}{dz^2}S_m+\sum_{\substack{a+b=m+1\\a,b\ge 2}}
\frac{dS_a}{dz}\frac{dS_b}{dz}
\Bigg)+c_2(z)\frac{d}{dz}S_m,
\label{rec_x2_result}
\end{align}
where
\begin{align}
&
c_1(z)=\frac{-x(z)}{(2y(z)-(w_0+w_1+x(z)))\frac{dx}{dz}}, 
\nonumber\\
&
c_2(z)=\frac{-1}{2y(z)-(w_0+w_1+x(z))}\left(-x(z)\left(\frac{dx}{dz}\right)^{-2}\frac{d^2x}{dz^2}+2x(z)\frac{dS_1}{dx}\bigg|_{x=x(z)}+1\right).
\nonumber
\end{align}
Using the local coordinate (\ref{local_coord_hyp}), after some short computations we obtain
\begin{align}
c_1(z)=2p(z), \qquad c_2(z)=2\frac{dp(z)}{dz}+\frac{4}{z^2-1}p(z).
\nonumber
\end{align}
Thus if one chooses $z_*=-1$ s.t. $x(-1)=\infty$, the recursion relation (\ref{rec_x2_result}) agrees with the recursion relation (\ref{top_rec_F_m}) for $F_m$'s for $X = X_{\bm{w},\lambda}$
under the choice $z_*=-1$.
\end{proof}

Propositions \ref{prop:quantum-curve-equivariant-CP1} and \ref{prop:quantum-curve-comp1-in-CP1} show that the GKZ equations \eqref{GKZ_eqv_CP1} and \eqref{GKZ_eqv_deg1_hyp} are reconstructible as 
quantum curves associated with the GKZ curves \eqref{GKZ1} and \eqref{GKZ2}, respectively.

\subsection{The GKZ equation from the global topological recursion  {\`a} la Bouchard-Eynard}\label{subsec:Bouchard_Eynard}

In Section \ref{subsec:Mulase_Sulkowski}, 
it was proven for the two models (\ref{GKZ_eqv_CP1}) and (\ref{GKZ_eqv_deg1_hyp}) 
that
their GKZ equations are reconstructible by the local topological recursion (\ref{top_recursion}), if endpoints of integrals of $\omega_n^{(g)}(z_1,\ldots,z_n)$ in the WKB reconstruction (\ref{wave_function}) are chosen to be at $x(z_*)=\infty$ in the global coordinate $x=x(z)$. To generalize these results, we will apply the consequences of the WKB reconstruction of quantum curves in \cite{Bouchard:2016obz} (summarized shortly in Section \ref{subsec-rec_q_curve}) 
to the GKZ curve (\ref{A_comp}) for the equivariant Gromov-Witten theory of $\mathbb{C}\textbf{P}^{N-1}_{\bm{w}}$ and the complete intersection  $X_{\bm{w},\bm{\lambda}}=X_{\bm{l}=\bm{1};\bm{w},\bm{\lambda}}$ 
of the degree $l_a=1$ ($a=1,\ldots,n$) hypersurfaces in $\mathbb{C}\textbf{P}^{N-1}_{\bm{w}}$. 
The defining polynomial $A_{X_{\bm{w},\bm{\lambda}}}$
of the GKZ curve for this model is
\begin{align}
A_{X_{\bm{w},\bm{\lambda}}}(x,y)=
\prod_{i=0}^{N-1}(y-w_i)-x\prod_{a=1}^n(y-\lambda_a),\qquad
x\in {\IC}^*,\ \ y\in {\IC},
\label{gkz_hld1}
\end{align}
and the GKZ equation is reconstructed subsequently.

In order to apply the consequences of the WKB reconstruction directly to the GKZ curve, 
we will change the presentation of the GKZ curve (\ref{gkz_hld1}):
$$
\Sigma_{X_{\bm{w},\bm{\lambda}}}=
\Big\{\, (x,y)\in {\IC}^*\times {\IC}\; \Big|\; 
A_{X_{\bm{w},\bm{\lambda}}}(x,y)=0\, \Big\},\qquad
\omega(x) = y(x)\, \frac{dx}{x},
$$
since the spectral curve $\Sigma$ considered in \cite{Bouchard:2016obz} is defined as a Lagrangian (\ref{sp_curve_be}) in $\mathbb{C}^2$:
$$
\Sigma=\Big\{\,
(x,y)\in\mathbb{C}^2\;\Big|\; P(x,y)=\sum_{k=0}^{r}p_{k}(x)y^{r-k}=0\,
\Big\},\qquad
\omega(x) = y(x)dx.
$$
The WKB reconstruction of quantum curves in \cite{Bouchard:2016obz} is based only on the global topological recursion (\ref{g_top_recursion}), and what we need are the 1-form $\omega(x)$ 
and the Bergman kernel on $\Sigma$ as the inputs. Therefore, changing the presentation of the GKZ curve $\Sigma_{X_{\bm{w},\bm{\lambda}}}$
with
\begin{align}
y=xY,\quad 
A_{X_{\bm{w},\bm{\lambda}}}(x,xY)=
P_{X_{\bm{w},\bm{\lambda}}}(x,Y),
\nonumber
\end{align}
we can utilize the remarkable results in \cite{Bouchard:2016obz}. In the following we will consider a local coordinate $z$ defined by
\begin{align}
x(z)=\frac{\prod_{i=0}^{N-1}(z-w_i)}{\prod_{a=1}^{n}(z-\lambda_a)},\qquad
Y(x(z))=\frac{\prod_{a=1}^{n}(z-\lambda_a)}{\prod_{i=0}^{N-1}(z-w_i)}\,z.
\label{gkz_hld1_l}
\end{align}
Now we will prove Theorem \ref{thm:reconstruction} by applying this presentation of the GKZ curve to Proposition \ref{prop:BE_reconstruction}. At first we will show it for  the equivariant Gromov-Witten theory of the Fano complete intersection of degree $l_a=1$ ($a=1,\ldots,n<N$) hypersurfaces in $\mathbb{C}\textbf{P}^{N-1}_{\bm{w}}$.

\begin{thm} \label{theorem:reconstruction-of-GKZ-general}
The GKZ equation (\ref{GKZ_comp}) for the equivariant Gromov-Witten theory of the Fano complete intersection of degree $l_a=1$ ($a=1,\ldots,n<N$) hypersurfaces in $\mathbb{C}\textbf{P}^{N-1}_{\bm{w}}$ is reconstructible 
as a quantum curve for the GKZ curve $\Sigma_{X_{\bm{w},\bm{\lambda}}}$ 
by specifying the integration divisor $D$ to be
\begin{align}
D=[z]-[\infty],
\nonumber
\end{align}
in the local coordinate (\ref{gkz_hld1_l}). Here by $n<N$ the reference point $z_*=\infty$ corresponds to $x(z_*)=\infty$ in the global coordinate $x=x(z)$.\footnote{
The convergence in Remark \ref{rem:wkb_reconst} is confirmed by looking at the $z\to\infty$ behavior of the building blocks of the topological recursion (\ref{g_top_recursion}) such as $Y(x(z))$ and $\int_{w_*}^{w} B(\cdot,z)$.
} 
\end{thm}

\begin{proof}
\begin{figure}[t]
\centering%
\includegraphics[width=7cm,height=7cm,keepaspectratio]{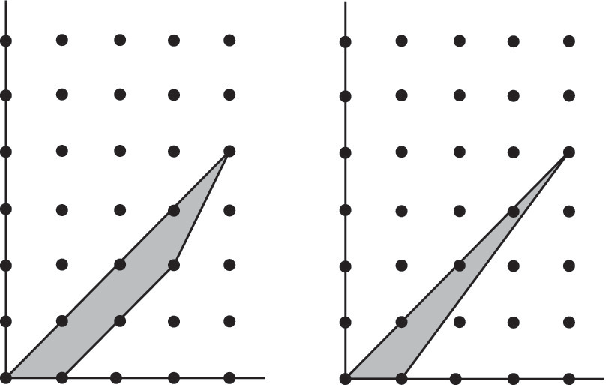}%
\caption{\label{fig:Newton2}The Newton polygons for the defining polynomial $P_{X}(x,Y)$. 
Left: $X$ is the complete intersection of two degree $1$ hypersurfaces in $\mathbb{C}\textbf{P}^{3}_{\bm{w}}$.
Right: $X=\mathbb{C}\textbf{P}^{3}_{\bm{w}}$}
\end{figure}
Consider 
the defining polynomial 
$P_{X_{\bm{w},\bm{\lambda}}}(x,Y)=
A_{X_{\bm{w},\bm{\lambda}}}(x,xY)$
of the GKZ curve for the Fano complete intersection 
$X_{\bm{w},\bm{\lambda}}$
of degree $l_a=1$ ($a=1,\ldots,n< N$) hypersurfaces in $\mathbb{C}\textbf{P}^{N-1}_{\bm{w}}$:
\begin{align}
&
P_{X_{\bm{w},\bm{\lambda}}}(x,Y)=
\prod_{i=0}^{N-1}(xY-w_i)-x\prod_{a=1}^n(xY-\lambda_a)=\sum_{k=0}^{N}p_{k}(x)Y^{N-k},
\label{GKZ0_CP}
\end{align}
where
\begin{align}
p_k(x)=
\left\{
\begin{array}{ll}
(-1)^k e_{k}(\bm{w})x^{N-k}, & \textrm{for } 0\le k\le N-n-1, \\
(-1)^k e_{k}(\bm{w})x^{N-k}-(-1)^{k-N+n}e_{k-N+n}(\bm{\lambda})x^{N-k+1}, & \textrm{for } N-n\le k\le N.
\end{array}
\right.
\nonumber
\end{align}
Here $e_{k}(\bm{w})$ and $e_{k}(\bm{\lambda})$ denote the elementary symmetric polynomials of degree $k$ in $\bm{w}=\{w_0,\ldots,w_{N-1}\}$ and $\bm{\lambda}=\{\lambda_1,\ldots,\lambda_{n}\}$, respectively. 
Clearly this defining polynomial satisfies the admissibility condition in Definition \ref{def:admissibility} (see Figure \ref{fig:Newton2}), and one finds 
\begin{align}
\lfloor\alpha_{r}\rfloor=r,
\qquad r=0,\ldots,N.
\nonumber
\end{align}
For the spectral curve $\Sigma_{X_{\bm{w},\bm{\lambda}}}$
with the defining polynomial (\ref{GKZ0_CP}), the differential operators $D_k$ ($k=1,\ldots,N$) in  (\ref{differential_op}) become
\begin{align}
D_k=\hbar x\frac{d}{dx},
\label{Dk}
\end{align}
and the coefficients $p_{N-k}(x)/x^{\lfloor\alpha_{k}\rfloor}$ in the quantum curve (\ref{WKB_quantum_curve}) are given by
\begin{align}
& 
\frac{p_{N-k}(x)}{x^{\lfloor\alpha_{k}\rfloor}}=
\left\{
\begin{array}{ll}
(-1)^{N-k}e_{N-k}(\bm{w}), & \textrm{for } n+1\le k\le N, \\
(-1)^{N-k}e_{N-k}(\bm{w})-(-1)^{n-k}e_{n-k}(\bm{\lambda})x,
& \textrm{for } 0\le k\le n.
\end{array}
\right.
\label{ak2}
\end{align}
The coefficients $C_k$ in  (\ref{Ck}) for the integration divisor $D=[z]-[\infty]$ are evaluated as follows:
\begin{align}
C_k&=\lim_{z \to\infty}\frac{P_{k+1}(x,Y(x))}{x^{\lfloor\alpha_{N-k}\rfloor+1}}
\nonumber \\
&=\lim_{z \to\infty}\left(
\frac{P_{k+1}(x,Y)}{x^{N-k+1}}-
\frac{P_{X_{\bm{w},\bm{\lambda}}}(x,Y)}{x^{N-k+1}Y^{N-k}}\right)\Bigg|_{Y=Y(x)}
\nonumber \\
&=-\lim_{z \to\infty}\Biggl(
\frac{p_{k}(x)}{x^{N-k+1}}+\frac{p_{k+1}(x)}{x^{N-k+1}Y}+\cdots +\frac{p_N(x)}{x^{N-k+1}Y^{N-k}}
\Biggr)\Bigg|_{Y=Y(x)}
\nonumber\\
&=\left\{
\begin{array}{ll}
0, & \textrm{for } 1\le k\le N-n-1, \\
(-1)^{k-N+n}e_{k-N+n}(\bm{\lambda}),
& \textrm{for } N-n\le k\le N-1,
\end{array}
\right.
\label{Ck3}
\end{align}
where $Y(x)$ obeys $P_{X_{\bm{w},\bm{\lambda}}}(x,Y(x))=0$.
To rewrite the differential equation further, use a key identity for $\ell\ge 1$:
\begin{align}
x\left(\hbar x\frac{d}{dx}+\hbar\right)^{\ell}f(x)=\left(\hbar x\frac{d}{dx}\right)^{\ell-1}\left(\hbar x^2\frac{d}{dx}f(x)\right)+\hbar\left(\hbar x\frac{d}{dx}\right)^{\ell-1}\bigl(xf(x)\bigr).
\label{key_identity}
\end{align}
This identity is proven by induction with respect to $\ell$.
Adopt the equations (\ref{Dk}) -- (\ref{key_identity}) to the equation (\ref{WKB_quantum_curve}), then one finds
\begin{align}
0&=\sum_{k=0}^{N}(-1)^{N-k}e_{N-k}(\bm{w})\left(\hbar x\frac{d}{dx}\right)^k
\psi_{X_{\bm{w},\bm{\lambda}}}(x)
-\sum_{\ell=0}^n(-1)^{n-\ell}e_{n-\ell}(\bm{\lambda})x\left(\hbar x\frac{d}{dx}+\hbar\right)^{\ell}
\psi_{X_{\bm{w},\bm{\lambda}}}(x)
\nonumber \\
&=\prod_{i=0}^{N-1}\left(\hbar x\frac{d}{dx}-w_i\right)
\psi_{X_{\bm{w},\bm{\lambda}}}(x)
-x\prod_{a=1}^{n}\left(\hbar x\frac{d}{dx}-\lambda_a+\hbar\right)
\psi_{X_{\bm{w},\bm{\lambda}}}(x).
\nonumber
\end{align}
Thus the GKZ equation (\ref{GKZ_comp}) for the equivariant Gromov-Witten theory of the Fano complete intersection of degree $l_a=1$ ($a=1,\ldots,n$) hypersurfaces in $\mathbb{C}\textbf{P}^{N-1}$ is correctly reconstructed from the GKZ curve (\ref{GKZ0_CP}). 
\end{proof}

\begin{remark}
Theorem \ref{theorem:reconstruction-of-GKZ-general} ensures the existence of 
the integration divisor $D$ for the WKB reconstruction of 
the GKZ equation (\ref{GKZ_comp}) with $l_a=1$ ($a=1,\ldots,n$) 
as a quantum curve.
\end{remark}
 
Specialized the number $n$ of hypersurfaces to be zero in the above proof,
we immediately find that the GKZ equation for the equivariant $\mathbb{C}\textbf{P}^{N-1}_{\bm{w}}$ theory is also reconstructible.

\begin{thm} \label{theorem:reconstruction-of-GKZ-general_p}
The GKZ equation (\ref{GKZ_CPN}) for the equivariant Gromov-Witten theory of the projective space $\mathbb{C}\textbf{P}^{N-1}_{\bm{w}}$ is reconstructible as 
a quantum curve for the GKZ curve (\ref{A_CPN}) by specifying the integration divisor $D$ to be $D=[z]-[\infty]$
in the local coordinate (\ref{gkz_hld1_l}). Here the reference point $z_*=\infty$ corresponds to $x(z_*)=\infty$ in the global coordinate $x=x(z)$.
\end{thm}

\subsection{Relation to the oscillatory integrals} \label{subsection:relation-to-oscillatory-integral-section4}

As a corollary of results proved in this section, we can find an explicit relation between the oscillatory integral ${\mathcal I}_X$ and the wave function $\psi_X$ for a complete intersection 
$X = X_{\bm{w},\bm{\lambda}}= X_{\bm{l} = {\bm 1};\bm{w},\bm{\lambda}}$ 
of degree $l_a=1$ ($a=1,\ldots,n$) hypersurfaces with $n<N$ in $\mathbb{C}\textbf{P}^{N-1}$ (we regard $X = \mathbb{C}\textbf{P}^{N-1}$ for the case $n=0$). 

Recall that, for generic $w_i$ and $\lambda_a$, there are $N$ critical points $({\bm u}^{\rm (c)}_{1},{\bm v}^{\rm (c)}_{1}), \dots, ({\bm u}^{\rm (c)}_{N}, {\bm v}^{\rm (c)}_{N})$ of $W_X$ (after taking the projection $(\widetilde{{\mathbb C}^\ast})^{N+n-1} \to {({\mathbb C}^\ast)}^{N+n-1}$) which give $N$ solutions ${\mathcal I}_{1}, \dots, {\mathcal I}_{N}$ of 
the GKZ equation \eqref{GKZ_comp} as oscillatory integrals over the associated Lefschetz thimbles. On the other hand, Theorem \ref{theorem:reconstruction-of-GKZ-general} shows that the topological recursion constructs $N$ formal solutions $\psi_{1}, \dots, \psi_N$ of \eqref{GKZ_comp} as the wave function with the integration divisors $D_1, \dots, D_N$ specified as follows. For a fixed $x \in {\mathbb C}^\ast$ away from branch points, we can find $z_1(x), \dots, z_N(x) \in {\mathbb C}{\bf P}^1$ (or ${\Sigma_X}$) satisfying $x(z_i(x)) = x$ ($i=1,\dots,N$), where $x(z)$ is given by \eqref{gkz_hld1_l}. Then, we define $D_i = [z_i(x)] - [\infty]$ ($i=1,\dots,N$).

The correspondence between the critical points and points on the GKZ curve $\Sigma_X$ is given in Proposition \ref{def:GKZ_curve}.  
More explicitly, we can choose a label so that 
\begin{equation} \label{eq:relation-crit-and-spec}
x \frac{d}{dx} W_X({\bm u}^{\rm (c)}_{i}, {\bm v}_i^{\rm (c)};x) = z_i(x), \quad \text{$i = 1, \dots, N$}
\end{equation}
holds.  In view of \eqref{gkz_hld1_l}, for sufficiently large $x$, we can arrange the label\footnote{
Precisely speaking, these labels are not well-defined since the labels exchange if $x$ move around a ramification point ($\infty$ is a ramification point). Here we consider a situation that $x$ moves along a straight path to infinity.} so that 
\begin{eqnarray} \label{eq:behavior-of-zi-x}
z_i(x) \sim 
\begin{cases}
\zeta^i x^{\frac{1}{N-n}}  &  \text{for $i=1,\dots,N-n$,} \\
\lambda_{a_i}  &  \text{for $i=N-n+1, \dots, N$},
\end{cases}
\end{eqnarray}
when $x \to \infty$.
Here $\zeta = \exp(2 \pi {\rm i} / (N-n))$ and $\{a_{N-n+1}, \dots, a_{N} \} = \{1, \dots, n \}$. 
This is consistent with Lemma \ref{lemm:asymptotic-critical-pt}, where $z_1(x), \dots, z_{N-n}(x)$ correspond to the critical points satisfying \eqref{eq:critical-bahavior-1} while $z_{N-n+1}(x), \dots, z_{N}(x)$ correspond to the rest $n$ critical points satisfying \eqref{eq:critical-bahavior-2}.

The following claim shows that the oscillatory integrals associated with the critical points satisfying \eqref{eq:critical-bahavior-1} coincide with 
the wave functions (up to some numerical factor) after taking the asymptotic expansion:

\begin{cor} \label{cor:relation-oscillatory-integral-and-wave-function}
Let $\psi_{i}(x)$ be the wave function \eqref{wave_function} for 
$X = X_{\bm{w},\bm{\lambda}}$
with the integration divisor $D_i=[z_i(x)]-[\infty]$ where the point $z_i(x)$ is specified as \eqref{eq:behavior-of-zi-x}.  
Also, let ${\mathcal I}_{i}$ be the oscillatory integral \eqref{eq:equiv-osci-int} for the mirror Landau Ginzburg potential $W_{X}$ defined over the Lefschetz thimble associated with a critical point $({\bm u}^{\rm (c)}_{i}, {\bm v}_i^{\rm (c)})$ specified by \eqref{eq:relation-crit-and-spec}.
Then, for $i=1,\dots,N-n$, these (formal) solutions of the GKZ equation (\ref{GKZ_CPN}) are related through the asymptotic expansion for $\hbar \to 0$:
\begin{align}
\mathcal{I}_{i}(x) \sim C_i \, (-2\pi \hbar)^{\frac{N+n-1}{2}} \, \psi_{i}(x).
\label{eq:relation-between-oscillatory-and-TR-wave-function} 
\end{align}
Here the constant $C_i$ is determined by 
\begin{equation}
C_i = \lim_{x \to \infty} \frac{u^{\rm (c)}_{i,1} \cdots u^{\rm (c)}_{i,N-1} \, \sqrt{{\rm Hess}({\bm u}^{\rm (c)}_{i}, {\bm v}_i^{\rm (c)})}}{\exp(F_1)},
\end{equation}
where we write $({\bm u}^{\rm (c)}_{i}, {\bm v}_i^{\rm (c)})=(u_{i,1}^{\rm (c)}, \dots, u_{i,N-1}^{\rm (c)}, v_{i,1}^{\rm (c)}, \dots, v_{i,n}^{\rm (c)} )$.
\end{cor}
This follows from the fact that the coefficients $S_m(x)$ of WKB expansion are uniquely determined up to an additive constant (cf. \eqref{eq:recursion-for-S-equive-CP1}) and the behavior of \eqref{eq:normalization-condition-for-S} when $x \to \infty$ which is valid if we choose a critical point satisfying \eqref{eq:critical-bahavior-1}. In particular, for $X =  \mathbb{C}\textbf{P}^{N-1}_{\bm w}$,  the relation \eqref{eq:relation-between-oscillatory-and-TR-wave-function} is valid for all $N$ solutions because $n=0$ in this case.

\section{Several different vantage points of the $J$-function}\label{sec-glsm_J}

In this section we will give a physical derivation of Theorem \ref{thm:reconstruction} (referred to as \textit{reconstruction theorem}) by reinterpreting the equivariant $J$-functions as the brane partition functions in topological strings on local Calabi-Yau 3-folds. In the subsequent sections we will discuss the following vantage points of the $J$-function:
\begin{align}
\begin{split}
&
1.\quad
\textrm{$J$-function as the vortex partition function},
\\
&
2.\quad
\textrm{$J$-function as the brane partition function in the local A-model},
\\
&
3.\quad
\textrm{$J$-function as the brane partition function in the local B-model}.
\nonumber\
\end{split}
\end{align}
At first we will summarize the physical interpretation of the equivariant $J$-function as the vortex partition function in the $\mathcal{N}=(2,2)$ gauged linear sigma model (GLSM) on ${\IS}^2$ \cite{Dimofte:2010tz,Bonelli:2013mma} (Section \ref{subsec-vortex_J}). Next we will reconsider it via the geometric engineering as a particular type of vortex partition function obtained from a brane partition function in the topological A-model on a local toric Calabi-Yau 3-fold $Y$ \cite{Dimofte:2010tz} (Section \ref{subsec-a_J}).
And then, we will move to the local B-model picture via the local mirror symmetry,
and give yet another description of the brane partition function as the wave function via the topological recursion on a mirror curve residing in the mirror local Calabi-Yau 3-fold $Y^{\vee}$ 
on the basis of remodeling conjecture \cite{Bouchard:2007ys} (Section \ref{subsec-b_J}).
As a consequence of physical discussions, we will find the reconstruction theorem.
\begin{remark}
Via the string dualities, we find a novel picture of the $J$-function and GKZ equation. 
One of the most curious but interesting aspects of this picture is the following point.
Originally the $J$-function is defined in regard to the genus $0$ closed string theory on Fano manifold $X$, and the variable $x=\mathrm{e}^{t}$ denotes the \textit{closed string modulus} which measures the area of the closed string worldsheet around the 2-cycle in $X$.
On the other hand, the brane partition function is defined for all genus open string theory on local toric Calabi-Yau 3-fold $Y$ involving a special Lagrangian submanifold $L$, and the variable $x=\mathrm{e}^{u}$ denotes the \textit{open string modulus} which measures the area of the open string worldsheet ending on $L\subset Y$. 
\end{remark}

\subsection{Vantage point 1: $J$-function as the vortex partition function}\label{subsec-vortex_J}

In \cite{Dimofte:2010tz,Bonelli:2013mma} it was argued that the equivariant $J$-function is reinterpreted as the vortex partition function in the $\mathcal{N}=(2,2)$ GLSM on ${\IS}^2$ \cite{Witten:1993yc}. The GLSM consists of gauge multiplet $\mathcal{V}$ with a gauge group $G$ and matter chiral multiplets $\Phi_i$'s with some representations of $G$. The vacuum moduli space of the GLSM is defined by D- and F-term equations. More precisely, the D-term contains Fayet-Iliopoulos (FI) parameters $\bm{\xi}$ (and theta-angles $\bm{\theta}$) associated with the generators of the center of the gauge group $G$, and the F-term is described by a gauge invariant function $W(\bm{\Phi})$ of matter multiplets $\Phi_i$'s called \textit{superpotential}.

Via the renormalization flow, the GLSM flows into 
the geometric regime $\xi\gg 0$, and one finds the non-linear sigma model with a target space $X$ defined by D- and F-term equations, if $X$ is the Fano or Calabi-Yau variety. Indeed the FI parameters in the D-term are associated with the K\"ahler moduli of $X$, and the solutions of the F-term equation are associated with the complex structure moduli of $X$. 
In addition, the $U(1)$ equivariant parameter $\hbar$ on ${\IS}^2$ is introduced, if we consider the A-twisted $\mathcal{N}=(2,2)$ GLSM on the $\Omega$-deformed sphere 
which has the generator of ${\IS}^1$ acting on ${\IS}^2$ \cite{Closset:2015rna} (see also \cite{Benini:2015noa}). The $\Omega$-deformation parameter is given by $\hbar$ and the ${\IS}^1$ action has two fixed points at the north and south poles on ${\IS}^2$.

\begin{table}[htb]
\begin{center}
\begin{tabular}{|c|c|c|c|}
\hline
Field & $U(1)$ & Twisted mass & $U(1)_V$ \\ \hline
$\Phi_i$ & +1 & $-w_i$ & $0$ \\
$P_a$ & $-l_a$ & $\lambda_a$ & $2$ \\
\hline
\end{tabular}
   \vspace{0.3cm}
\caption{Matter content for the complete intersection $X_{\bm{l}}\subset \mathbb{C}\textbf{P}^{N-1}$. Here $i=0, \ldots, N-1$ and $a=1, \ldots, n$.}
\label{glsm_mat}
\end{center}
\end{table}
In the following we will  focus mainly on the smooth complete intersection $X=X_{\bm{l}}$, defined by homogeneous degree $l_{a=1,\ldots,n}$ polynomial equations $F_a(\bm{\phi})=0$ in $\mathbb{C}\textbf{P}^{N-1}\ni (\phi_0:\ldots:\phi_{N-1})$ with $l_1+\cdots +l_n\le N$. The equivariant Gromov-Witten theory of the complete intersection $X_{\bm{l}}$ corresponds to the $G=U(1)$ GLSM with the matter contents listed in Table \ref{glsm_mat} and a superpotential $W(P_{a}, \bm{\Phi})=\sum_{a=1}^{n}P_{a}F_a(\bm{\Phi})$. The D-term equation has the FI parameter $\xi$ and the theta-angle $\theta$, and it realizes $\mathbb{C}\textbf{P}^{N-1}$ as the moduli space. 
On the other hand, twisted masses $-w_i$ (resp. $\lambda_a$)
for the matter multiplets $\Phi_i$ (resp. $P_a$)
are identified with the equivariant parameters of the Gromov-Witten theory. 

The GLSM also has the vector $U(1)_V$ R-symmetry and the superpotential needs to have R-charge $2$. The R-charge of matter multiplets $\Phi_i$ (resp. $P_a$) are assigned to be  $0$ (resp. $2$). In \cite{Closset:2015rna} (see also \cite{Benini:2015noa}) it is found that the A-twisted correlator for a function (operator) $\mathcal{O}^{(\textrm{N})}(\sigma)$ (resp. $\mathcal{O}^{(\textrm{S})}(\sigma)$) of the complex scalar field $\sigma$ in the $U(1)$ gauge multiplet inserted at the north (resp. south) pole of ${\IS}^2$, is given exactly by 
\begin{align}
&\left<\mathcal{O}^{(\textrm{N})}(\sigma)\mathcal{O}^{(\textrm{S})}(\sigma)\right>
\nonumber \\
&=\sum_{d=0}^{\infty}x^d
\oint_{\gamma}\frac{d\sigma}{2\pi i}\;
\frac{\prod_{a=1}^n\prod_{m=0}^{l_a d}\left(l_a\sigma-\lambda_a-\frac{l_ad}{2}\hbar+m\hbar\right)}
{\prod_{i=0}^{N-1}\prod_{m=0}^{d}\left(\sigma-w_i-\frac{d}{2}\hbar+m\hbar\right)}
\;\mathcal{O}^{(\textrm{N})}\big(\sigma-\frac{d}{2}\hbar\big)\mathcal{O}^{(\textrm{S})}\big(\sigma+\frac{d}{2}\hbar\big),
\label{correl_glsm}
\end{align}
where the contour $\gamma\subset\mathbb{C}$ encloses the poles $\sigma=w_i-\frac{d}{2}\hbar+p\hbar$ ($i=0,\ldots,N-1$, $p=0,\ldots,d$) of the integrand. Here
$$
x=\mathrm{e}^{-2\pi \xi+\mathrm{i} \theta},
$$
and 
$x$ must be replaced with $\mu^{N-(l_1+\cdots +l_n)}x$
which is modified by the RG invariant energy scale $\mu$
for the Fano ($l_1+\cdots +l_n< N$) case, 
because the FI parameter $\xi$ runs under the renormalization group (RG) flow. 
In the following we will use the same symbol $x$ for the modified one.
Actually it is found that the correlator (\ref{correl_glsm}) is factorized as \cite{Ueda:2016wfa}:
\begin{align}
\left<\mathcal{O}^{(\textrm{N})}(\sigma)\mathcal{O}^{(\textrm{S})}(\sigma)\right>
=\sum_{i=0}^{N-1}
\oint_{\sigma=w_i}\frac{d\sigma}{2\pi i}\;
Z_{\textrm{1-loop}}(\sigma;\bm{w},\bm{\lambda})
Z_{\textrm{vortex}}^{(\textrm{N})}(\sigma;x,\bm{w},\bm{\lambda},\hbar)
Z_{\textrm{vortex}}^{(\textrm{S})}(\sigma;x,\bm{w},\bm{\lambda},-\hbar),
\nonumber
\end{align}
where
\begin{align}
Z_{\textrm{1-loop}}(\sigma;\bm{w},\bm{\lambda})&=
\frac{\prod_{a=1}^n\left(l_a\sigma-\lambda_a\right)}
{\prod_{i=0}^{N-1}\left(\sigma-w_i\right)},
\nonumber\\
Z_{\textrm{vortex}}^{(\textrm{N},\textrm{S})}(\sigma;x,\bm{w},\bm{\lambda},\hbar)&=\sum_{d=0}^{\infty}x^d
\frac{\prod_{a=1}^n\prod_{m=1}^{l_a d}\left(l_a\sigma-\lambda_a+m\hbar\right)}
{\prod_{i=0}^{N-1}\prod_{m=1}^{d}\left(\sigma-w_i+m\hbar\right)}\;
\mathcal{O}^{(\textrm{N},\textrm{S})}(\sigma-d\hbar).
\nonumber
\end{align}
In this factorization the factor $Z_{\textrm{vortex}}^{(\textrm{N})}$ (resp. $Z_{\textrm{vortex}}^{(\textrm{S})}$) is interpreted as ``off-shell'' vortex partition function with the operator $\mathcal{O}^{(\textrm{N})}(\sigma)$ (resp. $\mathcal{O}^{(\textrm{S})}(\sigma)$) at the north (resp. south) pole of ${\IS}^2$. 
Excluded the operators $\mathcal{O}^{(\textrm{N},\textrm{S})}(\sigma)$, this vortex partition function agrees with the $I$-function for the complete intersection $X_{\bm{l}}$. In particular for the Fano ($l_1+\cdots +l_n< N$) case this agrees with the $J$-function.
By taking the residue at $\sigma=w_0$ ($\sigma=w_i$ in general) in the above correlator, 
we obtain the (``on-shell'') vortex partition function \cite{Shadchin:2006yz,Dimofte:2010tz,Yoshida:2011au,Bonelli:2011fq}:
\begin{align}
Z_{\textrm{vortex}}^{X_{\bm{l};\bm{w},\bm{\lambda}}}(x)=\sum_{d=0}^{\infty}
\frac{x^d}{d!\hbar^d}
\frac{\prod_{a=1}^n\prod_{m=1}^{l_a d}\left(l_aw_0-\lambda_a+m\hbar\right)}
{\prod_{i=1}^{N-1}\prod_{m=1}^{d}\left(w_0-w_i+m\hbar\right)}.
\label{vortex_glsm}
\end{align}
This (``on-shell'') vortex partition function agrees with the (``on-shell'') equivariant $J$-function introduced in Appendix \ref{app:J_function}, and obeys the GKZ equation (see Lemma \ref{lemm:GKZ_calJ}):
\begin{align}
\widehat{A}_{X_{\bm{l};\bm{w},\bm{\lambda}}}(\widehat{x},\widehat{y})\bigl(x^{w_0/\hbar}Z_{\textrm{vortex}}^{X_{\bm{l};\bm{w},\bm{\lambda}}}(x)\bigr)=0.
\label{A_vortex}
\end{align}

\begin{remark}
In \cite{Benini:2012ui,Doroud:2012xw} the $\mathcal{N}=(2,2)$ GLSM partition function on ${\IS}^2$ is computed exactly, and the factorization into the vortex partition functions is shown. In \cite{Bonelli:2013mma} under the identification of the inverse radius $r^{-1}$ of ${\IS}^2$ with the equivariant parameter $\hbar$, this vortex partition function is reinterpreted as the $I$($J$)-function. In \cite{Hori:2013ika} it is shown that the GLSM partition on ${\IS}^2$ is also factorized into two hemisphere partition functions
and one annulus partition function, where the hemisphere partition function is shown to give the D-brane central charge (see also \cite{Sugishita:2013jca,Honda:2013uca}).
\end{remark}

\begin{remark}
In \cite{Benini:2015noa} the twisted partition function of the 3d $\mathcal{N}=2$ gauge theory on ${\IS}^2\times {\IS}^1$ with the $\Omega$-deformation is exactly computed, and shown to be factorized into the K-theoretic vortex partition functions (see also \cite{Benini:2013yva,Fujitsuka:2013fga}). In  \cite{Benini:2015noa} it is also discussed the factorization into the elliptic vortex partition functions of the twisted partition function of the 4d $\mathcal{N}=1$ gauge theory on ${\IS}^2\times T^2$ with the $\Omega$-deformation (see also \cite{Yoshida:2014qwa,Peelaers:2014ima}).
\end{remark}

\subsection{Vantage point 2: $J$-function as the brane partition function in the local A-model}\label{subsec-a_J}

Here we will consider the 4d $\mathcal{N}=2$ gauge theory on $\Omega$-deformed ${\IR}^4\cong {\IC}^2 \ni (z_1,z_2)$ with the gauge group $G=SU(N)$.
In this gauge theory we can put a half-BPS surface operator \cite{Gukov:2006jk} as the codimension 2 defect along the $z_1$-plane $D$ at $z_2=0$.
For the equation of motion of the gauge theory with a surface operator,
the solitonic solutions which are the composite of 4d instantons and 2d vortices can be found, 
and they are called \textit{ramified instantons} \cite{KM1,KM2}.
The moduli space for the ramified instantons is characterized by the flag manifold $G/\mathbb{L}$ 
where  $\mathbb{L}=S\left[U(n_1)\times U(n_2)\times \cdots \times U(n_M)\right]$ with $N=n_1+n_2+\cdots+n_M$ denotes the Levi subgroup of $G=SU(N)$ \cite{Braverman:2010ef,Kanno:2011fw} (see also \cite{Braverman:2004vv,Braverman:2004cr,Alday:2010vg,Kozcaz:2010yp}). 

The generating function of the number of ramified instantons is called the \textit{ramified instanton partition function}.
If we take a (decoupling) limit for the instanton counting parameter in the ramified instanton partition function and suppress the counting of the 4d instantons  (i.e. focus only on the ramified instantons with the instanton number zero), the generating function reduces to the vortex partition function in an $\mathcal{N}=(2,2)$ GLSM described by the map:
\begin{align}
\overline{D}=\mathbb{C}\textbf{P}^1\ \overset{\bm{d}}{\longrightarrow}\ G/\mathbb{L},\qquad
\bm{d} \in H_2(G/\mathbb{L},{\IZ}),
\label{surf_vortex}
\end{align}
where $\overline{D}$ denotes a one-point compactification of the $z_1$-plane $D$.

In particular, the surface operator is referred to as \textit{simple type}, if it has the Levi subgroup $\mathbb{L}=U(1)\times SU(N-1)$.
Here we will consider the surface operator of the simple type and $D=\mathbb{R}^2$ (i.e. $\overline{D}=\mathbb{S}^2$).
In this case we find $G/\mathbb{L}\cong \mathbb{C}\textbf{P}^{N-1}$.

If the 4d gauge theory does not involve any matter fields (i.e. pure Yang-Mills theory),
the resulting GLSM on $\overline{D}=\mathbb{S}^2$ consists of a $U(1)$ vector multiplet and the matter multiplets $\Phi_i$'s 
listed in Table \ref{glsm_mat}, and the superpotential is absent in this GLSM.
By taking the decoupling limit of the ramified instanton partition function, we obtain the  (``on-shell'') vortex partition function $Z_{\textrm{vortex}}^{\mathbb{C}\textbf{P}^{N-1}_{\bm{w}}}$ for this GLSM, and 
it is given by the specialization $n=0$ of (\ref{vortex_glsm}) because multiplets $P_a$'s are absent in this case:
\begin{align}
Z_{\textrm{vortex}}^{\mathbb{C}\textbf{P}^{N-1}_{\bm{w}}}(x)=\sum_{d=0}^{\infty}
\frac{x^d}{d!\hbar^d}
\frac{1}
{\prod_{i=1}^{N-1}\prod_{m=1}^{d}\left(w_0-w_i+m\hbar\right)}.
\label{vortex_glsm_cpN}
\end{align}

If the 4d $U(N)$ gauge theory involves $n$  ($n\le N$) matter hypermultiplets in the fundamental representation, 
we will find the same GLSM that we have considered in the vantage point 1.
The decoupling limit of the ramified instanton partition function agrees with the  (``on-shell'') vortex partition function $Z_{\textrm{vortex}}^{X_{\bm{l};\bm{w},\bm{\lambda}}}(x)$ in (\ref{vortex_glsm}).

Subsequently we will survey on the punchline of the geometric engineering which realizes the (``on-shell'') vortex partition function
as the brane partition function in the open topological A-model. And then we will see how the GKZ equation appears in the open topological A-model on the strip geometry.

\subsubsection{Geometric engineering of the (``on-shell'') equivariant $J$-function for $\mathbb{C}\textbf{P}^{N-1}_{\bm{w}}$}

Our starting point is the open topological A-model on the local toric Calabi-Yau 3-fold $Y_{\bm{l}}$ which is specified by charge vectors $\bm{l}_i\in \mathbb{Z}^{m+3}$ ($i=1,\ldots,m$).
The local toric Calabi-Yau 3-fold $Y_{\bm{l}}$ is the quotient such that
\begin{align}
Y_{\bm{l}}=
\Big\{\;(X_1,\ldots,X_{m+3})\in {\IC}^{m+3}\;
\Big|\;\sum_{\alpha=1}^{m+3}l_{i,\alpha}|X_\alpha|^2=\mathrm{Re}\left(\log Q_i\right)\;
\Big\}/U(1)^m,
\label{local_a_geom}
\end{align}
where $U(1)$ charge vectors $\bm{l}_i=(l_{i,1},\ldots,l_{i,m+3})$, $(i=1,2,\ldots,m)$ obey the Calabi-Yau condition $\sum_{\alpha=1}^{m+3}l_{i,\alpha}=0$.
Here $Q_i$'s denote $Q_i=\mathrm{e}^{t_i}$ with the K\"ahler parameters $t_i$, and $U(1)^m$ acts on $X_{\alpha}$ as $X_{\alpha} \to \mathrm{e}^{\mathrm{i}\sum_{i=1}^m \epsilon_i l_{i,\alpha}}X_{\alpha}$.

If $Y_{\bm{l}}$ is  chosen to be the $A_{N-1}$-fibration over $\mathbb{C}\textbf{P}^1$, the physical spectra (i.e. vector multiplets, hypermultiplets, etc.) of the 4d $U(N)$ gauge theory are realized from the topological A-model in the string theoretical way \cite{Katz:1996fh,Katz:1997eq}.
Such a realization of the 4d gauge theory is known as the \textit{geometric engineering}.

In the framework of the geometric engineering,
the surface operator of the simple type in the 4d gauge theory is realized from the topological A-model on $Y_{\bm{l}}$ 
by introducing the 3d object of the topological A-model referred to as \textit{toric brane}, which wraps around the special Lagrangian submanifold $L \in Y_{\bm{l}}$ \cite{Dimofte:2010tz}.
In a local atlas of $Y_{\bm{l}}$ which covers $X_{\alpha}=X_{\beta}=X_{\gamma}=0$, the special Lagrangian submanifold $L \cong {\IC}\times {\IS}^1$ is found as the following locus  (see Theorem 3.1 in \cite{Harvey:1982}):
\begin{align}
|X_{\alpha}|^2-|X_{\gamma}|^2=\mathrm{Re}\left(\log \mathsf{x}\right),\ \
|X_{\beta}|^2-|X_{\gamma}|^2=0,\ \ \mathrm{Im}(X_\alpha X_{\beta}X_{\gamma})=0,\ \ 
\mathrm{Re}(X_{\alpha}X_{\beta}X_{\gamma})\ge 0,
\label{lag_a_geom}
\end{align}
where $\mathsf{x}=\mathrm{e}^u$ denotes a open string modulus of the toric brane. 
The toric brane is represented by 
a ray attached on one of the lines in the web diagram of the toric variety \cite{Aganagic:2000gs}.
In Figure \ref{local_A_geom}  a toric brane insertion at the lowest leg in the web diagram is depicted.

On the basis of the string theoretical discussions\footnote{There are no mathematically rigorous definition for the brane partition function
 as the generating function of the open Gromov-Witten invariants in general. But switched to the type IIA superstring picture, we can find it as the enumeration of degeneracies of the open BPS states \cite{Ooguri:1999bv} which arise from D0-D2-D4 brane bound states
 for the case of local toric Calabi-Yau 3-fold.  In this sense we have only the string theoretical definition of the brane partition function.},
the brane partition $Z_{\textrm{A-brane}}^{Y_{\bm{l}}}(\mathsf{x})$ is defined as the generating function for the number of the
holomorphic embedding maps of the open Riemann surface (referred to as the world-sheet) which ends on the toric brane $L\in Y_{\bm{l}}$ \cite{Ooguri:1999bv}. 
For the case that $Y_{\bm{l}}$ is the $A_{N-1}$-fibration over $\mathbb{C}\textbf{P}^1$,
the brane partition function is computed by various physical techniques such as the topological vertex \cite{Aganagic:2003db}, the open BPS state counting \cite{Ooguri:1999bv} and the open BPS wall-crossing \cite{Aganagic:2009cg}.
From various observations (see e.g. \cite{Awata:2010bz}),
it is proposed that the brane partition function $Z_{\textrm{A-brane}}^{Y_{\bm{l}}}(\mathsf{x})$ agrees with the K-theoretic generalization of the ramified instanton partition function.

Here we will consider the decoupling limit of the 4d instantons at the level of the toric geometry.
Using the dictionary of the geometric engineering, we find that the decoupling limit of the 4d instantons
corresponds to the large volume limit of the base $\mathbb{C}\textbf{P}^1$ in $A_{N-1}$-fibration over $\mathbb{C}\textbf{P}^1$. 
After taking this large volume limit,  $Y_{\bm{l}}$ reduces to a local toric Calabi-Yau 3-fold $Y_{N-1}$ which consists of the $N-1$ copies of the local Calabi-Yau 3-fold: $\mathcal{O}(-2)\oplus \mathcal{O}(0) \to \mathbb{C}\textbf{P}^1$ (i.e. $(-2,0)$ curve).
More precisely $Y_{N-1}$ is defined by $N-1$ charge vectors $\bm{l}_{i=1,\ldots,N-1}$:
\begin{align}
\begin{split}
\bm{l}_1&=(0,1,-2,1,0,0,0,\ldots,0,0,0),\\
\bm{l}_2&=(0,0,1,-2,1,0,0,\ldots,0,0,0),\\
\bm{l}_3&=(0,0,0,1,-2,1,0,\ldots,0,0,0),\\
\vdots\\
\bm{l}_{N-1}&=(0,0,0,0,0,0,0,\ldots,1,-2,1).
\label{02_chain}
\end{split}
\end{align}
Such a local toric Calabi-Yau 3-fold $Y_{N-1}$ is known as the \textit{strip geometry} \cite{Iqbal:2004ne} of the $(-2,0)$ curves, and 
the web diagram of $Y_3$  is depicted in Figure \ref{local_A_geom}.

\begin{figure}[t]
 \centering
  \includegraphics[width=40mm]{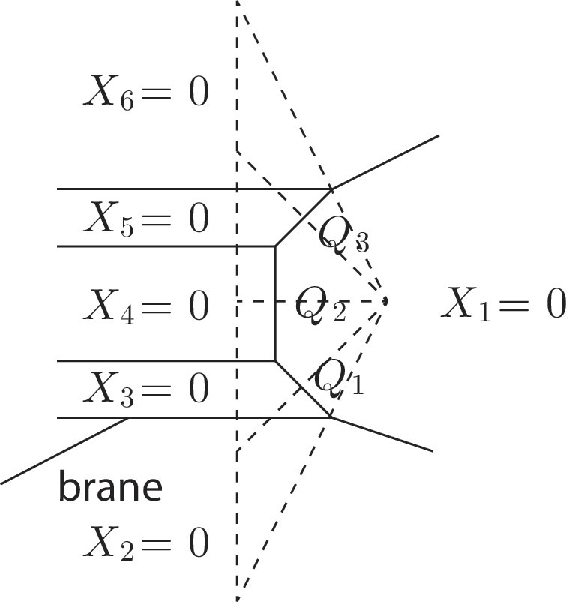}
 \caption{Strip geometry $Y_3$ of three $(-2,0)$ curves. The dashed line and solid line describe the toric diagram and the dual web diagram, respectively. In this diagram a toric brane is inserted at $X_2=0$, and this gives a Lagrangian submanifold (\ref{lag_a_geom}) with $\alpha=1, \beta=3, \gamma=2$.}
\label{local_A_geom}
\end{figure}


Now we will see the brane partition function $Z_{\textrm{A-brane}}^{Y_{N-1}}(\mathsf{x})$ for the strip geometry $Y_{N-1}$ of the $(-2,0)$ curves.
 As a result of topological vertex computations \cite{Iqbal:2004ne,Dimofte:2010tz}, we obtain the manifest form of the brane partition function which is normalized 
s.t. $Z_{\textrm{A-brane}}^{Y_{N-1}}(\mathsf{x}=0)=1$:
\begin{align}
Z_{\textrm{A-brane}}^{Y_{N-1}}(\mathsf{x})=\sum_{d=0}^{\infty}\frac{1}{\prod_{i=0}^{N-1}\prod_{m=1}^{d}(1-\widetilde{Q}_iq^m)}\;
(q^{1/2}\mathsf{x})^{d},
\label{Zopen}
\end{align}
up to the framing ambiguity (see Remark \ref{rem:framing}). Here
$$
\widetilde{Q}_0=1,\qquad
\widetilde{Q}_i=q^{-1}\prod_{1\le j \le i} Q_j,\qquad
q=\mathrm{e}^{-g_s},
$$
and $g_s$ is the topological string coupling constant.



By construction,
the brane partition function $Z_{\textrm{A-brane}}^{Y_{N-1}}(\mathsf{x})$ should agree with the K-theoretic version of the 
vortex partition function.
To obtain the vortex partition function, we need to take a cohomological limit.
For this purpose, we will reparametrize parameters in  $Z_{\textrm{A-brane}}^{Y_{N-1}}(\mathsf{x})$:
$$
g_s=\beta \hbar,\qquad
\mathsf{x}=\beta^{N}x,\qquad
\widetilde{Q}_i=\mathrm{e}^{-\beta (w_{0}-w_{i})}.
$$
After taking the cohomological limit $\beta \to 0$ we find that the brane partition function (\ref{Zopen}) reduces to 
(\ref{vortex_glsm_cpN}) \cite{Dimofte:2010tz,Bonelli:2011fq}:
$$
Z_{\textrm{A-brane}}^{Y_{N-1}}(\mathsf{x})\quad
\mathop{\longrightarrow}^{\beta\to 0}\quad
Z_{\textrm{vortex}}^{\mathbb{C}\textbf{P}^{N-1}_{\bm{w}}}(x).
$$
Thus we see that the vortex partition function $Z_{\textrm{vortex}}^{\mathbb{C}\textbf{P}^{N-1}_{\bm{w}}}(x)$ is found from the open topological A-model on the strip geometry. 
Combing with the consequences in the vantage point 1, we find yet another realization of the (``on-shell'') equivariant  $J$-function for $\mathbb{C}\textbf{P}^{N-1}_{\bm{w}}$ as the brane partition function in the topological A-model on the local toric Calabi-Yau 3-fold $Y_{N-1}$ defined by the charge vectors (\ref{02_chain}) \cite{Dimofte:2010tz}.

\begin{remark}\label{rem:framing}
In the computation (\ref{Zopen}), there is a framing ambiguity $f\in {\IZ}$ of the brane at infinity as
$$
\mathsf{x}^d\ \ \longrightarrow\ \
(-1)^{fd}q^{fd(d-1)/2}\,\mathsf{x}^d.
$$
But this ambiguity becomes irrelevant under the cohomological limit $\beta \to 0$.
\end{remark}

\subsubsection{Geometric engineering of the (``on-shell'') equivariant  $J$-functions for degree $1$ complete intersections in $\mathbb{C}\textbf{P}^{N-1}_{\bm{w}}$}

Next we will consider the geometric engineering of the 4d $SU(N)$ gauge theory with $n$ ($n\le N$) matter hypermultiplets in the fundamental representation. 
For this purpose we choose the local toric Calabi-Yau 3-fold $Y_{\bm{l}}$ to be
the $A_{N-1}$-fibration over $\mathbb{C}\textbf{P}^1$ with blow-ups at $n$ points \cite{Katz:1996fh,Katz:1997eq}.

For a particular case $n=N$, 
after taking the large volume limit of the base $\mathbb{C}\textbf{P}^1$ (which corresponds to the decoupling limit of the 4d instantons),
$Y_{\bm{l}}$ reduces to the strip geometry $Y_{N-1,N}$ which consists of $2N-1$ copies of the local Calabi-Yau $\mathcal{O}(-1)\oplus \mathcal{O}(-1)\to \mathbb{C}\textbf{P}^1$ ($(-1,-1)$ curves). 
Such a strip geometry $Y_{N-1,N}$ is defined by $2N-1$ charge vectors $\bm{l}_{\lambda,i=1,\ldots,N}$ and $\bm{l}_{w,i=1,\ldots,N-1}$ \cite{Iqbal:2004ne}:
\begin{align}
\begin{split}
\bm{l}_{\lambda,1}&=(1,-1,-1,1,0,0,0,\ldots,0,0,0,0),\\
\bm{l}_{w,1}&=(0,1,-1,-1,1,0,0,\ldots,0,0,0,0),\\
\bm{l}_{\lambda,2}&=(0,0,1,-1,-1,1,0,\ldots,0,0,0,0),\\
\vdots\\
\bm{l}_{w,N-1}&=(0,0,0,0,0,0,0,\ldots,-1,-1,1,0),\\
\bm{l}_{\lambda,N}&=(0,0,0,0,0,0,0,\ldots,1,-1,-1,1).
\label{11_chain}
\end{split}
\end{align}
According to  (\ref{local_a_geom}),  the K\"ahler moduli parameters $Q_{\lambda,i}=\mathrm{e}^{t_{\lambda,i}}$ and $Q_{w,i}=\mathrm{e}^{t_{w,i}}$ are associated with the charge vectors $\bm{l}_{\lambda,i}$ and $\bm{l}_{w,i}$, respectively. 

\begin{figure}[t]
 \centering
  \includegraphics[width=45mm]{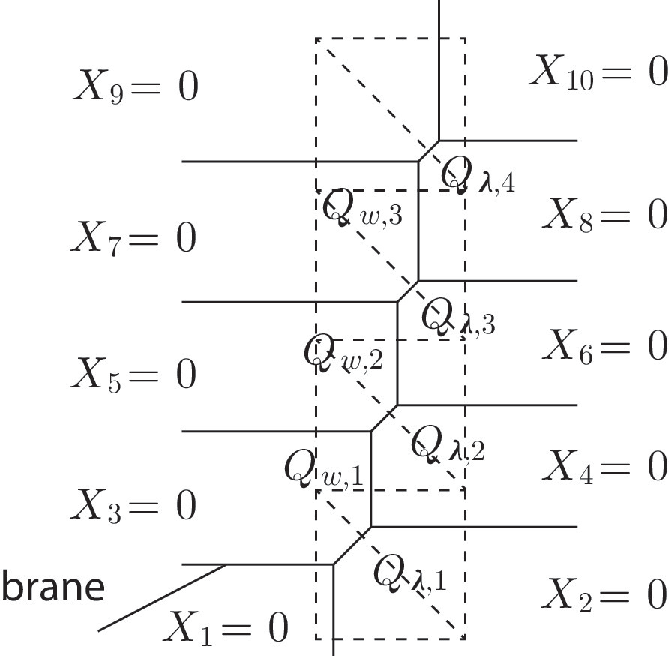}
 \caption{Strip geometry $Y_{3,4}$ consisting of seven $(-1,-1)$ curves ($N=4$). The dashed line and solid line describe the toric diagram and the dual web diagram, respectively. In the diagram a toric brane is inserted at $X_1=0$, and this gives a Lagrangian submanifold (\ref{lag_a_geom}) with $\alpha=2, \beta=3, \gamma=1$.}
\label{local_A_geom_m}
\end{figure}

Now we will introduce a toric brane which warps around the Lagrangian submanifold (\ref{lag_a_geom}) in $Y_{N-1,N}$.
In Figure \ref{local_A_geom_m}, a toric brane wrapping around the Lagrangian submanifold with $\alpha=2, \beta=3, \gamma=1$ is depicted as an insertion in the lowest leg in the web diagram.
For this geometric set-up, the brane partition function $Z_{\textrm{A-brane}}^{Y_{N-1,N}}(\mathsf{x})$  is computed in \cite{Iqbal:2004ne,Dimofte:2010tz}:
\begin{align}
Z_{\textrm{A-brane}}^{Y_{N-1,N}}(\mathsf{x})=\sum_{d=0}^{\infty}\frac{\prod_{i=1}^{N}\prod_{m=1}^{d}(1-\widetilde{Q}_{\lambda,i}q^{m-1})}{\prod_{i=0}^{N-1}\prod_{m=1}^{d}(1-\widetilde{Q}_{w,i}q^m)}\;
(q^{1/2}\mathsf{x})^{d},
\label{Zopen_m}
\end{align}
up to the framing ambiguity, where
$$
\widetilde{Q}_{w,0}=1,\qquad
\widetilde{Q}_{w,i}=q^{-1}\prod_{1\le j \le i} Q_{\lambda,j}Q_{w,j},\qquad
\widetilde{Q}_{\lambda,i}=Q_{\lambda,1}\prod_{1\le j \le i-1} Q_{w,j}Q_{\lambda,j+1},\qquad
q=\mathrm{e}^{-g_s}.
$$
Since the brane partition function on the strip geometry realizes the K-theoretic version of the vortex partition function, 
we adopt the following reparametrizations:
$$
g_s=\beta \hbar,\qquad
\mathsf{x}=x,\qquad
\widetilde{Q}_{w,i}=\mathrm{e}^{-\beta (w_{0}-w_{i})},\qquad
\widetilde{Q}_{\lambda,i}=\mathrm{e}^{-\beta (w_{0}-\lambda_{i}+\hbar)},
$$
and take the cohomological limit. In $\beta\to 0$ we find that the brane partition function (\ref{Zopen_m}) reduces to the vortex partition function (\ref{vortex_glsm})
for the GLSM on $\mathbb{S}^2$ with $n=N$:
\begin{align}
Z_{\textrm{A-brane}}^{Y_{N-1,N}}(\mathsf{x})\quad
\mathop{\longrightarrow}^{\beta\to 0}\quad
Z_{\textrm{vortex}}^{X_{\bm{l}=\bm{1};\bm{w},\{\lambda_1,\ldots,\lambda_N\}}}(x).
\nonumber
\end{align}
Thus we also find the realization of the vortex partition function $Z_{\textrm{vortex}}^{X_{\bm{l}=\bm{1};\bm{w},\{\lambda_1,\ldots,\lambda_N\}}}(x)$ from the open topological A-model on the strip geometry $Y_{N-1,N}$

\begin{figure}[t]
 \centering
  \includegraphics[width=90mm]{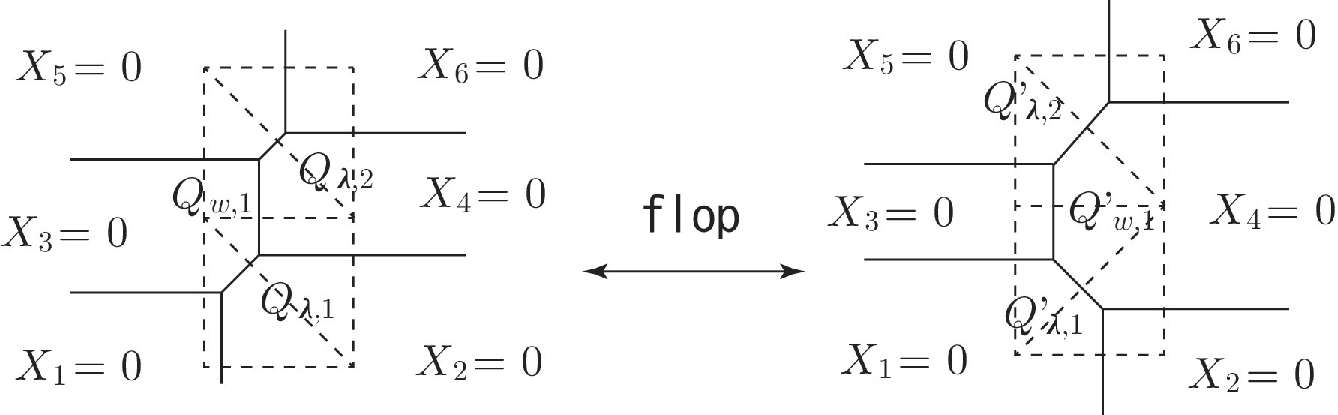}
 \caption{A flop of strip consisting of three $(-1,-1)$ curves gives a strip made of one $(-2,0)$ curve and two $(-1,-1)$ curves.}
\label{local_A_geom_flop}
\end{figure}
For the case $n<N$ we obtain the  the strip geometry $Y_{N-1,n}$ by taking 
the large volume limit of the base $\mathbb{C}\textbf{P}^1$ in $Y_{\bm{l}}$.
Such a strip geometry $Y_{N-1,n}$ is made of $N-1$ copies of $(-2,0)$ and $n$ copies of $(-1,-1)$ curves, 
and such local toric Calabi-Yau 3-fold is found by acting ``flops'' and ``decouplings'' repeatedly to the strip geometry $Y_{N-1,N}$.
To see the actions of  ``flops'' and ``decouplings'' manifestly,
we will focus on the $N=2$ case:
\begin{align}
\begin{split}
\bm{l}_{\lambda,1}&=(1,-1,-1,1,0,0),\\
\bm{l}_{w,1}&=(0,1,-1,-1,1,0),\\
\bm{l}_{\lambda,2}&=(0,0,1,-1,-1,1).
\label{flop_ex1}
\end{split}
\end{align}
The K\"ahler moduli parameters $Q_{\lambda,1}$, $Q_{w,1}$, and $Q_{\lambda,2}$ are associated with the charge vectors $\bm{l}_{\lambda,1}$, $\bm{l}_{w,1}$, and $\bm{l}_{\lambda,2}$, respectively. By the flop transitions described in Figure \ref{local_A_geom_flop}, we obtain a strip geometry made of one $(-2,0)$ curve and two $(-1,-1)$ curves, and it is given by the charge vectors
\begin{align}
\begin{split}
\bm{l}'_{\lambda,1}&=(-1,1,1,-1,0,0)=-\bm{l}_{\lambda,1},\\
\bm{l}'_{w,1}&=(1,0,-2,0,1,0)=\bm{l}_{\lambda,1}+\bm{l}_{w,1},\\
\bm{l}'_{\lambda,2}&=(0,0,1,-1,-1,1)=\bm{l}_{\lambda,2}.
\label{flop_ex2}
\end{split}
\end{align}
The K\"ahler moduli parameters $Q'_{\lambda,1}$, $Q'_{w,1}$, and $Q'_{\lambda,2}$ are associated with the charge vectors $\bm{l}'_{\lambda,1}$, $\bm{l}'_{w,1}$, and $\bm{l}'_{\lambda,2}$, respectively, and
they are related with $Q_{\lambda,1}$, $Q_{w,1}$, and $Q_{\lambda,2}$ by
\begin{align}
(Q'_{\lambda,1}, Q'_{w,1}, Q'_{\lambda,2})=(Q_{\lambda,1}^{-1}, Q_{\lambda,1}Q_{w,1}, Q_{\lambda,2}).
\label{flop_kahler}
\end{align}
After taking the decoupling limit $Q'_{\lambda,1}\to 0$ or $Q'_{\lambda,2}\to 0$, we obtain the strip geometry made of one $(-2,0)$ curve and one $(-1,-1)$ curve. 

Similarly for general $N$ one can find a strip geometry $Y_{N-1,n}$ with $n<N$. 
As a consequence, the brane partition function for a toric brane in $Y_{N-1,n}$ takes the form \cite{Iqbal:2004ne,Dimofte:2010tz}:
\begin{align}
Z_{\textrm{A-brane}}^{Y_{N-1,n}}(\mathsf{x})=\sum_{d=0}^{\infty}\frac{\prod_{a=1}^{n}\prod_{m=1}^{d}(1-\widetilde{Q}_{\lambda,a}q^{m-1})}{\prod_{i=0}^{N-1}\prod_{m=1}^{d}(1-\widetilde{Q}_{w,i}q^m)}\;
(q^{1/2}\mathsf{x})^{d},\qquad
\widetilde{Q}_{w,0}=1.
\label{Zopen_m_g}
\end{align}
up to the framing ambiguity and the normalization which is independent of the open string moduli.
After taking the cohomological limit $\beta \to 0$ under the reparametrizations:
$$
g_s=\beta \hbar,\qquad
\mathsf{x}=\beta^{N-n}x,\qquad
\widetilde{Q}_{w,i}=\mathrm{e}^{-\beta (w_{0}-w_{i})},\qquad
\widetilde{Q}_{\lambda,a}=\mathrm{e}^{-\beta (w_{0}-\lambda_{a}+\hbar)},
$$
the brane partition function (\ref{Zopen_m_g}) reduces to the vortex partition function $Z_{\textrm{vortex}}^{X_{\bm{l}=\bm{1};\bm{w},\bm{\lambda}}}$ 
in  (\ref{vortex_glsm}) \cite{Dimofte:2010tz,Bonelli:2011fq}:
$$
Z_{\textrm{A-brane}}^{Y_{N-1,n}}(\mathsf{x})\quad
\mathop{\longrightarrow}^{\beta\to 0}\quad
Z_{\textrm{vortex}}^{X_{\bm{l}=\bm{1};\bm{w},\bm{\lambda}}}(x).
$$
Combining this result with the vantage point 1 again for $n<N$, this shows a realization of the (``on-shell'') equivariant $J$-function for the Fano complete intersection of degree $l_{a=1,\ldots,n}=1$ hypersurfaces in $\mathbb{C}\textbf{P}^{N-1}$ as the brane partition function in the topological A-model on the strip geometry \cite{Dimofte:2010tz}.

\subsubsection{$q$-difference equation for the brane partition function}\label{rem:q_curve_brane}

The brane partition function 
obeys a $q$-difference equation known as the Schr\"odinger equation or 
the quantum curve\footnote{
This $q$-difference equation is defined simply as the annihilating equation of the brane partition function.
Via mirror symmetry (see discussions in Section \ref{subsec-b_J}), the brane partition function is regarded as the wave function 
\cite{Aganagic:2003qj}, and in this sense, we can identify this $q$-difference equation as the quantum curve. (See Remark \ref{rem:def-of-omega-x-2}.)
} \cite{Dijkgraaf:2007sw,Dijkgraaf:2008fh}. For the brane partition function $Z_{\textrm{A-brane}}^{Y_{N-1,n}}(\mathsf{x})$ in  (\ref{Zopen_m_g}), one finds the $q$-difference equation (c.f. Remark \ref{rem:def-of-omega-x-2})
\begin{align}
\widehat{A}_{Y}^K(\widehat{\mathsf{x}},\widehat{\mathsf{y}})Z_{\textrm{A-brane}}^{Y_{N-1,n}}(\mathsf{x})
=\left[\prod_{i=0}^{N-1}\left(1-\widetilde{Q}_{w,i}\,\widehat{\mathsf{y}}\right)
-\widehat{\mathsf{x}}\prod_{a=1}^{n}\left(1-\widetilde{Q}_{\lambda,a}\,\widehat{\mathsf{y}}\right)\right]
Z_{\textrm{A-brane}}^{Y_{N-1,n}}(\mathsf{x})=0,
\label{top_k_q_curve}
\end{align}
where the operators $\widehat{\mathsf{x}}$ and $\widehat{\mathsf{y}}$ obey the relation $\widehat{\mathsf{y}}\,\widehat{\mathsf{x}}=q\,\widehat{\mathsf{x}}\,\widehat{\mathsf{y}}$, and act on the brane partition function as
$$
\widehat{\mathsf{x}}\,Z_{\textrm{A-brane}}^{Y_{N-1,n}}(\mathsf{x})=\mathsf{x}\,Z_{\textrm{A-brane}}^{Y_{N-1,n}}(\mathsf{x}),\qquad
\widehat{\mathsf{y}}\,Z_{\textrm{A-brane}}^{Y_{N-1,n}}(\mathsf{x})=Z_{\textrm{A-brane}}^{Y_{N-1,n}}(q\mathsf{x}).
$$

Let us consider the cohomological limit of this $q$-difference equation.
For this purpose we will replace the operators $\widehat{\mathsf{x}}$ and $\widehat{\mathsf{y}}$ 
by $\widehat{\mathsf{x}}=\beta^{N-n}x$ and $\widehat{\mathsf{y}}=\mathrm{e}^{-\beta\hbar x\frac{d}{dx}}$ respectively,
and reparametrize $\widetilde{Q}_{w,i}$ and $\widetilde{Q}_{\lambda,a}$ as $\widetilde{Q}_{w,i}=\mathrm{e}^{\beta \widetilde{w}_{i}}$ and $\widetilde{Q}_{\lambda,a}=\mathrm{e}^{\beta (\widetilde{\lambda}_{a}-\hbar)}$.
Then taking the cohomological limit $\beta\to 0$ for the quantum curve (\ref{top_k_q_curve}),
we obtain a differential equation:
\begin{align}
\left[
\left(\hbar x\frac{d}{dx}\right)\prod_{i=1}^{N-1}\left(\hbar x\frac{d}{dx}-\widetilde{w}_i\right)-
x\prod_{a=1}^{n}\left(\hbar x\frac{d}{dx}-\widetilde{\lambda}_a+\hbar\right)
\right]Z_{\textrm{A-brane}}^{\mathrm{coh}}(x)=0.
\label{GKZ_brane}
\end{align}
Using this differential equation, we find that 
$\mathrm{e}^{w_0/\hbar}Z_{\textrm{A-brane}}^{\mathrm{coh}}(x)$
obeys the GKZ equation (\ref{GKZ_comp}) for the equivariant $J$-function.

The classical limit of this $q$-difference equation is found from the WKB expansion of the brane partition function:
$$
Z_{\textrm{A-brane}}^{Y_{N-1,n}}(\mathsf{x}) \sim 
\exp \left(\sum_{m=0}^{\infty}g_s^{m-1}F_m^{K}(\mathsf{x})\right).
$$
Denoting $\log\mathsf{y}=\mathsf{x}\,\partial_{\mathsf{x}}F_0^{K}(\mathsf{x})$, we find that the $q$-difference equation (\ref{top_k_q_curve}) for the brane partition function
reduces to the defining equation of a classical curve  $\Sigma_{Y_{N-1,n}}^{K}$:
\begin{align}
\Sigma_{Y_{N-1,n}}^{K}=
\Big\{\; (\mathsf{x},\mathsf{y})\in {\IC}^*\times {\IC}^*\; \Big|\; 
A_{Y}^K(\mathsf{x},\mathsf{y})=
\prod_{i=0}^{N-1}\left(1-\widetilde{Q}_{w,i}\,\mathsf{y}\right)
-\mathsf{x}\prod_{a=1}^{n}\left(1-\widetilde{Q}_{\lambda,a}\,\mathsf{y}\right)=0
\; \Big\},
\label{top_k_c_curve}
\end{align}
under the classical limit $g_s \to 0$.
In the next subsection we will see that this classical curve agrees with the mirror curve in the mirror B-model picture. 

\subsection{Vantage point 3: $J$-function as the brane partition function in the local B-model}\label{subsec-b_J}

Via the local mirror symmetry, we will study the brane partition function $Z_{\textrm{A-brane}}^{Y}(x)$ in the topological A-model 
in terms of the topological B-model on the mirror local Calabi-Yau 3-fold $Y_{\bm{l}}^{\vee}$.
The mirror local Calabi-Yau 3-fold $Y_{\bm{l}}^{\vee}$ corresponding to the local toric Calabi-Yau 3-fold $Y_{\bm{l}}$ of  (\ref{local_a_geom}) 
is defined with a substitution by $|x_{\alpha}|=\mathrm{e}^{|X_{\alpha}|^2}$  \cite{Hori:2000kt,Hori:2000ck}:
\begin{align}
Y_{\bm{l}}^{\vee}=
\Big\{\;(\omega_+,\omega_-,x_1,\ldots,x_{m+3})\in {\IC}^2\times ({\IC}^*)^{m+3}\;
\Big|\;\omega_+\omega_-=\sum_{\alpha=1}^{m+3}x_{\alpha},\ \prod_{\alpha=1}^{m+3}x_{\alpha}^{l_{i,\alpha}}=z_i\;
\Big\},
\label{local_b_geom}
\end{align}
where $z_i$'s parametrize the complex moduli space. 
We eliminate local coordinates $X_{\delta}$ ($\delta=1,\ldots,m+3$) in the defining equation of (\ref{local_b_geom}) except for $\delta=\alpha,\beta,\gamma$ and fix $X_{\gamma}=0$.
By choosing local coordinates in such a way, the Lagrangian submanifold (\ref{lag_a_geom}) defined on the local atlas which covers $X_{\alpha}=X_{\beta}=X_{\gamma}=0$ can be described well.
As a consequence, the defining equation (\ref{local_b_geom}) of the mirror Calabi-Yau 3-fold $Y_{\bm{l}}^{\vee}$
is rewritten as the hypersurface in ${\IC}^2\times ({\IC}^*)^2$:
\begin{align}
Y_{\bm{l}}^{\vee}=
\Big\{\;(\omega_+,\omega_-,x,y)\in {\IC}^2\times ({\IC}^*)^2\;
\Big|\;\omega_+\omega_-=A_{Y_{\bm{l}}^{\vee}}^K(x,y)\;
\Big\},
\nonumber
\end{align}
where $x=x_{\alpha}$ and $y=x_{\beta}$, and the open string modulus $\mathsf{x}$ of the toric brane is mapped to $x$. 
\textit{Mirror curve} $\Sigma_{Y_{\bm{l}}^{\vee}}^{K}$ 
is defined as the complex 1 dimensional submanifold 
which resides in $Y_{\bm{l}}^{\vee}$:
\begin{align}
\Sigma_{Y_{\bm{l}}^{\vee}}^{K}=
\Big\{\;(x,y)\in ({\IC}^*)^2\;
\Big|\;A_{Y_{\bm{l}}^{\vee}}^K(x,y)=0\;
\Big\} \subset Y_{\bm{l}}^{\vee}.
\label{local_b_curve}
\end{align}

\begin{remark}
For the mirror curve $\Sigma_{Y_{\bm{l}}^{\vee}}^{K}$ one can consider degrees of freedom for the framing $f\in{\IZ}$ of the brane mentioned in Remark \ref{rem:framing}
by an $SL(2,{\IZ})$ transformation $x \to x y^f$ and $y \to y$, which preserves the symplectic form $d\omega=d \log x \wedge d \log y$ on $\Sigma_{Y_{\bm{l}}^{\vee}}^{K}$ \cite{Bouchard:2007ys}.
\end{remark}

\begin{defi}[Mirror map  \cite{Chiang:1999tz,Lerche:2001cw}]
The mirror map between $Q_i$'s in  (\ref{local_a_geom}) and $z_i$'s in  (\ref{local_b_geom}) is given by the logarithmic solutions to the Picard-Fuchs equations $\mathcal{D}_if(\bm{z})=0$ for periods of the holomorphic 3-form on $Y_{\bm{l}}^{\vee}$. 
Here
\begin{align}
\mathcal{D}_i=\prod_{l_{i,\alpha}>0}\left(\frac{\partial}{\partial x_{\alpha}}\right)^{l_{i,\alpha}} 
-\prod_{l_{i,\alpha}<0}\left(\frac{\partial}{\partial x_{\alpha}}\right)^{-l_{i,\alpha}},
\nonumber
\end{align}
and local coordinates $x_{\alpha}$'s are related with $z_i$'s by $\prod_{\alpha=1}^{m+3}x_{\alpha}^{l_{i,\alpha}}=z_i$ in  (\ref{local_b_geom}). Explicitly  the inverse mirror map between $Q_i$'s and $z_i$'s (i.e. the logarithmic solution of the Picard-Fuchs equations)
is given by
\begin{align}
\log Q_i = \log z_i - 
\sum_{\begin{subarray}{c}\bm{n} \in {\IZ}^m_{\ge 0} \\
\bm{n}\ne (0,\ldots,0)\end{subarray}}
\sum_{\gamma, (\mathfrak{m}_{\gamma}>0)}l_{i,\gamma}
(-1)^{\mathfrak{m}_{\gamma}}\;
\frac{(\mathfrak{m}_{\gamma}-1)!}{\prod_{\alpha\ne \gamma}\big(\sum_{j} l_{j,\alpha}n_j\big)!}\;
z_1^{n_1}\cdots z_m^{n_m},
\label{cl_mirror}
\end{align}
where $\mathfrak{m}_{\gamma} = -\sum_j l_{j,\gamma}n_j$. 

The mirror map between open string moduli $\mathsf{x}$ in  (\ref{lag_a_geom}) and $x$ in  (\ref{local_b_curve}) is also obtained by extending the charge vectors $\bm{l}_i$ to $(\bm{l}_i;0,0)$ and adding one more charge $\bm{l}_0=(\underbrace{\ldots};1,-1)$, where the underbrace means $l_{0,\alpha}=1$, $l_{0,\gamma}=-1$, and $l_{0,\beta}=0$ for $\beta\ne \alpha, \gamma$. Explicitly the inverse mirror map between $\mathsf{x}$ and $x$ is given by
\begin{align}
\log \mathsf{x} = \log x - 
\sum_{\begin{subarray}{c}\bm{n} \in {\IZ}^m_{\ge 0} \\
\bm{n}\ne (0,\ldots,0)\end{subarray}}
\sum_{\gamma, (\mathfrak{m}_{\gamma}>0)}l_{0,\gamma}
(-1)^{\mathfrak{m}_{\gamma}}\;
\frac{(\mathfrak{m}_{\gamma}-1)!}{\prod_{\alpha\ne \gamma}\big(\sum_{j} l_{j,\alpha}n_j\big)!}\;
z_1^{n_1}\cdots z_m^{n_m}.
\label{op_mirror}
\end{align}
\end{defi}

\subsubsection{Geometric engineering of the (``on-shell'') equivariant $J$-function for $\mathbb{C}\textbf{P}^{N-1}_{\bm{w}}$}

Consider the local toric Calabi-Yau 3-fold defined by charge vectors (\ref{02_chain}), and a Lagrangian submanifold (\ref{lag_a_geom}) with $\alpha=1, \beta=3, \gamma=2$ as depicted in Figure \ref{local_A_geom} for example. From the vantage point 2 this brane partition function gives the (``on-shell'') equivariant $J$-function for $\mathbb{C}\textbf{P}^{N-1}_{\bm{w}}$. For the mirror Calabi-Yau 3-fold (\ref{local_b_geom}), by taking local coordinate $x_1=x$, $x_2=1$, $x_3=y$ respecting the Lagrangian submanifold $L$ in the A-model, 
we find a defining equation of the mirror curve (\ref{local_b_curve}):
\begin{align}
A_{Y_{\bm{l}}^{\vee}}^K(x,y)=\sum_{i=1}^{N-1}\widetilde{z}_{i}y^{i+1}+y+x+1=0,\qquad
\widetilde{z}_{i}= \prod_{1\le j \le i}z_j^{i-j+1}.
\label{m_curve_ex_pure}
\end{align}
Adopting the mirror maps  (\ref{cl_mirror}) and  (\ref{op_mirror}) to this geometry, we see that for the local coordinate $x_1=x$, $x_2=1$, $x_3=y$ the quantum corrections are absent for $x$, i.e. $x=\mathsf{x}$. On the other hand there are quantum corrections for $z_i$'s such that
\begin{align}
Q_i=\mathrm{e}^{-g_{i-1}+2g_{i}-g_{i+1}}z_i.
\nonumber
\end{align}
Here $g_{0}=g_{N}= 0$, and $g_i= g_i(\bm{z})$ for $i=1,2,\ldots,N-1$ are defined by
\begin{align}
g_i(\bm{z})=
\sum_{\begin{subarray}{c}\bm{n} \in {\IZ}^{N-1}_{\ge 0} \\
\bm{n}\ne (0,\ldots,0)\end{subarray}}
(-1)^{n_{i-1}+n_{i+1}}\;
\frac{(-n_{i-1}+2n_i-n_{i+1}-1)!(n_{i-1}-2n_i+n_{i+1})!}{\prod_{j=0}^{N}(n_{j-1}-2n_j+n_{j+1})!}\;
z_1^{n_1}\cdots z_{N-1}^{n_{N-1}},
\nonumber
\end{align}
where $n_{-1}=n_0=n_N=n_{N+1}= 0$. For functions $f_i(\bm{Q})$ of $Q_i$ such that
\begin{align}
\log f_i(\bm{Q})=g_i\big(\bm{z}(\bm{Q})\big),\qquad
f_{0}(\bm{Q})=f_{N}(\bm{Q})= 1,
\nonumber
\end{align}
we obtain the inverse of the mirror map (\ref{cl_mirror}):
\begin{align}
z_i=\frac{f_{i-1}(\bm{Q})f_{i+1}(\bm{Q})}{f_{i}(\bm{Q})^2}\;Q_i.
\label{inv_m_map_p}
\end{align}
We find that such functions $f_a(\bm{Q})$ are given by\footnote{We have directly checked this up to some orders.}
\begin{align}
f_{i}(\bm{Q})=\frac{1}{\prod_{j=1}^{i-1}\widetilde{Q}_j}
\sum_{0 \le t_1< \ldots < t_i \le N-1}\widetilde{Q}_{t_1}\cdots \widetilde{Q}_{t_i},\qquad
\widetilde{Q}_0 = 1,\quad
\widetilde{Q}_i = \prod_{1\le j \le i} Q_j.
\nonumber
\end{align}

\begin{exam}
In the case of $N=4$ in Figure \ref{local_A_geom} we obtain
\begin{align}
\begin{split}
z_1&=\frac{Q_1(1+Q_2+Q_1Q_2+Q_2Q_3+Q_1Q_2Q_3+Q_1Q_2^2Q_3)}{(1+Q_1+Q_1Q_2+Q_1Q_2Q_3)^2},\\
z_2&=\frac{Q_2(1+Q_1+Q_1Q_2+Q_1Q_2Q_3)(1+Q_3+Q_2Q_3+Q_1Q_2Q_3)}{(1+Q_2+Q_1Q_2+Q_2Q_3+Q_1Q_2Q_3+Q_1Q_2^2Q_3)^2},\\
z_3&=\frac{Q_3(1+Q_2+Q_1Q_2+Q_2Q_3+Q_1Q_2Q_3+Q_1Q_2^2Q_3)}{(1+Q_3+Q_2Q_3+Q_1Q_2Q_3)^2}.
\nonumber
\end{split}
\end{align}
\end{exam}

By the mirror map (\ref{inv_m_map_p}) the defining equation of the mirror curve (\ref{m_curve_ex_pure}) yields
\begin{align}
A_{Y_{\bm{l}}^{\vee}}^K(x,y)=
\prod_{i=0}^{N-1}\left(1+f_1(\bm{Q})^{-1}\widetilde{Q}_i y\right)+x=0.
\label{m_curve_ex_pm}
\end{align}
After a change of variables:
\begin{align}
y\ \to \ -f_1(\bm{Q})y,\qquad
x\ \to \ -x,
\nonumber
\end{align}
we find that the mirror curve (\ref{m_curve_ex_pm}) agrees with the classical curve (\ref{top_k_c_curve}) for $Z_{\textrm{A-brane}}^{Y_{N-1}}(\mathsf{x})$ in the A-model.

\subsubsection{Geometric engineering of the (``on-shell'') equivariant  $J$-functions for degree $1$ complete intersections in $\mathbb{C}\textbf{P}^{N-1}_{\bm{w}}$}

Next we will consider the local toric Calabi-Yau 3-fold defined by charge vectors (\ref{11_chain}), and a Lagrangian submanifold (\ref{lag_a_geom}) with $\alpha=2, \beta=3, \gamma=1$ as depicted in Figure \ref{local_A_geom_m}. For the mirror Calabi-Yau 3-fold (\ref{local_b_geom}), by taking local coordinate $x_1=1$, $x_2=x$, $x_3=y$ we find a mirror curve (\ref{local_b_curve}) 
that describes this brane
\begin{align}
A_{Y_{\bm{l}}^{\vee}}^K(x,y)=\sum_{i=1}^{N-1}\widetilde{z}_{\lambda,i}\widetilde{z}_{w,i}y^{i+1}+y+x+1
+\sum_{i=1}^{N}\widetilde{z}_{\lambda,i}\widetilde{z}_{w,i-1}xy^{i}=0,
\label{m_curve_ex_bu}
\end{align}
where
$$
\widetilde{z}_{\lambda,i=1,\ldots,N}= \prod_{1\le j \le i}z_{\lambda,j}^{i-j+1},\qquad
\widetilde{z}_{w,0}= 0, \qquad
\widetilde{z}_{w,i=1,\ldots,N-1}= \prod_{1\le j \le i}z_{w,j}^{i-j+1}.
$$
For this mirror curve the open string modulus $x$ receives no quantum corrections, namely $x=\mathsf{x}$, and the mirror map (\ref{cl_mirror}) is given by
\begin{align}
Q_{\lambda,i}=\mathrm{e}^{-g_{w,i-1}+g_{\lambda,i-1}+g_{w,i}-g_{\lambda,i}}z_{\lambda,i},\qquad
Q_{w,i}=\mathrm{e}^{-g_{\lambda,i-1}+g_{w,i}+g_{\lambda,i}-g_{w,i+1}}z_{w,i}.
\nonumber
\end{align}
Here $g_{w,0}=g_{\lambda,0}=g_{w,N}=g_{\lambda,N}= 0$, and $g_{w,i}= g_{w,i}(\bm{z})$, $g_{\lambda,i}= g_{\lambda,i}(\bm{z})$ for $i=1,2,\ldots,N-1$ are defined by
\begin{align}
\begin{split}
g_{w,i}(\bm{z})&=
\sum_{\begin{subarray}{c}\bm{n} \in {\IZ}^{2N-1}_{\ge 0} \\
\bm{n}\ne (0,\ldots,0)\end{subarray}}
(-1)^{n_{w,i-1}+n_{w,i}+n_{\lambda,i}+n_{\lambda,i+1}}\;
z_{\lambda,1}^{n_{\lambda,1}}z_{w,1}^{n_{w,1}}\cdots z_{w,N-1}^{n_{w,N-1}}z_{\lambda,N}^{n_{\lambda,N}}
\nonumber\\
&\hspace{4em}\times
\frac{(-n_{w,i-1}+n_{\lambda,i}+n_{w,i}-n_{\lambda,i+1}-1)!(n_{w,i-1}-n_{\lambda,i}-n_{w,i}+n_{\lambda,i+1})!}{\prod_{j=0}^{N}(n_{w,j-1}-n_{\lambda,j}-n_{w,j}+n_{\lambda,j+1})!(n_{\lambda,j}-n_{w,j}-n_{\lambda,j+1}+n_{w,j+1})!},
\\
g_{\lambda,i}(\bm{z})&=
\sum_{\begin{subarray}{c}\bm{n} \in {\IZ}^{2N-1}_{\ge 0} \\
\bm{n}\ne (0,\ldots,0)\end{subarray}}
(-1)^{n_{\lambda,i}+n_{\lambda,i+1}+n_{w,i}+n_{w,i+1}}\;
z_{\lambda,1}^{n_{\lambda,1}}z_{w,1}^{n_{w,1}}\cdots z_{w,N-1}^{n_{w,N-1}}z_{\lambda,N}^{n_{\lambda,N}}
\nonumber\\
&\hspace{4em}\times
\frac{(-n_{\lambda,i}+n_{w,i}+n_{\lambda,i+1}-n_{w,i+1}-1)!(n_{\lambda,i}-n_{w,i}-n_{\lambda,i+1}+n_{w,i+1})!}{\prod_{j=0}^{N}(n_{w,j-1}-n_{\lambda,j}-n_{w,j}+n_{\lambda,j+1})!(n_{\lambda,j}-n_{w,j}-n_{\lambda,j+1}+n_{w,j+1})!},
\nonumber
\end{split}
\end{align}
where $n_{w,-1}=n_{\lambda,0}=n_{w,0}=n_{w,N}=n_{\lambda,N+1}=n_{w,N+1}= 0$.
For functions $f_{w,i}(\bm{Q})$ and $f_{\lambda,i}(\bm{Q})$ of $Q_{\lambda,i}$ and $Q_{w,i}$ such that
\begin{align}
\begin{split}
&
\log f_{w,i}(\bm{Q})=g_{w,i}\big(\bm{z}(\bm{Q})\big),\qquad
\log f_{\lambda,i}(\bm{Q})=g_{\lambda,i}\big(\bm{z}(\bm{Q})\big),\\
&
f_{w,0}(\bm{Q})=f_{\lambda,0}(\bm{Q})=f_{w,N}(\bm{Q})=f_{\lambda,N}(\bm{Q})= 1,
\nonumber
\end{split}
\end{align}
we obtain the inverse of the mirror map (\ref{cl_mirror}):
\begin{align}
z_{\lambda,i}=\frac{f_{w,i-1}(\bm{Q})f_{\lambda,i}(\bm{Q})}{f_{\lambda,i-1}(\bm{Q})f_{w,i}(\bm{Q})}\;Q_{\lambda,i},\qquad
z_{w,i}=\frac{f_{\lambda,i-1}(\bm{Q})f_{w,i+1}(\bm{Q})}{f_{w,i}(\bm{Q})f_{\lambda,i}(\bm{Q})}\;Q_{w,i}.
\label{inv_m_map_up}
\end{align}
We find that such functions $f_{w,i}(\bm{Q})$ and $f_{\lambda,i}(\bm{Q})$ are given by\footnote{We have directly checked this up to some orders.}
\begin{align}
\begin{split}
&
f_{w,i}(\bm{Q})=\frac{1}{\prod_{j=1}^{i-1}\widetilde{Q}_{w,j}}
\sum_{0 \le t_1< \ldots < t_i \le N-1}\widetilde{Q}_{w,t_1}\cdots \widetilde{Q}_{w,t_i},\qquad
\widetilde{Q}_{w,0} = 1,\ \
\widetilde{Q}_{w,i} = \prod_{1\le j \le i} Q_{\lambda,j}Q_{w,j},
\\
&
f_{\lambda,i}(\bm{Q})=\frac{1}{\prod_{j=1}^{i}\widetilde{Q}_{\lambda,j}}
\sum_{1 \le t_1< \ldots < t_i \le N}\widetilde{Q}_{\lambda,t_1}\cdots \widetilde{Q}_{\lambda,t_i},\qquad
\widetilde{Q}_{\lambda,i} = Q_{\lambda,1}\prod_{1\le j \le i-1} Q_{w,j}Q_{\lambda,j+1}.
\nonumber
\end{split}
\end{align}
By the mirror map (\ref{inv_m_map_up}) the mirror curve (\ref{m_curve_ex_bu}) yields
\begin{align}
A_{Y_{\bm{l}}^{\vee}}^K(x,y)=
\prod_{i=0}^{N-1}\left(1+f_{w,1}(\bm{Q})^{-1}\widetilde{Q}_{w,i} y\right)
+x\prod_{i=1}^{N}\left(1+f_{w,1}(\bm{Q})^{-1}\widetilde{Q}_{\lambda,i} y\right)=0.
\label{m_curve_ex_bum}
\end{align}
After a change of variables:
\begin{align}
y\ \to \ -f_{w,1}(\bm{Q})y,\qquad
x\ \to \ -x,
\nonumber
\end{align}
we find that the mirror curve (\ref{m_curve_ex_bum}) agrees with the classical curve (\ref{top_k_c_curve}) for $Z_{\textrm{A-brane}}^{Y_{N-1,N}}(\mathsf{x})$ in the A-model.

To obtain the defining equation of the mirror curve for $Y_{N-1,n}$ with $n<N$ we will act the flop transitions which change the toric charges from (\ref{flop_ex1}) to  (\ref{flop_ex2}).
 In this case, by the local mirror symmetry, the complex structure moduli parameters $z_{\lambda,1}$, $z_{w,1}$, and $z_{\lambda,2}$ (resp. $z'_{\lambda,1}$, $z'_{w,1}$, and $z'_{\lambda,2}$) are associated with the charge vectors $\bm{l}_{\lambda,1}$, $\bm{l}_{w,1}$, and $\bm{l}_{\lambda,2}$ (resp. $\bm{l}'_{\lambda,1}$, $\bm{l}'_{w,1}$, and $\bm{l}'_{\lambda,2}$), respectively. 
Under the flop transitions,
the same relation as  (\ref{flop_kahler}) for the K\"ahler moduli parameters
holds for the complex structure moduli parameters:
\begin{align}
(z'_{\lambda,1}, z'_{w,1}, z'_{\lambda,2})=(z_{\lambda,1}^{-1}, z_{\lambda,1}z_{w,1}, z_{\lambda,2}).
\label{flop_cpx}
\end{align}
Combining this relation (\ref{flop_cpx}) with the relation (\ref{flop_kahler}) together, the inverse mirror map can be considered after the flop transitions. 
As a consequence of ``flops'' and ``decouplings'' for some of complex structure moduli parameters in the mirror curve (\ref{m_curve_ex_bum}),
we obtain the classical curve $\Sigma_{Y_{N-1,n}}^{K}$ in  (\ref{top_k_c_curve}).

\subsubsection{Wave function for the mirror curve and the remodeling conjecture}

Regarding the mirror curve (\ref{local_b_curve}) as the spectral curve,
one can find the wave function $\psi_{Y_{\bm{l}}^{\vee}}^{K}(x)$ for this curve via the topological recursion.
The parameters $z_i$'s and $x$ are mapped to $Q_i$'s and $\mathsf{x}$ by the inverse mirror maps (\ref{cl_mirror}) and (\ref{op_mirror}).
As a pullback of the wave function $\psi_{Y_{\bm{l}}^{\vee}}^{K}(x)$ by the inverse mirror map, we define the wave function $\psi_{Y_{\bm{l}}}^{K}(\mathsf{x})$. 
On the other hand the brane partition function $Z_{\textrm{B-brane}}^{Y_{\bm{l}}^{\vee}}(x)$
in the local B-model is defined as the pullback of $Z_{\textrm{A-brane}}^{Y_{\bm{l}}}(\mathsf{x})$ by the mirror map.
The relation between these brane partition functions $Z_{\textrm{A-brane}}^{Y_{\bm{l}}}(\mathsf{x})$ and $Z_{\textrm{B-brane}}^{Y^{\vee}_{\bm{l}}}(x)$ and wave functions $\psi_{Y_{\bm{l}}^{\vee}}^{K}(x)$ and $\psi_{Y_{\bm{l}}}^{K}(\mathsf{x})$ is found from
the remodeling conjecture proposed by V.~Bouchard, A.~Klemm, M.~Mari\~no, and S.~Pasquetti \cite{Marino:2006hs,Bouchard:2007ys,Bouchard:2008gu} (see \cite{Zhou:2009gh,Eynard:2012nj,Fang:2013dna,Fang:2016svw} for proofs and generalizations).

\begin{conj}[Remodeling conjecture]
Consider the topological A-model on a local toric Calabi-Yau 3-fold $Y_{\bm{l}}$ in (\ref{local_a_geom}) with a special Lagrangian submanifold $L$ in  (\ref{lag_a_geom}). 
Let $F_n^{(g)}(\mathsf{x}_1,\ldots,\mathsf{x}_n)$ be the generating function  of the open Gromov-Witten invariants that enumerate the world sheet instantons for the map from the genus $g$ Riemann surface with $n$ boundaries $\Sigma_{g,n}$ (resp. the boundaries $\partial\Sigma_{g,n}$ of $\Sigma_{g,n}$) to $Y_{\bm{l}}$ (resp. $L$).
Via the mirror maps (\ref{cl_mirror}) and (\ref{op_mirror}),
the generating function $F_n^{(g)}(\mathsf{x}_1,\ldots,\mathsf{x}_n)$ is given by
\begin{align}
\begin{split}
&
F_1^{(0)}(\mathsf{x}_1)=\int_{z_1^*}^{z_1}\omega (x(z_1')),
\qquad 
\omega (x(z)) = \log y(x(z))\, \frac{dx(z)}{x(z)},
\\
&
F_2^{(0)}(\mathsf{x}_1, \mathsf{x}_2)=\int_{z_1^*}^{z_1}\int_{z_2^*}^{z_2}
\left(B(z_1',z_2')-\frac{dx(z_1')dx(z_2')}{(x(z_1')-x(z_2'))^2}\right),
\\
&
F_n^{(g)}(\mathsf{x}_1,\ldots,\mathsf{x}_n)=\int_{z_1^*}^{z_1}\cdots\int_{z_n^*}^{z_n}\omega_n^{(g)}(z_1',\ldots,z_n'),\ \ \textrm{for $(g,n)\ne (0,1), (0,2)$},
\label{Fgn0}
\end{split}
\end{align}
up to the framing ambiguity.
Here $B(z_1,z_2)$ is the Bergman kernel and $\omega_n^{(g)}$ ($(g,n)\ne (0,1), (0,2)$) are the multilinear meromorphic differentials recursively defined by the topological recursion (\ref{top_recursion}) or (\ref{g_top_recursion}) on a mirror curve $\Sigma_{Y_{\bm{l}}^{\vee}}^{K}$.
$z_i$'s denote points on the mirror curve $\Sigma_{Y_{\bm{l}}^{\vee}}^{K}$ in a local coordinate, and $z_i^*$'s denote reference points in $\Sigma_{Y_{\bm{l}}^{\vee}}^{K}$ so that the integrals converge to $0$ at these points.
\end{conj}

Following the WKB reconstruction (\ref{wave_function}),  we can define the wave function by (\ref{Fgn0}) such that
\begin{align}
\psi_{Y_{\bm{l}}}^{K}(\mathsf{x})=\exp\left(
\sum_{g=0,n=1}^{\infty}\frac{1}{n!}\hbar^{2g-2+n}F_n^{(g)}(\mathsf{x},\ldots,\mathsf{x})
\right).
\nonumber
\end{align}
From the remodeling conjecture we see that the WKB reconstruction of this wave function $\psi_{Y_{\bm{l}}}^{K}(\mathsf{x})$ agrees with the WKB expansion of brane partition function $Z_{\textrm{A-brane}}^{Y_{\bm{l}}}(\mathsf{x})$ in the local A-model. 
In particular for the spectral curve $\Sigma_{Y_{N-1,n}}^{K}$ in  (\ref{top_k_c_curve}) the wave function $\psi_{Y_{N-1,n}}^{K}(\mathsf{x})$ is defined in this way, and it gives the WKB reconstruction of the brane partition function $Z_{\textrm{A-brane}}^{Y_{N-1,n}}(\mathsf{x})$ in  (\ref{Zopen_m_g}):
\begin{align}
Z_{\textrm{A-brane}}^{Y_{N-1,n}}(\mathsf{x}) \sim
\psi_{Y_{N-1,n}}^{K}(\mathsf{x}).
\label{brane_wave}
\end{align}
From the vantage points 1 and 2, for $n<N$, the cohomological limit $\beta \to 0$ of $\psi_{Y}^{K}(\mathsf{x})$ gives a WKB reconstruction of the (``on-shell'')  equivariant $J$-function for the Fano complete intersection of the degree $l_{a=1,\ldots,n}=1$ hypersurfaces in $\mathbb{C}\textbf{P}^{N-1}$ via the topological recursion. 

\bigskip

In summary, from these 3 vantage points, we have found results as follows.
\begin{itemize}
\item $Z_{\textrm{A-brane}}^{Y_{N-1,n}}(\mathsf{x})$ obeys the differential equation (\ref{GKZ_brane})
which comes from the $q$-difference equation (\ref{top_k_q_curve}) in the cohomological limit.
This differential equation agrees with the GKZ equation (\ref{GKZ_comp})
compensated by a factor $x^{w_0/\hbar}$. (See Lemma \ref{lemm:GKZ_calJ}.)
\item The classical limit of the $q$-difference equation (\ref{top_k_q_curve}) defines the classical curve (\ref{top_k_c_curve}), and it agrees with the mirror curve $\Sigma_{Y_{N-1,n}}^K$ for $Y_{N-1,n}$.
\item The wave function $\psi_{Y_{N-1,n}}^{K}(\mathsf{x})$ found from the topological recursion for the mirror curve $\Sigma_{Y_{N-1,n}}^K$ is regarded as the 
WKB expansion of the brane partition function $Z_{\textrm{A-brane}}^{Y_{N-1,n}}(\mathsf{x})$.
\end{itemize}
At the level of the cohomological limit, the above results suggest that 
the GKZ equation (\ref{GKZ_comp}) is regarded as a quantum curve for the GKZ curve (\ref{A_comp}),
because the differential equation (\ref{GKZ_brane}) for the cohomological limit of the brane partition function 
satisfies the properties in Definition \ref{def:q_curve}.
Although some physical (but mathematical obscure) definitions and conjectures are used to obtain the above results, 
we find the physical derivation of the reconstruction theorem at last.

\section{Stokes matrix for ${\mathbb C}{\bf P}^1_{\bm w}$}
\label{section:Stokes-eqP1}

In this section we consider the quantum curve 
\begin{equation} \label{eq:equiv-P1}
\left[ \left( \hbar x \dfrac{d}{dx} - w_0 \right) 
\left( \hbar x \dfrac{d}{dx} - w_1 \right) - x \right] \psi = 0
\end{equation}
arising from the GKZ curve 
\begin{equation}
\Sigma_{{\mathbb C}{\bf P}^1_{\bm w}} = 
\{ (x,y) \in {\mathbb C}^{\ast}\times{\mathbb C} ~|~ 
A_{{\mathbb C}{\bf P}^1_{\bm w}}(x,y) = 0 \}, \quad 
A_{{\mathbb C}{\bf P}^1_{\bm w}}(x,y) = (y-w_0)(y-w_1) - x
\end{equation}
as is discussed in Section \ref{subsec:Mulase_Sulkowski}. 
This equation is also known as the quantum differential equation
(Dubrovin's first structure connection) for the equivariant 
Gromov-Witten theory of ${{\mathbb C}{\bf P}^1}$.
The goal of this section is to compute the Stokes matrix 
for the WKB solution of the equation \eqref{eq:equiv-P1} 
using the exact WKB method 
(see \cite{KT05,IN14} for the foundation of the exact WKB method).
As we will see below, integrals over the GKZ curve play
a crucially important role in the description of the Stokes matrices.

\subsection{Normalization of the WKB solution}
From the view point of the WKB method, it is convenient to 
transform  \eqref{eq:equiv-P1} to the following Schr{\"o}dinger-type equation:
\begin{equation} \label{eq:equiv-P1-Sch}
\left( \hbar^2 \frac{d^2}{dx^2} - Q \right) \varphi = 0, \quad
Q = Q_0(x) + \hbar^2 Q_2(x) =
\frac{4x+(w_0-w_1)^2}{4x^2} - \hbar^2 \frac{1}{4x^2}
\end{equation}
through the gauge transform 
\begin{equation} \label{eq:gauge-transform}
\psi = \exp\left( \frac{w_0 + w_1 - \hbar}{2\hbar} \, \log x \right) \varphi.
\end{equation}
The equation \eqref{eq:equiv-P1-Sch} has a unique turning point 
(i.e., the zero of the leading term of $Q$) at 
\[
v=-\frac{(w_0-w_1)^2}{4}.
\] 
In what follows we assume 
\begin{equation}
w_0 - w_1 \ne 0
\end{equation}
to avoid the case that the turning point coalesces with the pole of $Q$ .

Although the construction of the WKB solution via the topological recursion
has given in Section \ref{sec-q_curve_rec}, here we reformulate the construction 
and introduce a ``normalized WKB solution at a turning point"
to use the so-called \textit{Voros' formula} 
(see Theorem \ref{thm:Voros-formula} below). 

Two independent WKB solutions of  \eqref{eq:equiv-P1-Sch} can be written as
\begin{equation}
\varphi_{\pm}(x,\hbar) = \exp\left( \int^{x} P^{(\pm)}(x', \hbar) dx' \right),
\end{equation}
where 
$
P^{(\pm)}(x,\hbar) = \sum_{n = 0}^{\infty} \hbar^{n-1} P^{(\pm)}_{n}(x),
$
are two formal solutions of the Riccati equation
$\hbar^2 \left( P^2 + \frac{dP}{dx} \right) = Q(x,\hbar)$.
That is, $P^{(\pm)}_{n}$ are determined by 
solving the recursion relation
\begin{eqnarray}
P^{(\pm)}_{0}(x) & = & \pm\sqrt{Q_0(x)} ~=~  
\pm \sqrt{\frac{4x+(w_0-w_1)^2}{4x^2}},
\label{eq:P0} \\
P^{(\pm)}_{n+1}(x) & = & \frac{1}{2 P^{(\pm)}_{0}} 
\biggl( \delta_{n,1} Q_2(x) - 
\sum_{\substack{n_1 + n_2 = n+1 \\ n_1, n_2 \ge 1}} 
P^{(\pm)}_{n_1} P^{(\pm)}_{n_2} - \frac{dP^{(\pm)}_{n}}{dx} \biggr) 
\quad (n \ge 0). \label{eq:riccati-recursion}
\end{eqnarray}
The functions $P^{(\pm)}_n(x)$ are defined on the Riemann surface of
$\sqrt{Q_0(x)}$, which can be identified with the spectral curve
$\Sigma_{{\mathbb C}{\bf P}^1_{\bm w}}$ (see  \eqref{eq:y-and-Q0} below).
After fixing a branch cut between the turning point $v$ and $\infty$, 
we regard them as meromorphic functions on the cut plane 
${\mathbb C}{\bf P}^{1} \setminus \{\rm cut \}$. 
In what follows, we choose the branch which behaves as 
\begin{equation} \label{eq:branch-on-first-sheet}
\sqrt{Q_0(x)} = \frac{w_0 - w_1}{2x} \left( 1 + O(x) \right), 
\quad x \to 0
\end{equation}
as the branch on the first sheet. 
We will also regard the coordinate $x$ of $\mathbb{C}\textbf{P}^1$ 
(restricted to the cut plane) as that of the first sheet of
$\Sigma_{{\mathbb C}{\bf P}^1_{\bm w}}$, and use the covering involution
$\sigma : \Sigma_{{\mathbb C}{\bf P}^1_{\bm w}} \to 
\Sigma_{{\mathbb C}{\bf P}^1_{\bm w}}$ to describe a point on the second sheet.

The following statements are consequence of 
\eqref{eq:P0}, \eqref{eq:riccati-recursion} and 
\cite[Remark 2.2]{KT05}.

\begin{lemm} \label{lem:asymptotics}~
\begin{itemize}
\item[(i)] 
The asymptotic behavior of $P^{(\pm)}_n (x)$ 
when $x$ tends to $0$ are given as follows:
\begin{equation}
P^{(\pm)}_0 (x) = \pm \, \frac{w_0 - w_1}{2x} \, (1 + O(x)), \qquad
P^{(\pm)}_{1}(x) = \frac{1}{2x} \,  (1 + O(x)), 
\end{equation}
and $P^{(\pm)}_n(x)$ for $n \ge 2$ are holomorphic at $0$.

\item[(ii)]
The asymptotic behavior of $P^{(\pm)}_n (x)$ 
when $x$ tends to $\infty$ are given as follows:
\begin{equation} \label{eq:pn-infinity}
P^{(\pm)}_n (x) = O(x^{-\frac{n}{2}- \frac{1}{2}}) 
\quad (n \ge 0).
\end{equation}

\item[(iii)]
If we define 
\begin{equation}
P_{\rm odd}(x,\hbar) = \frac{P^{(+)}(x,\hbar) - P^{(-)}(x,\hbar)}{2}, \quad 
P_{\rm even}(x,\hbar) = \frac{P^{(+)}(x,\hbar) + P^{(-)}(x,\hbar)}{2},
\end{equation} 
(i.e. $P^{(\pm)} = \pm P_{\rm odd} + P_{\rm even}$), then we have
\begin{equation}
P_{\rm even} (x,\hbar) = - \frac{1}{2 P_{\rm odd}(x,\hbar)} 
\frac{dP_{\rm odd}(x,\hbar)}{dx}.
\end{equation}

\end{itemize}
\end{lemm}

Here we note that the holomorphicity of $P^{(\pm)}_n(x)$ in (i) and (ii)
is a consequence of the topological recursion
(correlation functions must be holomorphic except for the 
ramification point $v$). 

We will use a special normalization of the WKB solution 
to compute Stokes matrices, following \cite[Section 2]{KT05}.
Thanks to (iii) of Lemma \ref{lem:asymptotics} implies that 
the WKB solutions can be written in the following form 
(up to some factor which is independent of $x$):
\begin{equation} 
\varphi_{\pm} (x,\hbar) = \frac{1}{\sqrt{P_{\rm odd}(x,\hbar)}} 
\exp \left( \pm \int^{x}_{v} P_{\rm odd}(x', \hbar) dx' \right).
\end{equation}
Here the lower end-point $v$ in  \eqref{eq:WKB-alt} 
is the unique turning point of  \eqref{eq:equiv-P1-Sch}, 
and the integral is defined in terms of contour integral
\begin{equation} \label{eq:normalization-at-tp}
\int^{x}_{v} P_{\rm odd}(x', \hbar) dx' 
= \frac{1}{2} \int_{\gamma_x} P_{\rm odd}(x',\hbar) dx'
\end{equation}
along the path depicted in Figure \ref{fig:contour} 
(see Remark \ref{remark:path-sp-curve}). 
Through the relation \eqref{eq:gauge-transform}, 
we also have an expression of 
the WKB solution of  \eqref{eq:equiv-P1}:
\begin{equation}  \label{eq:WKB-alt}
\psi_{\pm}(x,\hbar) = 
\frac{\exp\left( \frac{w_0 + w_1 - \hbar}{2\hbar} \, \log x \right)}
{\sqrt{P_{\rm odd}(x,\hbar)}}
\exp \left( \pm \int^{x}_{v} P_{\rm odd}(x', \hbar) dx' \right).
\end{equation}

\begin{figure}[t]
  \scalebox{0.55}{ \vspace{-7.em} \includegraphics{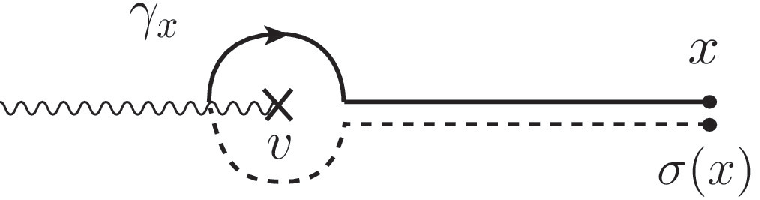}}
 \caption{For a given $x$, the path $\gamma_{x}$ starts from 
 the point $\sigma(x)$ and ends at $x$ after encircling the turning point $v$.  
 The wiggly lines designate a branch cut, and the solid (resp. dotted) 
 part represents a part of path on the first (resp. the second) 
 sheet of the spectral curve.} 
 \label{fig:contour}
 \end{figure}

\begin{remark} \label{remark:path-sp-curve}
When we integrate $P_{\rm odd} \, dx$, the path of integration in \eqref{eq:WKB-alt} should be taken on the spectral curve $\Sigma_{{\mathbb C}{\bf P}^1_{\bm w}}$. (Although the coefficients of WKB solution are defined on $\Sigma_{{\mathbb C}{\bf P}^1_{\bm w}}$, the Borel sum of WKB solution (defined in Section \ref{subsection:Borel-summation}) is single-valued around turning points (i.e., well-defined on the $x$-plane).
\end{remark}

\begin{remark} \label{rem:ambiguity}
Since $\sqrt{Q_0(x)}$ has a simple pole at $x=0$, 
the path $\gamma_x$ must avoid the point. 
If we choose different path from $v$ to $x$, then the corresponding 
WKB solutions are modified by diagonal matrix. For example, 
for the WKB solutions $\psi_{\pm}$ (resp. $\tilde{\psi}_{\pm}$)
normalized along $\gamma_x$ (resp. $\tilde{\gamma}_x$) 
depicted in Figure \ref{fig:2-contours}, we have 
\begin{equation} \label{eq:change-of-normalization}
\psi_{\pm} = \exp\left( \pm \frac{1}{2} V_{\gamma_0} \right) \,
\tilde{\psi}_{\pm},
\qquad 
V_{\gamma_0} = \oint_{\gamma_0 - \sigma_{\ast}\gamma_0} 
P_{\rm odd}(x,\hbar) \, dx = \frac{2\pi \mathrm{i} \, (w_0-w_1)}{\hbar}.
\end{equation}
(Cf.  \eqref{eq:branch-on-first-sheet}.)
Here $\gamma_0$ is a positively oriented cycle around $x=0$ 
on the first sheet of $\Sigma_{{\mathbb C}{\bf P}^1_{\bm w}}$, 
and $\sigma_{\ast}\gamma_0$ is the image of $\gamma_0$ by $\sigma$.
The integral $V_{\gamma_0}$ is called {\em Voros coefficient} 
for the closed cycle $\gamma_0$ (see \cite{DDP93,IN14}), 
which is important in the exact WKB analysis since it appears 
in the expression of monodromy or connection matrices 
of (Borel resumed) WKB solutions (\cite[Section 3]{KT05}). 
Note also that, although $V_{\gamma_0}$ is a priori a formal power series, 
$V_{\gamma_0}$ only consists of one term in our example
thanks to the holomorphicity of $P_n(x)$ for $n \ge 2$ in 
Lemma \ref{lem:asymptotics} 
(again recall that it is a consequence of topological recursion). 
\end{remark}

\begin{figure}[h]
  \scalebox{0.50}{ \vspace{-7.em} \includegraphics{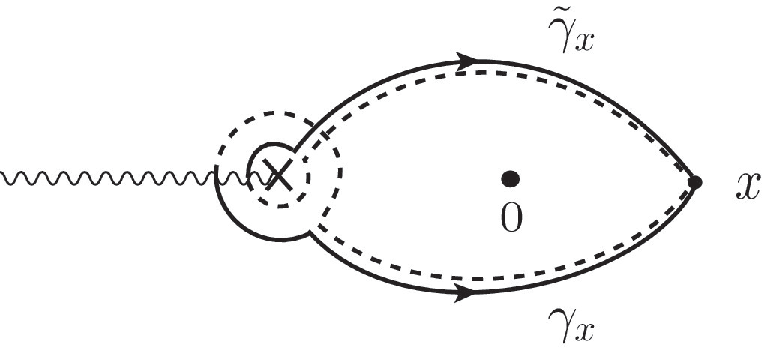}}
 \caption{The paths $\gamma_x$ and $\tilde{\gamma_x}$} 
 \label{fig:2-contours}
 \end{figure}

We also note that the formal series expression of 
 \eqref{eq:WKB-alt} can be arranged to 
\begin{equation} \label{eq:WKB-series}
\psi_{\pm}(x,\hbar) = \exp\left( \frac{1}{\hbar} S^{(\pm)}_0(x) \right) 
\,
\sum_{n = 0}^{\infty} \hbar^{n + \frac{1}{2}} \psi^{(\pm)}_{n}(x),
\end{equation}
where 
\begin{equation} \label{eq:y-and-Q0}
S^{(\pm)}_0(x) = \frac{w_0 + w_1}{2} \log x  \pm \int^{x}_v \sqrt{Q_0(x')} \, dx'
\end{equation}
(which coincides with the one computed in Table \ref{tab:eqv_CP1} 
in Appendix \ref{section:computational-results} up to an additive constant)
and $y_{\pm}(x) = x (dS^{(\pm)}_0/dx)$ 
satisfies the equation $A_{{\mathbb C}{\bf P}^1_{\bm w}}(x, y_{\pm}(x)) = 0$ 
for the GKZ curve $\Sigma_{{\mathbb C}{\bf P}^1_{\bm w}}$. 

\begin{lemm} \label{lem:limit-WKB-normalization}
The coefficients in the expansion \eqref{eq:WKB-series} satisfy 
$\lim_{x \to \infty} \psi^{(\pm)}_{n}(x) = 0$ 
for $n \ge 0$.
\end{lemm}

\begin{proof}
The term $\psi_{0}^{(\pm)}$ 
(which coincides with $\exp(S_1)$ in Table \ref{tab:eqv_CP1}
in Appendix \ref{section:computational-results}) is given by 
\begin{equation}
\psi_{0}^{(\pm)}(x) = \frac{\exp(-\frac{1}{2}\log x)}{Q_0(x)^{1/4}} 
= \frac{\sqrt{2}}{(4x+(w_0-w_1)^2)^{1/4}} = O(x^{-1/4})
\end{equation}
(which is independent of $\pm$). The behavior of subsequent terms
can be derived from the estimate  \eqref{eq:pn-infinity} 
and the equality: 
\[
\int^{x}_{v} P^{(\pm)}_{2m}(x') dx' = 
\int^{x}_{\infty} P^{(\pm)}_{2m}(x') dx' 
= O(x^{-\frac{1}{2}}) \quad (m \ge 1).
\]
The first equality holds since there is only one branch point $v$ 
on the spectral curve $\Sigma_{{\mathbb C}{\bf P}^1_{\bm w}}$, 
and $P^{(\pm)}_{2m}(x) dx$ has no residue at $x=0$ 
for $m \ge 1$ (see (i) in Lemma \ref{lem:asymptotics}). 
\end{proof}

The above lemma and Proposition \ref{prop:quantum-curve-equivariant-CP1} imply the relation between the WKB solution constructed here and the wave function constructed through the topological recursion at the level of formal power series.

\begin{prop} \label{prop:WKB-solution-and-TR-wave-function-equivr-CP1}
The WKB solution \eqref{eq:WKB-series} agrees with the wave function \eqref{phase_ms} (up to the overall factor $\hbar^{1/2}$) constructed through the topological recursion for the GKZ curve $\Sigma_{{\mathbb C}{\bf P}^1_{\bm w}}$ with the integration divisor $D = [z] - [\infty]$ (i.e., the reference point is chosen as $z_\ast=\infty$). 
\end{prop}

\subsection{Borel summation and the Stokes graph} 
\label{subsection:Borel-summation}

The expansion \eqref{eq:WKB-series} is a divergent series of $\hbar$. 
To give an analytic interpretation for  \eqref{eq:WKB-series}, we employ the 
\textit{Borel summation method} (for a formal series of $\hbar$). 
For the convenience of the readers, here we briefly recall
the Borel summation method (see \cite{Cos08} for details.)

For fixed $\theta \in {\mathbb R}$ and $x_{0} \in {\mathbb C}$ 
satisfying $x_0 \ne 0, v$, the WKB solution $\psi_{\pm}$ 
is said to be {\it Borel summable in the direction $\theta$} 
near $x_{0}$ if the following conditions are satisfied 
(see \cite[Definition 1.3]{KT05}): 
\begin{itemize}
\item The {\em Borel transform} 
\begin{equation} \label{eq:Borel-trans}
{\mathcal B}\psi_{\pm}(x,y) = 
\psi_{\pm, B}(x,y) = \sum_{n=0}^{\infty} \frac{\psi_{\pm, n}(x)}
{\Gamma(n+\frac{1}{2})} \bigl( y - a_{\pm}(x) \bigr)^{n-\frac{1}{2}}
\end{equation}
of $\psi_{\pm}$ is holomorphic on a domain 
\[
D=\{(x,y) \in U \times {\mathbb C}~|~ 
- \epsilon < {\rm Im} \bigl( {\rm e}^{-\mathrm{i} \theta}(y - a_{\pm}(x)) \bigr) 
< + \epsilon \}
\] 
with a sufficiently small $\epsilon > 0$.
Here $U$ is a neighborhood of $x_0$ and
$a_{\pm}(x) = - S^{(\pm)}_0(x)$. 
(The convergence of the Borel transform $\psi_{\pm,B}(x,y)$  
near $y=a_{\pm}(x)$ is always true; see \cite[Lemma 2.5]{KT05}.)

\item 
$|\psi_{\pm, B}(x,y)| < C_1 \mathrm{e}^{C_2|y|}$ holds with some $C_1, C_2 > 0$
on the above domain $D$.
\end{itemize}
If $\psi_{\pm}$ is Borel summable in the direction $\theta$, 
then the following Laplace integral defines 
a holomorphic function of both $x$ and $\hbar$ on 
$\{(x, \hbar) \in {\mathbb C}^2 ~|~ x \in U, \, 
|\arg \hbar - \theta|<\pi/2, \, |\hbar| \ll 1  \}$: 
\begin{equation} \label{eq:Borel-sum}
\Psi^{(\theta)}_{\pm} = 
\int_{\ell_{\theta}} {\rm e}^{- y/ \hbar} \, \psi_{\pm, B}(x,y) dy,
\end{equation}
where 
$\ell_{\theta} = \{ y = a_{\pm}(x) + r {\rm e}^{\mathrm{i} \theta} 
\in {\mathbb C} ~|~ r  \ge 0 \}$.
The function \eqref{eq:Borel-sum} is called 
the {\it Borel sum} of $\psi_{\pm}$ in the direction $\theta$.
If $\psi_{\pm}$ is Borel summable in the direction $\theta$, then
the Borel sum recovers the WKB solution as its asymptotic expansion 
(for any fixed $x \in U$):
\[
\Psi^{(\theta)}_{\pm} \sim \psi_{\pm} \quad 
\text{when $\hbar \rightarrow 0$ with $|\arg \hbar - \theta| < \pi/2$}.
\]
Moreover, for any fixed $\hbar$ satisfying $|\arg \hbar - \theta| < \pi/2$
and $|\hbar| \ll 1$, the Borel sum $\Psi^{(\theta)}_{\pm}$ is a holomorphic solution of 
the equation \eqref{eq:equiv-P1} on $U$. 
Singularities of $\psi_{\pm,B}$ on $y$-plane (Borel-plane) spoils 
the Borel summability of the WKB solutions, and hence causes 
the Stokes phenomenon. The Stokes multipliers of the WKB solutions
are discussed in \cite{Voros83,KT05} for example 
(see also Section \ref{section:Stokes-mat}).

To discuss the Borel summability and the Stokes phenomenon 
for the WKB solutions, let us recall the notion of the Stokes graph\footnote{
Stokes graph is also known as an example of spectral networks \cite{GMN}.}
for a fixed phase $\theta \in {\mathbb R}$
(see {\cite[Definition 2.6]{KT05}}).

\begin{itemize}
\item 
A {\em Stokes curve} of  \eqref{eq:equiv-P1} of phase $\theta$
is a real one-dimensional integral curve of the direction field 
\[
{\rm Im}\left( {\rm e}^{-\mathrm{i} \theta} \int^{x} \sqrt{Q_{0}(x)} \, 
dx \right) = {\rm const.}
\] 
emanating from a turning point.

\item
A {\em saddle connection} of phase $\theta$ is a Stokes curve of phase $\theta$ which connects turning points. 

\item
The {\em Stokes graph} of \eqref{eq:equiv-P1} of phase $\theta$ 
is defined as a graph on $x$-plane whose vertices are zeros and 
poles of $Q_{0}(x)dx^{2}$, and whose edges are Stokes curves 
emanating from turning points.
\end{itemize}

We will use the notation $G_\theta$ for the Stokes graph of \eqref{eq:equiv-P1} 
of phase $\theta$. Figure \ref{fig:Stokes-graphs} 
depicts $G_{\theta}$ for several $\theta$
between $0$ and $\pi$ where the equivariant parameters are chosen as 
$(w_0, w_1) = (1,0)$.

A sufficient condition for the Borel summability is given as follows. 
\begin{thm}[\cite{Koike-Schafke}] \label{thm:Koike-Schafke}
Fix $\theta \in {\mathbb R}$. 
The WKB solution \eqref{eq:WKB-alt} is Borel summable in the direction $\theta$
near any point on each face of the Stokes graph $G_{\theta}$
when the following conditions are satisfied: 
\begin{itemize}
\item[(i)] 
The upper end-point $x$ of the integral in  \eqref{eq:WKB-alt} 
does not lie on $G_\theta$. 
\item[(ii)] 
The path $\gamma_x$ of integration in  \eqref{eq:normalization-at-tp}
can be deformed in the spectral curve $\Sigma_{{\mathbb C}{\bf P}^1_{\bm w}}$ 
so that its projection by 
$\pi : \Sigma_{{\mathbb C}{\bf P}^1_{\bm w}} \to {\mathbb C}{\bf P}^1$
never intersects with saddle connections in $G_\theta$.
\end{itemize}
\end{thm}

A proof of Theorem \ref{thm:Koike-Schafke} 
will be given in forthcoming paper \cite{Koike-Schafke}. 
See also \cite[Section 3.1]{Takei}.

  \begin{figure}[t]
  \begin{minipage}{0.25\hsize}
  \begin{center}
  \includegraphics[width=27mm]{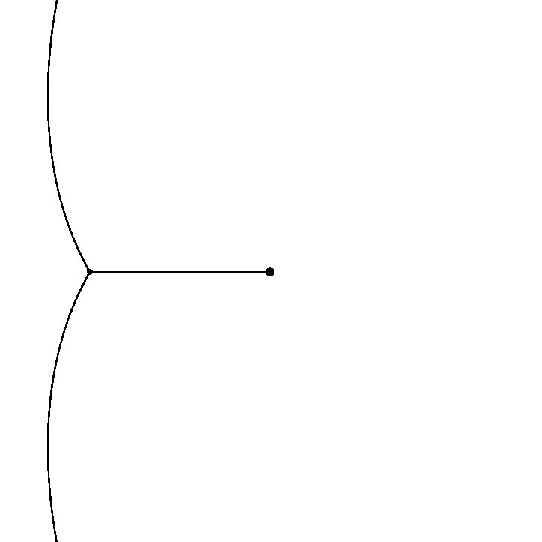} \\
  {\small{$\theta = 0$.}} 
  \end{center}
  \end{minipage}  \hspace{-1.em}
    \begin{minipage}{0.25\hsize}
  \begin{center}
  \includegraphics[width=27mm]{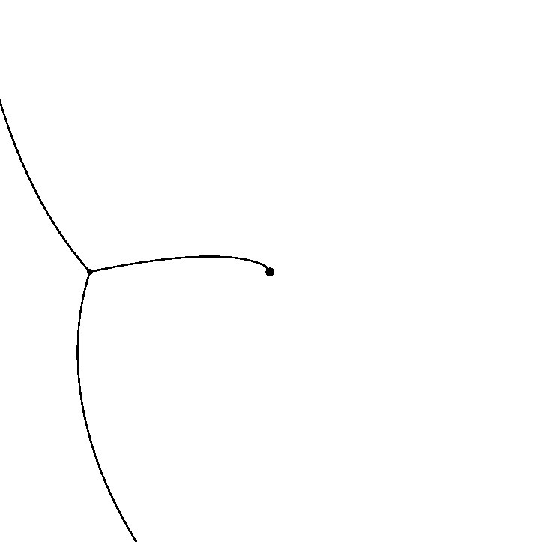} \\
  {\small{$\theta = \frac{2}{20}\pi$.}} 
  \end{center}
  \end{minipage}  \hspace{-1.em}
    \begin{minipage}{0.25\hsize}
  \begin{center}
  \includegraphics[width=27mm]{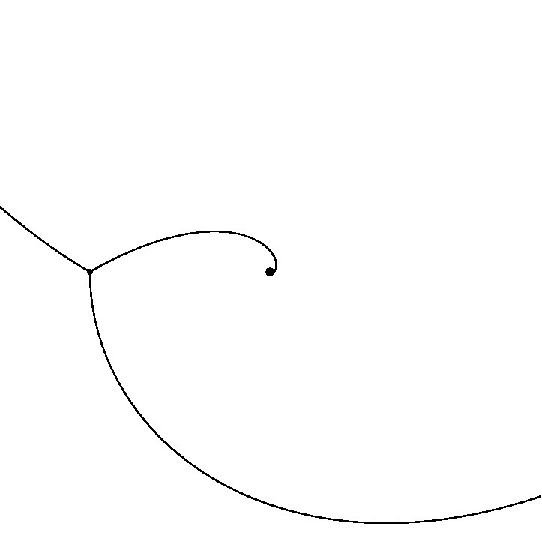} \\
  {\small{$\theta = \frac{5}{20}\pi$.}} 
  \end{center}
  \end{minipage}  \hspace{-1.em}
      \begin{minipage}{0.25\hsize}
  \begin{center}
  \includegraphics[width=27mm]{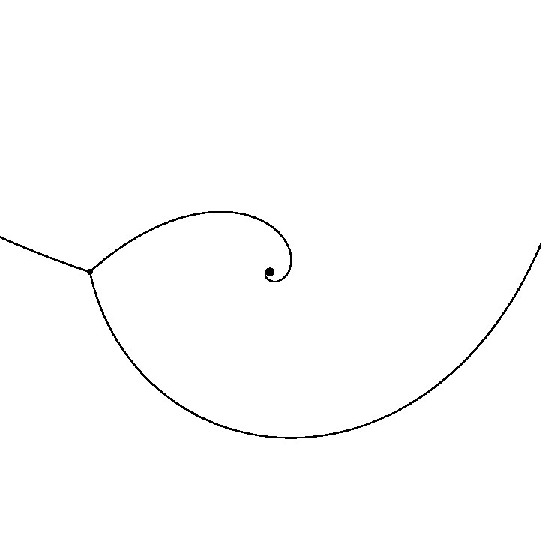} \\
  {\small{$\theta = \frac{7}{20}\pi$.}}
  \end{center}
  \end{minipage}  \\[+1.5em] %
  \begin{minipage}{0.25\hsize}
  \begin{center}
  \includegraphics[width=30mm]{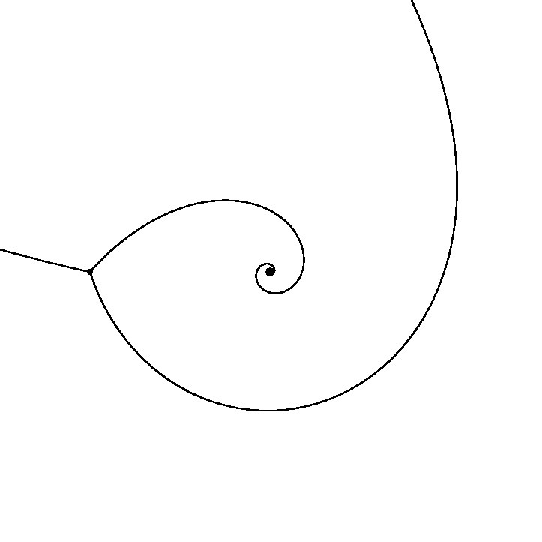} \\[-1.2em]
  {\small{$\theta = \frac{8}{20}\pi$.}} 
  \end{center} 
  \end{minipage}  \hspace{-1.em}
    \begin{minipage}{0.26\hsize}
  \begin{center}
  \includegraphics[width=29mm]{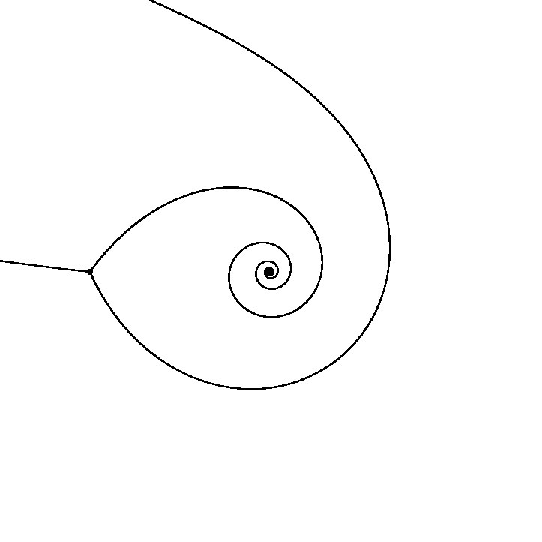} \\[-1.2em]
  {\small{$\theta = \frac{9}{20}\pi$.}} 
  \end{center}
  \end{minipage}   \hspace{-1.em}
    \begin{minipage}{0.26\hsize}
  \begin{center}
  \includegraphics[width=30mm]{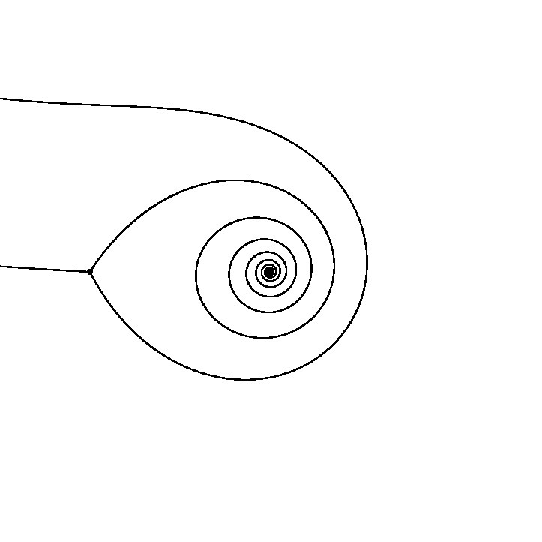} \\[-1.2em]
  {\small{$\theta = \frac{9.5}{20}\pi$.}}  
  \end{center}
  \end{minipage}  \hspace{-1.em}
    \begin{minipage}{0.24\hsize}
  \begin{center}
  \includegraphics[width=30mm]{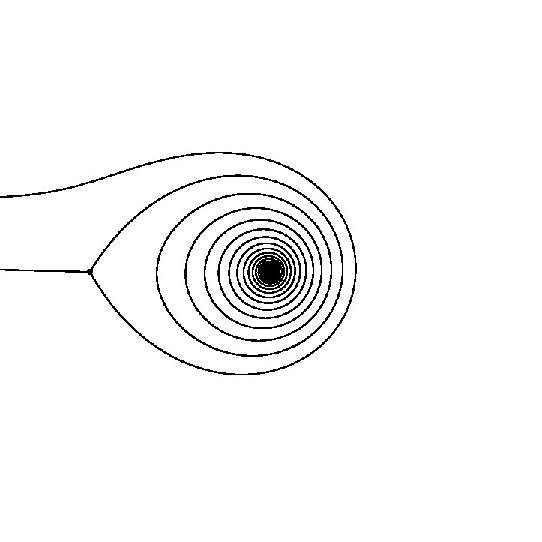} \\[-1.2em]
  {\small{$\theta = \frac{9.8}{20}\pi$.}}  
  \end{center}
  \end{minipage}  \\[-.0em] %
  \begin{minipage}{0.29\hsize}
  \begin{center}
  \includegraphics[width=36mm]{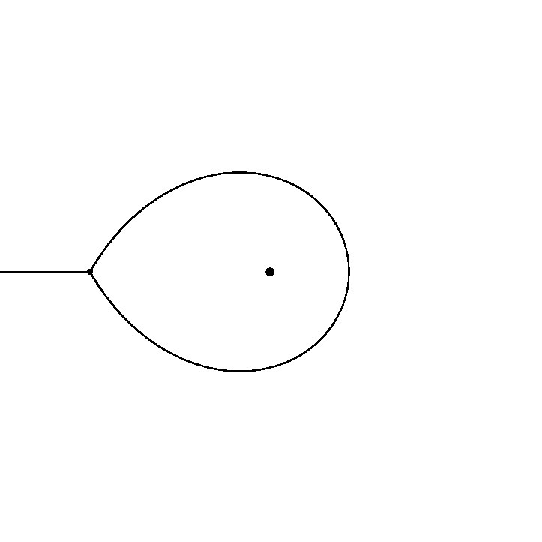} \\[-2.5em]
  {\small{$\theta = \frac{10}{20}\pi$.}}  
  \end{center}
  \end{minipage}  \\[-.3em] %
  \begin{minipage}{0.25\hsize}
  \begin{center}
  \includegraphics[width=30mm]{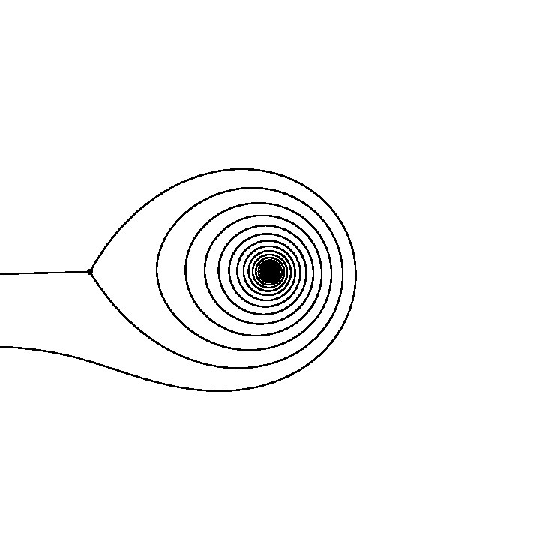} \\
  {\small{$\theta = \frac{10.2}{20}\pi$.}} 
  \end{center} 
  \end{minipage} \hspace{-1.em}
    \begin{minipage}{0.25\hsize}
  \begin{center}
  \includegraphics[width=30mm]{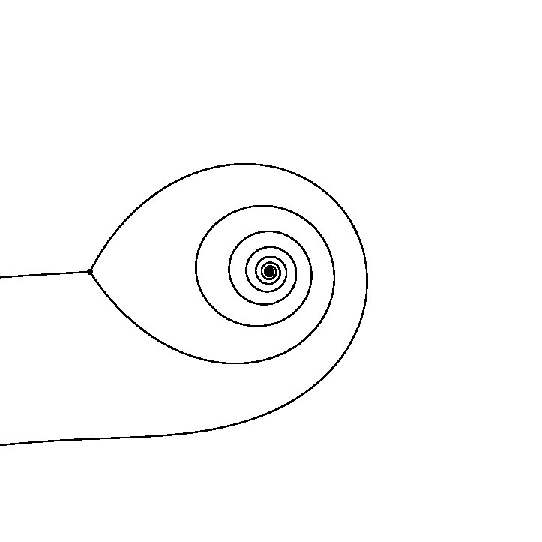} \\
  {\small{$\theta = \frac{10.5}{20}\pi$.}} 
  \end{center}
  \end{minipage} \hspace{-1.em}
    \begin{minipage}{0.25\hsize}
  \begin{center}
  \includegraphics[width=30mm]{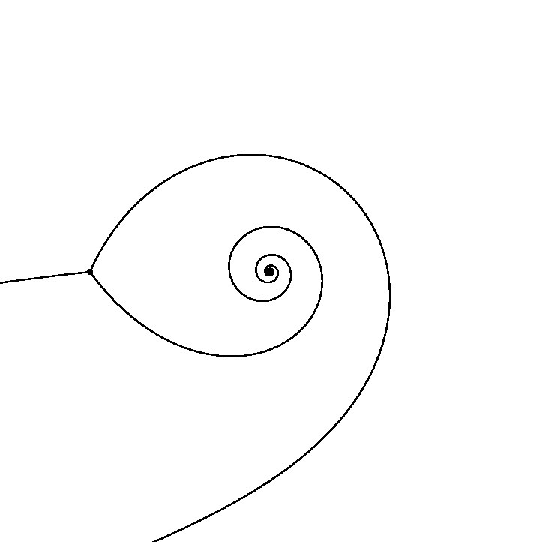} \\
  {\small{$\theta = \frac{11}{20}\pi$.}} 
  \end{center}
  \end{minipage}  \hspace{-1.em}
    \begin{minipage}{0.25\hsize}
  \begin{center}
  \includegraphics[width=30mm]{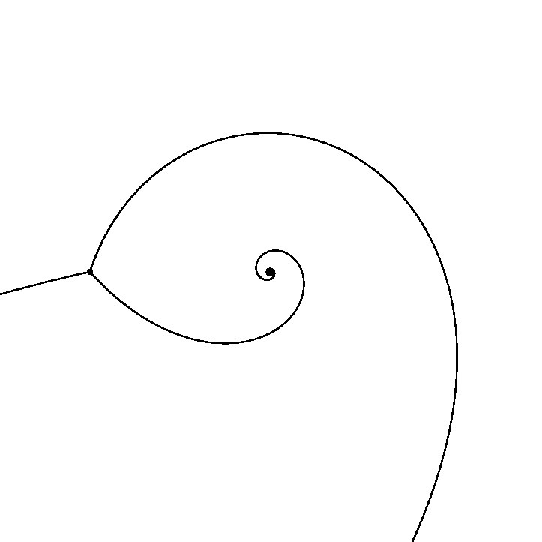} \\
  {\small{$\theta = \frac{12}{20}\pi$.}} 
  \end{center}
  \end{minipage} \\[+1.5em] %
    \begin{minipage}{0.25\hsize}
  \begin{center}
  \includegraphics[width=33mm]{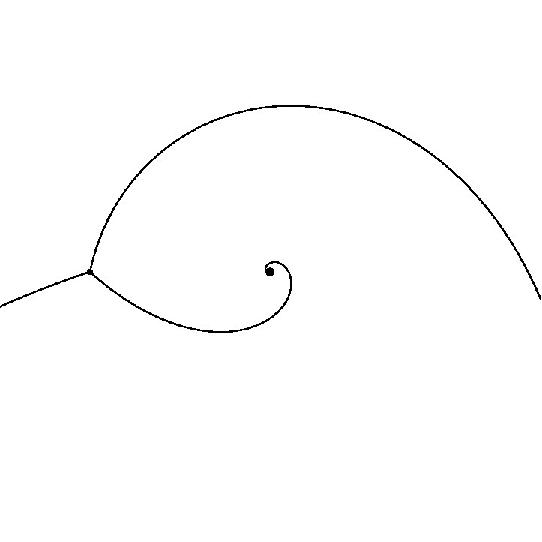} \\
  {\small{$\theta = \frac{13}{20}\pi$.}} 
  \end{center}
  \end{minipage}  \hspace{-1.em}
    \begin{minipage}{0.25\hsize}
  \begin{center}
  \includegraphics[width=33mm]{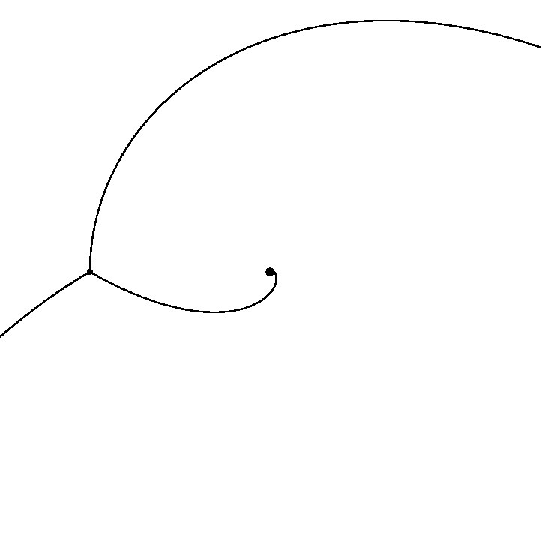} \\
  {\small{$\theta = \frac{15}{20}\pi$.}} 
  \end{center}
  \end{minipage}  \hspace{-1.em}
    \begin{minipage}{0.25\hsize}
  \begin{center}
  \includegraphics[width=33mm]{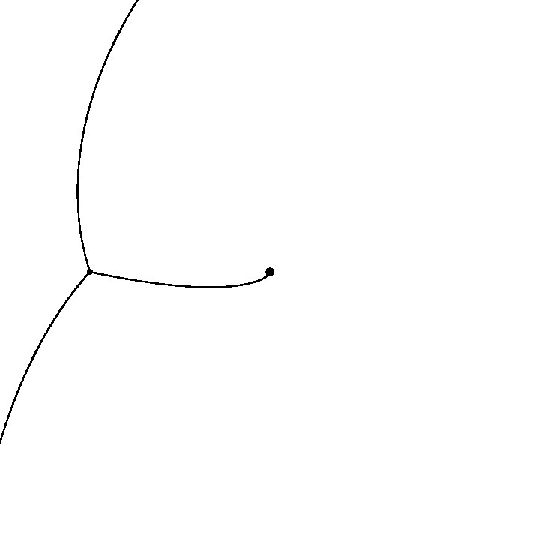} \\
  {\small{$\theta = \frac{18}{20}\pi$.}} 
  \end{center}
  \end{minipage}  \hspace{-1.em}
      \begin{minipage}{0.25\hsize}
  \begin{center}
  \includegraphics[width=33mm]{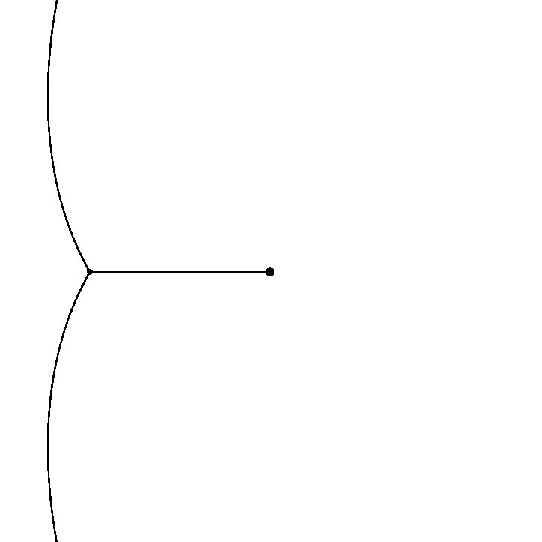} \\
  {\small{$\theta = \frac{20}{20}\pi$.}} 
  \end{center}
  \end{minipage}
  \caption{Stokes graphs of the equation \eqref{eq:equiv-P1} 
  for $w_0 - w_1 = 1$. 
  A loop-type saddle connection appears when $\theta = \pi/2$.}
  \label{fig:Stokes-graphs}
  \end{figure}

\subsection{Oscillatory integral and the Borel resummed WKB solution}
As is mentioned in Proposition \ref{prop:IGKZ}, the GKZ equation 
\eqref{eq:equiv-P1} is satisfied by the oscillatory integral
\begin{equation}
{\mathcal I}^{(\theta)}_{\pm}(x,\hbar) = \int_{\Gamma^{(\theta)}_{\pm}} 
\exp \left(\frac{W(u;x)}{\hbar} \right) 
\frac{du}{u},\quad 
W(u;x) = W_{{\mathbb C}{\bf P}^1_{\bm w}}(u;x) 
= u + \frac{x}{u} + w_0 \log u + w_1 \log \left(\frac{x}{u} \right).
\end{equation}
Here, for a fixed $\theta \in {\mathbb R}$, 
$\Gamma^{(\theta)}_{\pm}$ is the Lefschetz thimble of the phase $\theta$
(i.e. the steepest descent path for the function 
${\rm Re}({\rm e}^{- \mathrm{i} \theta} W$)) 
associated with the critical point 
\[
u_{\pm} = \frac{w_1 - w_0}{2} \pm \frac{\sqrt{4x+(w_0-w_1)^2}}{2} 
\]
of $W$. Precisely speaking, since the function $W$ contains the logarithm, we regard $u_{\pm}$ as a point on the universal cover $\widetilde{{\mathbb C}^\ast}$ of ${\mathbb C}^\ast$. Thus, any lift ${\rm e}^{2 k \pi {\rm i}}u_{\pm}$ ($k \in {\mathbb Z}$) of $u_{\pm}$ onto $\widetilde{{\mathbb C}^\ast}$ is a critical point of $W$, and the corresponding critical values satisfy $W({\rm e}^{2 k \pi {\rm i}}u_{\pm};x) = W(u_{\pm};x) + 2k \pi {\rm i} (w_0 - w_1)$. Note also that these critical points are non-degenerate as long as $x \ne v$.

Using Corollary \ref{cor:relation-oscillatory-integral-and-wave-function}, let us compare the oscillatory integral ${\mathcal I}^{(\theta)}_{\pm}$ and the Borel resummed WKB solution $\Psi^{(\theta)}_{\pm}$. 
For the purpose, we should know the well-definedness of the Lefschetz thimbles; that is, a sufficient condition which guarantees that the image $W(\Gamma_{\pm}^{(\theta)};x)$ of the Lefschetz thimble never hits other critical values of $W$. 
\begin{itemize}
\item 
If ${\rm e}^{- {\rm i} \theta} \, 2\pi {\rm i} (w_0 -w_1) \in {\mathbb R}$,  then the image $W(\Gamma_{\pm}^{(\theta)};x)$ of Lefschetz thimbles hits the critical value $W({\rm e}^{2 \pi {\rm i}} u_{\pm};x)$. Therefore, we assume that the phase $\theta$ satisfies
\begin{equation} \label{eq:tentative-condition}
{\rm e}^{- {\rm i} \theta} \, 2\pi {\rm i} (w_0 -w_1) \notin {\mathbb R}.
\end{equation}
Note that the above condition is satisfied if there is no loop-type saddle connection in the Stokes graph $G_{\theta}$ (see \eqref{eq:branch-on-first-sheet}). 

\item
Since $u_+ = u_-$ at $x=v$, we can verify $W(u_{+};x) - W(u_{-};x) = 2 \int^{x}_{v} \sqrt{Q_0(x')} dx' + 2 k \pi {\rm i} (w_0 - w_1)$, where $k$ is an integer specified by the branch of logarithm at the critical points.  In view of  \eqref{eq:branch-on-first-sheet}, we can take an appropriate path $\gamma_k$ from $v$ to $x$ (which turns around $x=0$ several times depending on $k$) satisfying $W(u_{+};x) - W(u_{-};x) = 2 \int_{\gamma_k} \sqrt{Q_0(x')} dx'$. 
Therefore, if we assume 
\begin{equation} \label{eq:tentative-condition-2}
{\rm e}^{- {\rm i} \theta} \int_{\gamma_k} \sqrt{Q_0(x')} dx' \notin {\mathbb R} \quad \text{for any $k \in {\mathbb Z}$},
\end{equation}
then the image $W(\Gamma_{\pm}^{(\theta)};x)$ of Lefschetz thimbles never intersect with each other.
Note that the condition is satisfied if $x$ does not lie on  Stokes curves of the phase $\theta$.
\end{itemize}
In summary, we can show that the Lefschetz thimbles are well defined if the conditions \eqref{eq:tentative-condition} and \eqref{eq:tentative-condition-2} are satisfied\footnote{These conditions are sufficient conditions because, even if the image of Lefschetz thimbles intersects, the oscillatory integral is well-defined and the saddle point approximation is valid if the corresponding vanishing cycles never intersect.}. 
The discussion given here also implies that the WKB solutions are Borel summable under these conditions, and hence, the Borel resummed WKB solutions $\Psi_{\pm}^{(\theta)}$ are well-defined.

Let us compare the asymptotic expansions (saddle point approximation). 
As is mentioned in Section \ref{subsection:GKZ-curves-critical-set}, the asymptotic expansion of
${\mathcal I}^{(\theta)}_{\pm}$ when $\hbar \to 0$, 
$|\arg \hbar - \theta| < \pi/2$ is given by
\begin{equation}
{\mathcal I}^{(\theta)}_{\pm}(x,\hbar) \sim 
\exp\left( \frac{1}{\hbar} W(u_{\pm};x) \right)
\frac{(-2\pi \hbar)^{1/2}}
{u_{\pm} \sqrt{{\rm Hess}(u_{\pm})}} 
\left( 1 + \sum_{n = 1}^{\infty} {\mathcal I}_n^{(\pm)}(x) \hbar^{n} \right),
\end{equation}
where ${\rm Hess}(u_{\pm}) = W^{\prime \prime}(u_{\pm};x)$ is the Hessian of $W$ at the critical point $u_{\pm}$. We can verify that $S_0^{(\pm)} = W(u_{\pm};x)$ (we fix the ambiguity in the branch of logarithm in $S_0^{(\pm)}$ so that this equality holds), and
\begin{equation}
\frac{(-2\pi)^{1/2}}{u_{\pm} \sqrt{{\rm Hess}(u_{\pm})}} 
= \frac{(\pm1)^{-1/2} (-2 \pi)^{1/2}}{(4x+(w_0-w_1)^2)^{1/4}}
= (\pm 1)^{-1/2} (- \pi)^{1/2} \psi_0^{(\pm)}(x).
\end{equation}
This equality together with Corollary \ref{cor:relation-oscillatory-integral-and-wave-function} and Proposition \ref{prop:WKB-solution-and-TR-wave-function-equivr-CP1} show that 
\begin{equation}
{\mathcal I}^{(\theta)}_{\pm} \sim 
(\pm 1)^{-1/2} (-\pi)^{1/2} \psi_{\pm}
\end{equation}
holds when $\hbar \to 0$, $|\arg \hbar - \theta| < \pi/2$.
Comparing the asymptotic expansion, we obtain the relationship between exact solutions of the GKZ equation \eqref{eq:equiv-P1}
(which is a refinement of Proposition \ref{prop:quantum-curve-equivariant-CP1}):
\begin{prop} \label{prop:relation-oscillatory-and-WKB}
${\mathcal I}^{(\theta)}_{\pm}(x,\hbar) = 
(\pm 1)^{-1/2} (-\pi)^{1/2} \Psi^{(\theta)}_{\pm}(x,\hbar)$ 
holds if the conditions \eqref{eq:tentative-condition} and \eqref{eq:tentative-condition-2} are satisfied.
\end{prop}
Hence, the computation of Stokes matrices in the subsequent sections
can be translated to results for the oscillatory integrals.

\subsection{Stokes matrices for WKB solutions normalized at the turning point}
\label{section:Stokes-mat}
Suppose that, in a direction $\theta_{0}$, 
the one of the following conditions is satisfied:
\begin{itemize}
\item[(i)] 
A Stokes curve of the phase $\theta_0$ hits the point $x$.
\item[(ii)] 
A saddle connection appears in $G_{\theta_0}$, 
and it intersects with the path $\gamma_x$.
\end{itemize}
The WKB solution is not Borel summable in these cases 
because certain singularities appear on the ray 
$\{ y = a_{\pm}(x) + r {\rm e}^{\mathrm{i} \theta_0} \in {\mathbb C}~|~ r \ge 0 \}$ 
in the Borel-plane.
Let us describe the Stokes matrices for both cases.

\subsubsection{The case {\rm (i)}}
First, we recall the Voros' connection formula which 
describes the Stokes phenomenon of the type (i)
for the Borel resummed WKB solutions.

Let us specify the situation to state the connection formula. 
Take any point $x$, and suppose that there exists 
a direction $\theta_0$ and a sufficiently small 
number $\varepsilon > 0$ satisfying the following conditions:
\begin{itemize}
\item The Stokes graphs $G_{\theta}$ have no saddle connection 
for any $\theta$ satisfying $\theta_0-\varepsilon \le \theta \le 
\theta_0+\varepsilon$.
\item 
The point $x$ lies on a Stokes curve $C$ of the phase $\theta_0$ 
emanating from the turning point $v$. Note that this assumption implies 
${\rm e}^{- \mathrm{i} \theta_0} \int^{x}_v \sqrt{Q_0(x')} \, dx' \in {\mathbb R}_{\ne 0}$ 
(where the integral is taken along $C$) by the definition of Stokes curves. 
\item 
The point $x$ does not lie on 
$G_{\theta}$ for any $\theta$ satisfying 
$\theta_0-\varepsilon \le \theta < \theta_0$ or
$\theta_0 < \theta \le \theta_0+\varepsilon$.
\end{itemize}
Let $\psi_{\pm}$ be the WKB solution \eqref{eq:WKB-alt} normalized at $v$
along the Stokes curve $C$ of phase $\theta_0$.
Denote by $\Psi^{(\theta_0-\varepsilon)}_{\pm}$
(resp. $\Psi^{(\theta_0+\varepsilon)}_{\pm}$) 
the Borel sum of $\psi_{\pm}$ in the direction 
$\theta_0-\varepsilon$ (resp. $\theta_0+\varepsilon$). 
Then, we have the following statement.
\begin{thm}[{\cite[Section 6]{Voros83}; see also \cite[Theorem 2.23]{KT05}}]
\label{thm:Voros-formula}
In the situation above, the Borel transformed WKB solution $\psi_{+,B}$ 
(resp. $\psi_{-,B}$) has the singular point at $y = a_{-}(x)$ 
(resp. at $y = a_{+}(x)$). 
Moreover, the following equalities hold:
\begin{itemize}
\item[(i)] 
When ${\rm e}^{-\mathrm{i} \theta_0} \int^{x}_{v} \sqrt{Q_0(x)}dx > 0$ on $C$, then
\begin{equation} \label{eq:Voros-formula-Stokes-1}
(\Psi^{(\theta_0-\varepsilon)}_{+}, \Psi^{(\theta_0-\varepsilon)}_{-})
= (\Psi^{(\theta_0+\varepsilon)}_{+}, \Psi^{(\theta_0+\varepsilon)}_{-})
\begin{pmatrix} 
1 & 0 \\ -\mathrm{i}  & 1
\end{pmatrix}.
\end{equation}

\item[(ii)] 
When ${\rm e}^{-\mathrm{i} \theta_0} \int^{x}_{v} \sqrt{Q_0(x)}dx < 0$ on $C$, then
\begin{equation}\label{eq:Voros-formula-Stokes-2}
(\Psi^{(\theta_0-\varepsilon)}_{+}, \Psi^{(\theta_0-\varepsilon)}_{-})
= (\Psi^{(\theta_0+\varepsilon)}_{+}, \Psi^{(\theta_0+\varepsilon)}_{-})
\begin{pmatrix} 
1 & -\mathrm{i} \\ 0  & 1
\end{pmatrix}.
\end{equation}

\end{itemize}
\end{thm}

The lower/upper triangular matrices in \eqref{eq:Voros-formula-Stokes-1} 
and \eqref{eq:Voros-formula-Stokes-2}
are called the {\em Stokes matrices} associated with the direction $\theta_0$.

\subsubsection{The case {\rm (ii)}}
Next let us show the formula for the Stokes phenomenon of the type (ii)
caused by the loop-type saddle connection. 

Suppose that a direction $\theta_0$ and a point $x$ satisfy the following conditions:
\begin{itemize}
\item The Stokes graph $G_{\theta_0}$ has 
a loop-type saddle connection around $0$.
\item The point $x$ does not lie on the Stokes graph $G_{\theta_0}$.
\end{itemize}
Set 
\begin{equation} \label{eq:period-delta}
\delta = \oint_{\gamma_0 - \sigma_{\ast}\gamma_0} \sqrt{Q_0(x)} dx 
= 2 \pi \mathrm{i} (w_0 - w_1).
\end{equation}
The first assumption implies that 
${\rm e}^{-\mathrm{i} \theta_0} \delta \in {\mathbb R}_{\ne 0}$.
To specify the situation, we further assume:
\begin{itemize}
\item The real part of ${\rm e}^{-\mathrm{i} \theta_0} \delta$ is positive. 
\end{itemize}
The second assumption implies $x$ lies on one of connected components 
$D_0$ and $D_{\infty}$ of ${\mathbb C}{\bf P}^1 \setminus \{ \rm loop \}$, 
where $D_0$ (resp. $D_{\infty}$) contains $x=0$ (resp. $x=\infty$).

Let $\psi_{\pm}$ be the WKB solution \eqref{eq:WKB-alt} normalized at $v$. 
Note that there is an ambiguity in the choice of the path from $v$ to $x$ 
(see Remark \ref{rem:ambiguity}), 
but the following formula holds for arbitrary choice.

\begin{thm}[{\cite{AIT}}]  \label{thm:AIT}
In the situation above, the following statements hold:
\begin{itemize}
\item[(i)] 
If $x \in D_{\infty}$, then the WKB solutions $\psi_{\pm}$
are Borel summable in the direction $\theta_0$. 
In particular, the Borel sum of the WKB solutions satisfy
\begin{equation}
(\Psi^{(\theta_0-\varepsilon)}_{+}, \Psi^{(\theta_0-\varepsilon)}_{-})
= (\Psi^{(\theta_0+\varepsilon)}_{+}, \Psi^{(\theta_0+\varepsilon)}_{-}),
\end{equation} 
where $\Psi^{(\theta_0-\varepsilon)}_{\pm}$
(resp. $\Psi^{(\theta_0+\varepsilon)}_{\pm}$) 
are the Borel sum of $\psi_{\pm}$ in the direction $\theta_0-\varepsilon$ 
(resp. $\theta_0+\varepsilon$) for sufficiently small $\varepsilon > 0$.

\item[(ii)] 
If $x \in D_{0}$, then the Borel transformed WKB solution $\psi_{\pm,B}$
has singular points at $y = a_{\pm}(x) + m \delta$ 
with $m \in {\mathbb Z}_{\ne 0}$ 
(and hence $\psi_{\pm}$ is not Borel summable in the direction $\theta_0$). 
The Borel sum of the WKB solutions satisfy
\begin{equation} \label{eq:4d-BPS}
(\Psi^{(\theta_0, R)}_{+}, \Psi^{(\theta_0, R)}_{-})
= (\Psi^{(\theta_0, L)}_{+}, \Psi^{(\theta_0, L)}_{-}) 
\begin{pmatrix}
1 - {\rm e}^{-V_{\gamma_0}} & 0 \\ 0 & (1 - {\rm e}^{-V_{\gamma_0}})^{-1}
\end{pmatrix},
\end{equation} 
where $\Psi^{(\theta_0, R)}_{\pm}$
(resp. $\Psi^{(\theta_0, L)}_{\pm}$) 
is the Borel sum of $\psi_{\pm}$ in the direction $\theta_0$ 
defined as the Laplace integral as \eqref{eq:Borel-sum} along 
the path $\ell_{\theta_0,R}$ (resp. $\ell_{\theta_0,L}$)
depicted in Figure \ref{fig:right-left-path}.
\end{itemize}
\end{thm}

\begin{figure}[h]
  \scalebox{0.65}{ \vspace{-7.em} \includegraphics{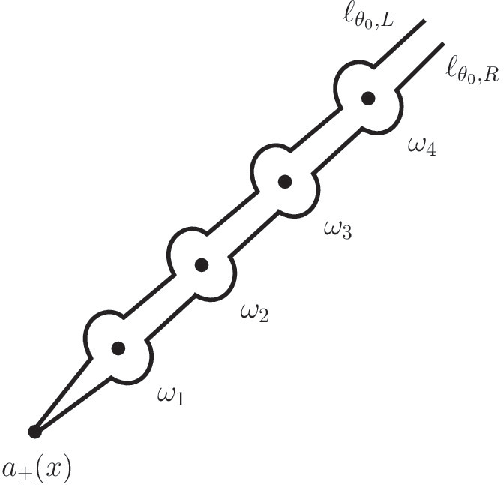}}
 \caption{The paths $\ell_{\theta_0,R}$ and $\ell_{\theta_0,L}$. 
 These paths are obtained by deforming the path $\ell_{\theta_0}$; 
 $\ell_{\theta_0,R}$ and $\ell_{\theta_0,L}$ avoid the singular point 
 $\omega_m = a^{(\pm)}(x) + m \delta$ 
 from the right and the left, respectively.} 
 \label{fig:right-left-path}
 \end{figure}

We also call the diagonal matrix in \eqref{eq:4d-BPS} 
the {\em Stokes matrix} for the direction $\theta_0$. 
This kind of Stokes phenomenon is never observed in the non-equivariant case
(i.e. $w_0 - w_1 = 0$). 
The condition for $\delta$ implies that the quantity 
${\rm e}^{- V_{\gamma_0}} = {\rm e}^{- \delta/\hbar}$ in the Stokes matrix is 
exponentially small when $\hbar \to 0$, $\arg \hbar = \theta_0$.

\begin{remark}
In the situation $x \in D_{0}$, for any sufficiently small $\varepsilon > 0$, 
there are infinitely many directions between $\theta_0 - \varepsilon$ and 
$\theta_0 + \varepsilon$ where the WKB solution is not Borel summable. 
Hence we employ the slightly different version of the Borel sum 
in the relation \eqref{eq:4d-BPS}. 
The effect of infinitely many Stokes phenomenon will be discussed in the next subsection.

\end{remark}

\subsection{Computation of the total Stokes matrix} 
In this subsection, for a fixed $x$, 
we compute the {\em ``total" Stokes matrix} defined as 
\begin{equation}
S_{\rm tot} = \prod_{0 \le \theta < \pi}^{\leftarrow} S_{\theta}.
\end{equation}
Here $S_\theta$ is the Stokes matrix for the WKB solution 
in the direction $\theta$, and they are multiplied from the left 
as $\theta$ increases. We regard $S_{\theta} = {\rm Id}$ when 
$\psi_{\pm}$ is Borel summable in the direction $\theta$. 
Therefore, $S_{\rm tot}$ relates the Borel sum of WKB solutions 
in opposite directions:
\begin{equation}
(\Psi_+^{(0)}, \Psi_-^{(0)}) = (\Psi_+^{(\pi)}, \Psi_-^{(\pi)}) S_{\rm tot}.
\end{equation}

In what follows, we consider the following situation:
\begin{itemize}
\item $w_0 - w_1 \in {\mathbb R}_{>0}$; that is, 
$\delta$ defined in \eqref{eq:period-delta} satisfies
${\rm e}^{- \pi \mathrm{i}/2}\delta \in {\mathbb R}_{>0}$. 

\item $x$ is fixed at one of $x_1$, $x_2$ or $x_3$ 
satisfying the following conditions 
(see Figure \ref{fig:base-point-and-path} (a) and (b)
which depicts the Stokes graph $G_0$ for $\theta = 0$
and $G_{\pi/2}$ for $\theta = \pi/2$, respectively): 
\begin{itemize}
\item
$x_1$ and $x_2$ are points on the same Stokes region, 
which has the point $x=0$ on its boundary in $G_0$.
We choose $x_1 \in D_{\infty}$ and $x_2 \in D_0$.

\item 
$x_3$ is on the other Stokes region in $G_0$, 
which does not contain $x_1$ and $x_2$.
\end{itemize}

\item 
We use the WKB solution normalized at the turning point $v$, 
whose normalization path from $v$ to $x_i$ ($i=1,2,3$) 
is chosen as indicated in Figure \ref{fig:base-point-and-path} (a).
\end{itemize}

In this situation, we can verify that 
$\int^{x}_{v} \sqrt{Q_0(x)} dx < 0$ 
on the Stokes curve of the phase $\theta = 0$ 
which flows into the origin. We will employ the technique 
developed in \cite[Section 3]{KT05}
to compute the Stokes matrix for the WKB solutions.

\begin{figure}[h]
  \begin{minipage}{0.45\hsize}
  \begin{center}
  \includegraphics[width=60mm]{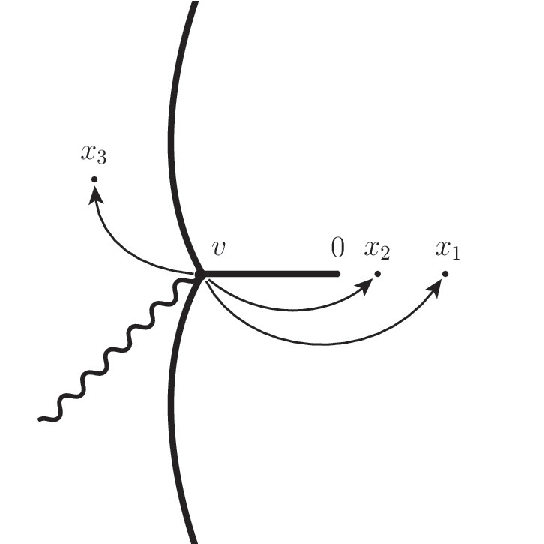} \\
  {\small{(a) : Stokes graph for $\theta = 0$.}} 
  \end{center}
  \end{minipage}  \hspace{-1.5em}
    \begin{minipage}{0.45\hsize}
  \begin{center}
  \includegraphics[width=60mm]{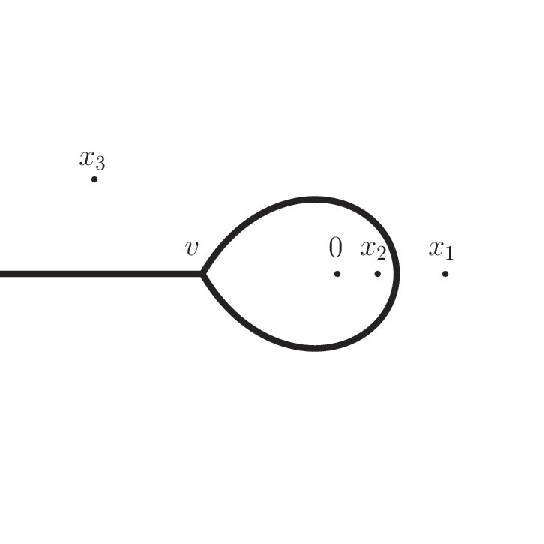} \\
  {\small{(b) : Stokes graph for $\theta = \pi/2$.}} 
  \end{center}
  \end{minipage} 
 \caption{The points $x_1$, $x_2$ and $x_3$. 
 Note that this figure depicts the specific situation where 
 $w_0 - w_1 = 1$, $x_1 = 0.8$, $x_2 = 0.3$ and $x_3 = -1.8 + 0.7 i$. 
 However, the formulas obtained in this subsection 
 holds in general (under the assumption above).} 
 \label{fig:base-point-and-path}
 \end{figure}

\subsubsection{Stokes matrix at $x_1$}
Since the point $x_1$ is contained in $D_{\infty}$, 
Figure \ref{fig:Stokes-graphs} shows that $x_1$ is hit by 
Stokes curves twice when we vary $\theta$ from $0$ to $\pi$.
The situation for the first hit (resp. the second hit) 
is depicted in Figure \ref{fig:critical-direction} (a) 
(resp. Figure \ref{fig:critical-direction} (b)), 
and it happens at some $\theta_1$ satisfying 
$0 < \theta_1 < \pi/2$ (resp. at some $\theta_2$ 
satisfying $\pi/2 < \theta_2 < \pi$).
These two Stokes curves causes Stokes phenomena for the WKB solutions,
and each of contribution to $S_{\rm tot}$ is given as follows:

\begin{itemize}
\item 
For the first hit depicted in Figure \ref{fig:critical-direction} (a), 
since ${\rm e}^{- \mathrm{i} \theta_1} \int_{\gamma_{x_1}} \sqrt{Q_0(x)}dx > 0$
holds in this situation, Theorem \ref{thm:Voros-formula} (i) implies that 
the corresponding Stokes matrix is given by 
\begin{equation}
(\Psi^{(\theta_1-\varepsilon)}_{+}, \Psi^{(\theta_1-\varepsilon)}_{-}) 
= 
(\Psi^{(\theta_1+\varepsilon)}_{+}, \Psi^{(\theta_1+\varepsilon)}_{-}) 
S_{\theta_1}, \quad 
S_{\theta_1} = \begin{pmatrix} 1 & 0 \\ -\mathrm{i} & 1 \end{pmatrix}.
\end{equation}

\item 
For the second hit depicted in Figure \ref{fig:critical-direction} (b), 
we need to care about the normalization of the WKB solutions. 
Since $\psi_{\pm}$ is not normalized along the Stokes curve which hits $x_1$
in Figure \ref{fig:critical-direction} (b), we cannot apply 
Theorem \ref{thm:Voros-formula} directly. However, 
Theorem \ref{thm:Voros-formula} can be applied to the WKB solution 
$\tilde{\psi}_{\pm}$ which is normalized along the path $\tilde{\gamma}_x$ 
in Figure \ref{fig:2-contours}. Since the WKB solutions $\psi_{\pm}$ 
and $\tilde{\psi}_{\pm}$ are related as \eqref{eq:change-of-normalization}, 
we can compute the Stokes matrix for $\psi_{\pm}$ and obtain
the Stokes matrix (note that 
${\rm e}^{- \mathrm{i} \theta_2} \int_{\tilde{\gamma}_{x_1}} \sqrt{Q_0(x)}dx < 0$
in this case): 
\begin{equation}
(\Psi^{(\theta_2-\varepsilon)}_{+}, \Psi^{(\theta_2-\varepsilon)}_{-}) 
= 
(\Psi^{(\theta_2+\varepsilon)}_{+}, \Psi^{(\theta_2+\varepsilon)}_{-}) S_{\theta_2}, \quad 
S_{\theta_2} = 
\begin{pmatrix} 1 & -\mathrm{i} \, {\rm e}^{-V_{\gamma_0}} \\ 0 & 1 \end{pmatrix}.
\end{equation}
\end{itemize}

\begin{figure}[t]
  \begin{minipage}{0.45\hsize}
  \begin{center}
  \includegraphics[width=50mm]{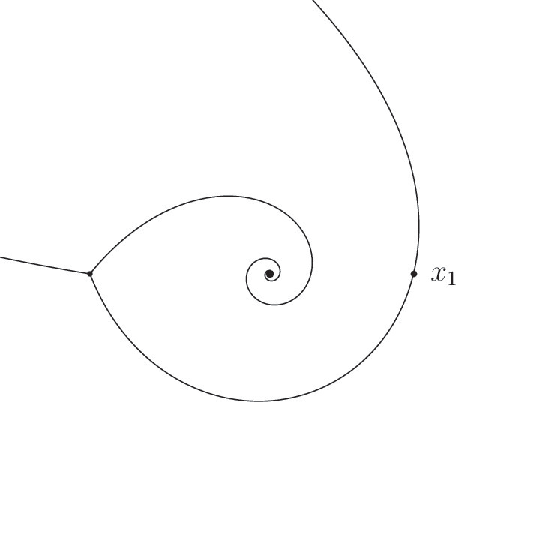} \\
  {\small{(a)}} 
  \end{center}
  \end{minipage}  \hspace{-1.em}
    \begin{minipage}{0.45\hsize}
  \begin{center}
  \includegraphics[width=50mm]{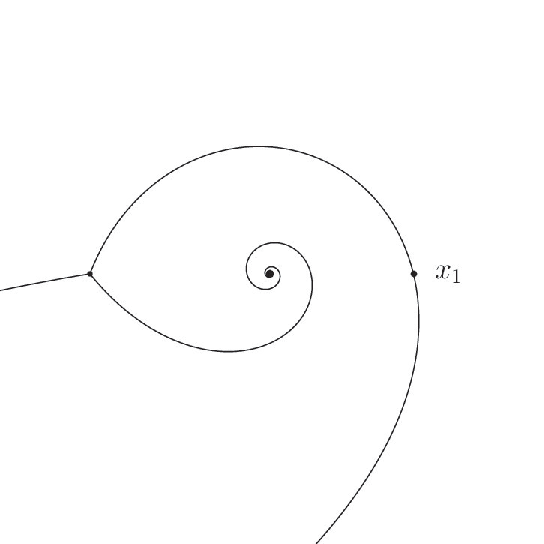} \\
  {\small{(b)}} 
  \end{center}
  \end{minipage}  \hspace{-1.em}
 \caption{The point $x_1$ is hit by Stokes curves. The situation (a) 
 (resp. (b)) occurs when 
 $\theta_1 = \arg \bigl( \int_{\gamma_{x_1}} \sqrt{Q_0(x)} dx \bigr)$ 
 (resp.  $\theta_2 = \arg \bigl( \int_{\tilde{\gamma}_{x_1}} 
 \sqrt{Q_0(x)} dx \bigr)$)
 where $\gamma_{x_1}$ and $\tilde{\gamma}_{x_1}$ are depicted in 
 Figure \ref{fig:2-contours}.} 
 \label{fig:critical-direction}
 \end{figure}

Although the loop-type saddle connection appears in $\theta = \pi/2$, 
WKB solution is Borel summable since $x_1 \in D_{\infty}$ 
(see Theorem \ref{thm:AIT} (i)), and hence 
no Stokes phenomenon occurs in $\theta = \pi/2$.
Therefore, we have

\begin{thm}
The total Stokes matrix at $x_1$ is given by 
\begin{eqnarray}
S_{\rm tot} =  S_{\theta_2} S_{\theta_1} = 
\begin{pmatrix}
1 - {\rm e}^{-V_{\gamma_0}} & - \mathrm{i} \, {\rm e}^{-V_{\gamma_0}} \\ -\mathrm{i} & 1 
\end{pmatrix}.
\end{eqnarray}
\end{thm}

\subsubsection{Stokes matrix at $x_2$}
From Figure \ref{fig:Stokes-graphs}, we can observe that, 
for $0 \le \theta < \pi/2$ and for $\pi/2 < \theta < \pi$, 
Stokes curves hit $x_2$ infinitely many times. 
Stokes directions (i.e. directions with non-trivial $S_{\theta}$) 
accumulates to $\theta = \pi/2$ due to the spiral behavior of the Stokes curve.
In fact, such situation has been analyzed in \cite[Appendix B]{IN14} 
and the total Stokes matrix which relates
$\Psi_{\pm}^{(0)}$ and $\Psi_{\pm}^{(\pi)}$ 
is given by a convergent infinite product of matrices as follows.

\begin{itemize}
\item 
Let us describe the Stokes matrix between $\Psi_{\pm}^{(0)}$
and $\Psi_{\pm}^{(\frac{\pi}{2},R)}$ which includes all contribution
from the Stokes curve which hits $x_2$ infinitely many times in 
$0 < \theta < \pi/2$.
For each hit, the computation of the Stokes matrix can be done in 
a similar manner as above. Since 
${\rm e}^{- \mathrm{i} \theta} \int^{x}_{v} \sqrt{Q_0(x)} dx < 0$ 
along the Stokes curve which hits $x_2$, each of Stokes matrix 
is upper triangular due to Theorem \ref{thm:Voros-formula}.
The result of the computation in \cite[Appendix B]{IN14} shows 
\begin{eqnarray}
(\Psi^{(0)}_{+}, \Psi^{(0)}_{-}) & = & 
(\Psi^{(\frac{\pi}{2}, R)}_{+}, \Psi^{(\frac{\pi}{2}, R)}_{-}) 
\prod_{n \ge 1} 
\begin{pmatrix} 1 & - \mathrm{i} \, {\rm e}^{-n V_{\gamma_0}} \\ 0 & 1 \end{pmatrix}. 
\label{eq:St-x2-1}
\end{eqnarray}
The infinite product converges as long as $\hbar$ 
lies on the upper half plane.

\item For $\theta = \pi/2$, we can use Theorem \ref{thm:AIT} (ii):
\begin{equation}
(\Psi^{(\frac{\pi}{2}, R)}_{+}, \Psi^{(\frac{\pi}{2}, R)}_{-})  = 
(\Psi^{(\frac{\pi}{2}, L)}_{+}, \Psi^{(\frac{\pi}{2}, L)}_{-}) 
\begin{pmatrix} \label{eq:St-x2-2}
1 - {\rm e}^{-V_{\gamma_0}} & 0 \\ 0 & (1-{\rm e}^{-V_{\gamma_0}})^{-1}
\end{pmatrix}.
\end{equation}

\item 
For $\pi/2 < \theta < \pi$, a Stokes curve hits $x_2$ infinitely many times again. 
Similarly to the case $0 < \theta < \pi/2$, we have
\begin{eqnarray}
(\Psi^{(\frac{\pi}{2}, L)}_{+}, \Psi^{(\frac{\pi}{2}, L)}_{-})  & = & 
(\Psi^{(\pi)}_{+}, \Psi^{(\pi)}_{-})  
\prod_{n \ge 0} \begin{pmatrix} 1 & 0 \\ - \mathrm{i} \, {\rm e}^{-n V_{\gamma_0}} & 1 
\end{pmatrix}. 
\label{eq:St-x2-3}
\end{eqnarray}
\end{itemize}

The total Stokes matrix is the product of the matrices computed in 
\eqref{eq:St-x2-1}, \eqref{eq:St-x2-2} and \eqref{eq:St-x2-3}:

\begin{thm}
The total Stokes matrix at $x=x_2$ is given by 
\begin{eqnarray}
S_{\rm tot} 
& = & 
\left(\prod_{n \ge 0} 
\begin{pmatrix} 1 & 0 \\ - \mathrm{i} \, {\rm e}^{-n V_{\gamma_0}} & 1 
\end{pmatrix} \right) 
\begin{pmatrix} 1 - {\rm e}^{-V_{\gamma_0}} & 0 \\ 
0 & (1-{\rm e}^{-V_{\gamma_0}} )^{-1} 
\end{pmatrix} 
\left( \prod_{n \ge 1} 
\begin{pmatrix} 1 & - \mathrm{i} \, {\rm e}^{-n V_{\gamma_0}} \\ 
0 & 1 \end{pmatrix} \right)
\nonumber \\
& = & 
\begin{pmatrix}
1 - {\rm e}^{-V_{\gamma_0}} & - \mathrm{i} \, {\rm e}^{-V_{\gamma_0}} \\ -\mathrm{i} & 1 
\end{pmatrix}.
\end{eqnarray}
\end{thm}

\begin{remark} \label{remark:wall-crossing-formula}
We have observed that the total Stokes matrix at $x_1$ and at $x_2$ are the same matrix.
This is because $x_1$ and $x_2$ lie on the same Stokes region in $G_0$ and 
the Borel resummed WKB solutions at $x_1$ and $x_2$ give the same basis of 
the space of the solution of the equation \eqref{eq:equiv-P1}. 
However, this is a consequence of a non-trivial identity of an infinite product 
of matrices:
\begin{multline} \label{eq:2d-4d-WCF}
\begin{pmatrix} 1 & -\mathrm{i} \, {\rm e}^{V_{\gamma_0}} \\ 0 & 1 \end{pmatrix} 
\begin{pmatrix} 1 & 0 \\ -\mathrm{i} & 1 \end{pmatrix}  \\ 
= \left(\prod_{n \ge 0} 
\begin{pmatrix} 1 & 0 \\ - \mathrm{i} \, {\rm e}^{-n V_{\gamma_0}} & 1 \end{pmatrix} 
\right) 
\begin{pmatrix} 1 - {\rm e}^{-V_{\gamma_0}} & 0 \\ 
0 & (1-{\rm e}^{-V_{\gamma_0}} )^{-1} \end{pmatrix}
\left( \prod_{n \ge 1} 
\begin{pmatrix} 1 & - \mathrm{i} \, {\rm e}^{-n V_{\gamma_0}} \\ 0 & 1 \end{pmatrix} 
\right).
\end{multline}
Note that, the identity \eqref{eq:2d-4d-WCF} is an example of the 
{\em 2d/4d wall-crossing formula} in the sense of 
\cite{GMN} by Gaiotto-Moore-Neitzke. 
Each upper and lower triangular matrix 
(contribution from the situation where $x$ lies on a Stokes curve) 
captures 2d-BPS states, while the diagonal matrix 
(contribution from the loop-type saddle connection) 
captures a 4d-BPS state. 
Here the loop-type Stokes curve plays the role of wall, 
and the above identity describes how BPS indices ``jump" when $x$ crosses the wall.
The 2d/4d wall-crossing formula in ${\mathbb C}{\bf P}^1$ sigma model\footnote{
In this context, the GKZ curve appears as the chiral ring studied by N.~Dorey
\cite{Dorey:1998yh}.
} 
has already been mentioned in \cite[Section 8.2]{GMN}, and we have given 
an exact WKB theoretic interpretation of the wall-crossing formula in this special case.
\end{remark}

\subsubsection{Stokes matrix at $x = x_3$ and Dubrovin's conjecture}
As well as the case of $x = x_1$, Figure \ref{fig:Stokes-graphs} shows that 
there are only two Stokes directions in this case. 
We can verify that the two Stokes matrices are both lower triangular, 
and the result is given as follows:
\begin{thm} \label{thm:Stokes-at-x3}
The total Stokes matrix at $x = x_3$ is given by 
\[
S_{\rm tot} = 
\begin{pmatrix} 1 & 0 \\ -\mathrm{i} \, {\rm e}^{V_{\gamma_0}} & 1 \end{pmatrix}  
\begin{pmatrix} 1 & 0 \\ -\mathrm{i} & 1 \end{pmatrix}  = 
\begin{pmatrix} 1 & 0 \\ -\mathrm{i} \, (1 + {\rm e}^{V_{\gamma_0}}) & 1 \end{pmatrix}. 
\]
\end{thm}

Using the relation proved in Proposition \ref{prop:relation-oscillatory-and-WKB}, 
we obtain the Stokes matrix for the oscillatory integral solutions: 
\begin{equation}
(I^{(0)}_{+}, I^{(0)}_{-}) = (I^{(\pi)}_{+}, I^{(\pi)}_{-}) 
\begin{pmatrix} 1 & 0 \\ 1+{\rm e}^{V_{\gamma_0}} & 1 \end{pmatrix}.
\end{equation}
In the non-equivariant limit $w_0 - w_1 \to 0$ 
the non-trivial Stokes multiplier $1 + {\rm e}^{V_{\gamma_0}}$ 
tends to $2$, and this coincides with the Euler pairing 
$\chi({\mathcal O}, {\mathcal O}(1)) = 2$ 
on the derived category $D^{b}{\rm Coh}({\mathbb C}{\bf P}^1)$.
This is consistent with the Dubrovin's conjecture \cite{dubrovin-conj} 
for the Stokes matrix of the quantum cohomology of ${\mathbb C}{\bf P}^1$.

Here we propose a statement which suggests that there is an 
``equivariant-version" of the Dubrovin's conjecture. 
We will verify that the above Stokes multiplier $1 + {\rm e}^{V_{\gamma_0}}$
for the quantum differential equation \eqref{eq:equiv-P1} 
of the equivariant ${\mathbb C}{\bf P}^1$
can be identified with the equivariant Euler pairing of equivariant coherent sheaves on 
${\mathbb C}{\bf P}^1$, following an idea of Iritani.

Let $T =  ({\mathbb C}^{*})^2$ be the algebraic torus. 
We regard ${\mathbb C}{\bf P}^{1}$ as a $T$-space by the action 
\[
T \times {\mathbb C}{\bf P}^1 \to {\mathbb C}{\bf P}^1, ~~~~\quad~~~~
(t_0, t_1) \cdot [x_0 : x_1] = [t_0^{-1} x_0: t_1^{-1} x_1],
\]
where $[x_0:x_1]$ is the homogeneous coordinate of ${\mathbb C}{\bf P}^1$. 
Let us also regard ${\mathcal O}$ and ${\mathcal O}(1)$ as 
$T$-equivariant coherent sheaves (see \cite[Section 5]{CG} for example) 
on ${\mathbb C}{\bf P}^1$ as follows:
\begin{itemize}
\item
${\mathcal O}$ is equipped with the trivial $T$-action,
\item 
To equip a  $T$-action on ${\mathcal O}(1)$, we use the expression
\[
{\mathcal O}(1) = \bigl(({\mathbb C}^2 \setminus \{0 \}) 
\times {\mathbb C} \bigr)/{\mathbb C}^{\ast}, 
\quad 
(x_0,x_1,s) ~\sim~(\lambda x_0, \lambda x_1, \lambda s) 
~\qquad~ (\lambda \in {\mathbb C}^{\ast}),
\]
and denote by $[x_0:x_1:s]$ the homogeneous coordinate of ${\mathcal O}(1)$.
We introduce a $T$-action on ${\mathcal O}(1)$ by 
\[
T \times {\mathcal O}(1) \to {\mathcal O}(1), ~~~~\quad~~~~
(t_0, t_1) \cdot [x_0:x_1:s] = [t_0^{-1} x_0: t_1^{-1} x_1 : t_0 s]
\]
which gives a $T$-equivariant structure on ${\mathcal O}(1)$. 
\end{itemize}

Our goal is to compute the $\hbar$-modified 
$T$-equivariant Euler pairing of ${\mathcal O}$ and ${\mathcal O}(1)$:
\[
{\chi}_T^{\hbar}({\mathcal O}, {\mathcal O}(1)) 
= \sum_{i} (-1)^i {\rm ch}_{T}^{{\hbar}}([H^i({\mathbb C}{\bf P}^1, {\mathcal O}(1))]), 
\]
where the right hand-side takes value in   
${\mathbb Z}[T] = {\mathbb Z}[{\rm e}^{\pm 2\pi \mathrm{i} w_0/\hbar}, 
{\rm e}^{\pm 2\pi \mathrm{i} w_1/\hbar}]$, 
and ${\rm ch}_T^{\hbar}$ is the $\hbar$-modified Chern character map 
(introduced in \cite[Section 3.1]{CIJ14}) from the $T$-representation ring 
$R[T]$ to ${\mathbb Z}[T]$. 
Note that $R[T] = {\mathbb Z}[{\rm e}^{\pm \mu_0}, {\rm e}^{\pm \mu_1}]$
(which can be identified with $T$-equivariant $K$-group 
$K_T^0({\rm pt})$ of a point) consists of the class of 
irreducible representations of $T$; the symbol 
${\rm e}^{m \mu_0 + n \mu_1}$ for $(m,n) \in {\mathbb Z}^2$ 
is the class of the representation spanned by a weight vector $v$ 
satisfying $(t_0,t_1) \cdot v = t_0^{m} \, t_1^n \, v$ for $(t_0, t_1) \in T$. 
The map ${\rm ch}_T^{\hbar}$ sends 
${\rm e}^{m \mu_0 + n \mu_1} \in R[T]$ to 
${\rm e}^{2\pi \mathrm{i} (m w_0+n w_1)/\hbar} \in {\mathbb Z}[T]$.

There are two independent global sections 
\[
s_i : {\mathbb C}{\bf P}^1 \to {\mathcal O}(1), \quad 
[x_0:x_1] \mapsto [x_0:x_1:x_i] \quad (i=0,1).
\] 
of ${\mathcal O}(1)$.
The $T$-actions on these global sections are given by
\[
(t_0, t_1) \cdot s_{0} = s_0, \quad 
(t_0, t_1) \cdot s_{1} = t_0 \, t_1^{-1} \, s_1 
\qquad (t_0, t_1) \in T.
\]
Therefore, the global section $s_0$ (resp. $s_1$) 
gives a weight vector of the weight $(0,0)$ (resp. weight $(1,-1)$), 
and hence
\begin{equation} \label{eq:weight-decomp}
[H^0({\mathbb C}{\bf P}^1, {\mathcal O}(1))] = 
[\Gamma({\mathbb C}{\bf P}^1, {\mathcal O}(1))] = 
1 + {\rm e}^{\mu_0 - \mu_1}
\end{equation}
holds in $R[T]$.  
The weight decomposition \eqref{eq:weight-decomp} 
(together with the fact that 
$H^i({\mathbb C}{\bf P}^1, {\mathcal O}(1)) = 0$ for $i \ne 0$)
implies the following:

\begin{prop}
The $\hbar$-modified $T$-equivariant Euler pairing of ${\mathcal O}$ and ${\mathcal O}(1)$ 
(regarded as $T$-equivariant coherent sheaves as above) is given by 
\begin{equation}
\chi^{\hbar}_{T}({\mathcal O}, {\mathcal O}(1)) = 
1 + {\rm e}^{2\pi \mathrm{i} (w_0 - w_1)/ \hbar}, 
\end{equation}
and this coincides with the Stokes multiplier $1 + {\rm e}^{V_{\gamma_0}}$
in the total Stokes matrix at $x_3$. 
\end{prop}

Thus, we conclude that the Stokes multiplier of 
the quantum differential equation  \eqref{eq:equiv-P1}  
for the equivariant ${\mathbb C}{\bf P}^1$ model can be identified 
with the equivariant Euler pairing of equivariant coherent sheaves 
on ${\mathbb C}{\bf P}^1$. 
This observation was pointed to the authors by Iritani.
We expect that similar coincidence 
(between Stokes multiplies and equivariant Euler pairings)
hold for a wider class of quantum differential equations for 
equivariant target spaces. We also expect that the coincidence 
follows from the categorical equivalence discussed in \cite{Fang,FLTZ} etc.

\appendix
\section{GKZ curve from the $J$-function}\label{app:J_function}
In this appendix we will discuss a heuristic derivation of the GKZ curves from the $J$-functions for the projective space and complete intersections. For this purpose, we will introduce the \textit{on-shell equivariant $J$-function}.\footnote{
The name \textit{on-shell} is inherited from
the (``on-shell'') vortex partition function (\ref{vortex_glsm}).
Indeed for the choice $p=w_0$, $\mathcal{J}_X(x)$ agrees with $Z_{\textrm{vortex}}^{X}(x)$.
}
Let $X$ be the projective space $\mathbb{C}\textbf{P}_{\bm{w}}^{N-1}$ 
or the complete intersection $X=X_{\bm{l};\bm{w},\bm{\lambda}}$ in $\mathbb{C}\textbf{P}^{N-1}$. 
Let $\ell:H^*_T(\mathbb{C}\textbf{P}^{N-1}) \to \mathbb{C}$ be a $\mathbb{C}$-linear map. 
For the equivariant $J$-function $J_{X}(x)$, the composite map 
$(\ell \circ J_{X})(x)$ is a 
$\mathbb{C}$-valued function satisfying GKZ equation. 
We assume that $\ell$ is a $\mathbb{C}$-algebra homomorphism. 
Then by 
\begin{align}
H^*_T(\mathbb{C}\textbf{P}^{N-1}) \cong \mathbb{C}[p]/(\prod_{i=0}^{N-1}(p-w_i)), 
\end{align}
$\ell(p)$ must be one of equivariant parameters $w_i$ $(i=0,\ldots,N-1)$. 
\begin{defi}\label{def:calJ}
The function $\mathcal{J}_X(x)$ is named \textit{on-shell equivariant $J$-function}, 
if the second equivariant cohomology element $p\in H_T^{2}(X)$ in 
Proposition \ref{prop:J_function} is replaced with one of equivariant parameters $w_i$ 
$(i=0,\ldots,N-1)$:
\begin{align}
\mathcal{J}_X(x)=J_{X}(x)|_{p=w_i}. 
\end{align}
\end{defi}
By the construction, we have the following lemma: 
\begin{lemm}\label{lemm:GKZ_calJ}
The on-shell equivariant $J$-function obeys the GKZ equation.
\begin{align}
\widehat{A}_X(\widehat{x},\widehat{y})\mathcal{J}_X(x)=0.
\end{align}
\end{lemm}
We remark that if $w_i \neq w_j$ for $i\neq j$, 
then $\mathbb{C}[p]/(\prod_{i=0}^{N-1}(p-w_i)) \cong
(\mathbb{C}[p]/(p-w_0))
\times \cdots \times 
(\mathbb{C}[p]/(p-w_{N-1}))$ 
is a product of $N$-copies of the $\mathbb{C}$-algebra $\mathbb{C}$. 
Then the $\mathbb{C}$-algebra homomorphisms 
$H^*_T(\mathbb{C}\textbf{P}^{N-1}) \to \mathbb{C}, p \mapsto w_i$ 
give a $\mathbb{C}$-basis of the space of $\mathbb{C}$-linear maps: 
$H^*_T(\mathbb{C}\textbf{P}^{N-1}) \to \mathbb{C}$. 
Thus the on-shell $J$-functions give 
basis of the solution of GKZ equation.

Now we will consider the asymptotic expansion of the on-shell equivariant $J$-function.
\begin{align}
\mathcal{J}_X(x)\sim\mathrm{exp}\left(\sum_{m=0}^{\infty}\hbar^{m-1}S_m(x)\right).
\label{WKB_J}
\end{align}
Using  (\ref{classical_A}) and this asymptotic expansion we find the defining equation of the GKZ equation $A_X(x,y)=0$ in $(x,y)\in\mathbb{C}^*\times\mathbb{C}$ from the relation:
\begin{align}
y=x\frac{dS_0(x)}{dx}\in\mathbb{C}.
\label{sol_A}
\end{align}
In the following we will evaluate the saddle point value $S_0(x)$ of the on-shell equivariant $J$-function $\mathcal{J}_X(x)$ for the projective space $X=\mathbb{C}\textbf{P}_{\bm{w}}^{N-1}$ and the smooth Fano complete intersection $X_{\bm{l};\bm{w},\bm{\lambda}}$ in a heuristic way. 
As a consequence we will show that the defining equation $A_X(x,y)=0$ of the GKZ curve is obtained for these two cases of $X$.

\vspace{0.3cm}
\noindent{(1)} Projective space $\mathbb{C}\textbf{P}_{\bm{w}}^{N-1}$:\\
Let $p$ denote one of equivariant parameters $w_i$ $(i=0,\ldots,N-1)$.
We focus on the factor $\prod_{m=1}^{d}(p-w+m\hbar)$ to find the saddle point value of the on-shell equivariant $J$-function $\mathcal{J}_{\mathbb{C}\textbf{P}^{N-1}_{\bm{w}}}(x)$ in  (\ref{J_CPN}).
In the $\hbar\to 0$ limit while keeping $d\hbar=z$ finite, we use the Riemann integral as follows:
\begin{align}
&\prod_{m=1}^{d}(p-w+m\hbar)=
\mathrm{exp}\left[\sum_{m=1}^d\log(p-w+m\hbar)\right]
\underset{\substack{\hbar\to 0\\ d\hbar=u:\mathrm{finite}}}{\sim} \mathrm{exp}\left[\frac{1}{\hbar}\int^u_0 du'\,\log(p-w+u')\right]
\nonumber \\
&=\mathrm{exp}\left[\frac{1}{\hbar}
\left((u+p-w)\log(u+p-w)-u-(p-w)\log(p-w)\right)\right].
\label{factor0}
\end{align}
Adopting this factor we can approximate
$\mathcal{J}_{\mathbb{C}\textbf{P}^{N-1}_{\bm{w}}}(x)$ by the integral on $z$ as
\begin{align}
&\mathcal{J}_{\mathbb{C}\textbf{P}^{N-1}_{\bm{w}}}(x)
\underset{\substack{\hbar\to 0}}{\sim} 
\int_{\gamma} du \,\exp\left[
\frac{1}{\hbar}\mathcal{W}_{\mathbb{C}\textbf{P}^{N-1}_{\bm{w}}}(u;x)
\right],
\label{J_int_CPN}\\
&\mathcal{W}_{\mathbb{C}\textbf{P}^{N-1}_{\bm{w}}}(u;x)
=(u+p)\log x-\sum_{i=0}^{N-1}\bigl((u+p-w_i)\log(u+p-w_i)-u-(p-w_i)\log(p-w_i)\bigr),
\nonumber
\end{align}
where we interpret the term $(p-w_i)\log(p-w_i)$ as $0$ when $p=w_i$. 
Here we call $\mathcal{W}_{\mathbb{C}\textbf{P}^{N-1}_{\bm{w}}}(u;x)$ \textit{effective superpotential}, and an analytical continuation can be performed by deforming the integration path $\gamma$ on the complex $u$ plane.\footnote{In \cite{Gukov:2003na,Hikami:2006cv} the similar analysis is discussed for the colored Jones polynomial of the knot in $\mathbb{S}^3$.}
Assuming such analytical continuation, we can approximate the integral (\ref{J_int_CPN})  by the saddle point value in $\hbar\to0$ limit:
\begin{align}
S_0(x)=\mathcal{W}_{\mathbb{C}\textbf{P}^{N-1}_{\bm{w}}}(u_c;x),\qquad  
\frac{\partial\mathcal{W}_{\mathbb{C}\textbf{P}^{N-1}_{\bm{w}}}(u;x)}
{\partial u}\Bigg|_{u=u_c}
=0.
\nonumber
\end{align}
The saddle point condition is then given by
\begin{align}
\frac{\partial \mathcal{W}_{\mathbb{C}\textbf{P}^{N-1}_{\bm{w}}}(u;x)}
{\partial u}\Bigg|_{u=u_c}
=\log x-\sum_{i=0}^{N-1}\log(u_c+p-w_i)
=0,
\label{saddle_c_cpn}
\end{align}
and by  (\ref{sol_A}) one has
\begin{align}
y=x\frac{\partial \mathcal{W}_{\mathbb{C}\textbf{P}^{N-1}_{\bm{w}}}(u;x)}{\partial x}\Bigg|_{u=u_c}
=u_c+p.
\label{y_def_cpn}
\end{align}
By eliminating the variable $u_c$ from these relations (\ref{saddle_c_cpn}) and (\ref{y_def_cpn}), we find a constraint equation
on $(x,y)\in\mathbb{C}^*\times\mathbb{C}$:
\begin{align}
\prod_{i=0}^{N-1}(y-w_i)-x=0,
\end{align}
which agrees with the defining equation $A_{\mathbb{C}\textbf{P}^{N-1}_{\bm{w}}}(x,y) = 0$ 
of the GKZ curve (\ref{A_CPN}).

\noindent{(2)} Complete intersection  $X_{\bm{l};\bm{w},\bm{\lambda}}$ in $\mathbb{C}\textbf{P}_{\bm{w}}^{N-1}$:\\
Adopting the similar approximation for the factor (\ref{factor0}) in  $\mathcal{J}_{\mathbb{C}\textbf{P}^{N-1}_{\bm{w}}}(x)$,
we obtain the effective superpotential $\mathcal{W}_{X_{\bm{l};\bm{w},\bm{\lambda}}}(u;x)$
for the on-shell equivariant $J$-function $\mathcal{J}_{X_{\bm{l};\bm{w},\bm{\lambda}}}(x)$ in (\ref{J_comp}) as
\begin{align}
&\mathcal{J}_{X_{\bm{l};\bm{w},\bm{\lambda}}}(x)\underset{\substack{\hbar\to 0}}{\sim} 
\int_{\gamma}du\,\exp\left[\mathcal{W}_{X_{\bm{l};\bm{w},\bm{\lambda}}}(u;x)\right],
\\
&\mathcal{W}_{X_{\bm{l};\bm{w},\bm{\lambda}}}(u;x)=(u+p)\log x
-\sum_{i=0}^{N-1}\bigl((u+p-w_i)\log(u+p-w_i)-u-(p-w_i)\log(p-w_i)\bigr)
\nonumber \\
&\quad\quad\quad\quad\quad\quad\quad
+\sum_{a=1}^n\bigl((l_a u+l_a p-\lambda_a)
\log(l_a u+l_a p-\lambda_a)-l_a u - 
(l_ap-\lambda_a)\log(l_ap-\lambda_a)\bigr),
\nonumber
\end{align}
where $p$ denotes one of equivariant parameters $w_i$ $(i=0,\ldots,N-1)$ 
and we interpret the term $(p-w_i)\log(p-w_i)$ as $0$ when $p=w_i$. 
The saddle point condition is then given by
\begin{align}
\frac{\partial\mathcal{W}_{X_{\bm{l};\bm{w},\bm{\lambda}}}(u;x)}
{\partial u}\Bigg|_{u=u_c}
=\log x-\sum_{i=0}^{N-1}\log(u_c+p-w_i) 
- \sum_{a=1}^{n}\log(l_a u_c+l_ap-\lambda_a)^{-l_a}
=0,
\label{saddle_c_cicf}
\end{align}
and by  (\ref{sol_A}) one has
\begin{align}
y=x\frac{\partial \mathcal{W}_{X_{\bm{l};\bm{w},\bm{\lambda}}}(u;x)}{\partial x}\Bigg|_{u=u_c}
=u_c+p.
\label{y_def_cicf}
\end{align}
As a result of the elimination of the variable $u_c$ from these relations (\ref{saddle_c_cicf}) and (\ref{y_def_cicf}), we find a constraint equation on $(x,y)\in\mathbb{C}^*\times\mathbb{C}$:
\begin{align}
\prod_{i=0}^{N-1}(y-w_i)-x\prod_{a=1}^n(l_ay-\lambda_a)^{l_a}=0,
\end{align}
which agrees with the defining equation $A_{X_{\bm{l};\bm{w},\bm{\lambda}}}(x,y) = 0$ 
of the GKZ curve (\ref{A_comp}).

\section{GKZ equations for oscillatory integrals}

In this appendix we will give a proof of Propositions \ref{prop:IGKZ} and 
\ref{prop:behavior-of-saddle-point-approximation}.

\subsection{Proof of Propositions \ref{prop:IGKZ}} \label{app:IGKZ}
We will show that the oscillatory integral $\mathcal{I}_X(x)$ satisfies the GKZ equation for the projective space $X=\mathbb{C}\textbf{P}^{N-1}_{\bm{w}}$ and the Fano complete intersection $X_{\bm{l};\bm{w},\bm{\lambda}}$ separately.

\begin{prop}\label{prop:ICKZ_CPN}
The oscillatory integral $\mathcal{I}_{\mathbb{C}\textbf{P}^{N-1}_{\bm{w}}}(x)$ in  (\ref{J_CPN_osc}) satisfies the GKZ equation (\ref{GKZ_CPN}).
\end{prop}
\begin{proof}
Act a differential operator $\left(\hbar x\frac{d}{dx}-w_i\right)\left(\hbar x\frac{d}{dx}-w_0\right)$ 
on the oscillatory integral $\mathcal{I}_{\mathbb{C}\textbf{P}^{N-1}_{\bm{w}}}(x)$,
\begin{align}
&\left(\hbar x\frac{d}{dx}-w_i\right)\left(\hbar x\frac{d}{dx}-w_0\right)\mathcal{I}_{\mathbb{C}\textbf{P}^{N-1}_{\bm{w}}}(x)
\nonumber \\
&=\left(\hbar x\frac{d}{dx}-w_i\right)\int_{\Gamma}\prod_{i=1}^{N-1}du_i\,\frac{x}{(u_1\cdots u_{N-1})^2}
\,\mathrm{e}^{\frac{1}{\hbar}W_{\mathbb{C}\textbf{P}^{N-1}_{\bm{w}}}(u_1,\ldots,u_{N-1};x)}
\nonumber \\
&=\int_{\Gamma}\prod_{i=1}^{N-1}du_i\,\frac{x}{(u_1\cdots u_{N-1})^2}
\left(\frac{x}{u_1\cdots u_{N-1}}+w_0-w_i+\hbar\right)
\mathrm{e}^{\frac{1}{\hbar}W_{\mathbb{C}\textbf{P}^{N-1}_{\bm{w}}}(u_1,\ldots,u_{N-1};x)},
\nonumber
\end{align}
where $W_{\mathbb{C}\textbf{P}^{N-1}_{\bm{w}}}$ is the Landau-Ginzburg potential given in  (\ref{LG_CPN}).
To manipulate further we will use the following integration by parts:
\begin{align}
0&=\int_{\Gamma}\prod_{i=1}^{N-1}du_i\,
\hbar\frac{d}{du_i}\left(\frac{1}{u_i}
\,\mathrm{e}^{\frac{1}{\hbar}W_{\mathbb{C}\textbf{P}^{N-1}_{\bm{w}}}(u_1,\ldots,u_{N-1};x)}\right)
\nonumber \\
&=\int_{\Gamma}\prod_{i=1}^{N-1}du_i\,\frac{1}{u_i^2}\Biggl[
\left(u_i+w_i-w_0-\hbar-\frac{x}{u_1\cdots u_{N-1}}\right)\Biggr]
\mathrm{e}^{\frac{1}{\hbar}W_{\mathbb{C}\textbf{P}^{N-1}_{\bm{w}}}(u_1,\ldots,u_{N-1};x)}.
\nonumber
\end{align}
In this  computation, any boundary contributions do not appear,
because the image of the Lefschetz thimble $\Gamma$ is a relative cycle starting from a non-degenerate critical point $p_{\mathrm{crit}}$ of the Landau-Ginzburg potential $W_X$ to the infinity. Then one finds that
\begin{align}
\left(\hbar x\frac{d}{dx}-w_i\right)\left(\hbar x\frac{d}{dx}-w_0\right)\mathcal{I}_{\mathbb{C}\textbf{P}^{N-1}_{\bm{w}}}(x)
=\int_{\Gamma}\prod_{i=1}^{N-1}du_i\,\frac{xu_i}{(u_1\cdots u_{N-1})^2}
\,\mathrm{e}^{\frac{1}{\hbar}W_{\mathbb{C}\textbf{P}^{N-1}_{\bm{w}}}(u_1,\ldots,u_{N-1};x)}.
\nonumber
\end{align}
Repeating the above manipulations for $\left[\prod_{i=1}^{N-1}\left(\hbar x\frac{d}{dx}-w_i\right)\right]\left(\hbar x\frac{d}{dx}-w_0\right)$, the following relation is obtained:
\begin{align}
\begin{split}
&\left[\prod_{i=1}^{N-1}\left(\hbar x\frac{d}{dx}-w_i\right)\right]\left(\hbar x\frac{d}{dx}-w_0\right)\mathcal{I}_{\mathbb{C}\textbf{P}^{N-1}_{\bm{w}}}(x)
\\
&=\int_{\Gamma}\prod_{i=1}^{N-1}du_i\,\frac{xu_1\cdots u_{N-1}}{(u_1\cdots u_{N-1})^2}
\,\mathrm{e}^{\frac{1}{\hbar}W_{\mathbb{C}\textbf{P}^{N-1}_{\bm{w}}}(u_1,\ldots,u_{N-1};x)}
=x\mathcal{I}_{\mathbb{C}\textbf{P}^{N-1}_{\bm{w}}}(x).
\label{IGKZW_CPN}
\end{split}
\end{align}
This differential equation  is the same as the GKZ equation (\ref{GKZ_CPN}) for the $J$-function $J_{\mathbb{C}\textbf{P}^{N-1}_{\bm{w}}}$.
\end{proof}

\begin{prop}\label{prop:ICKZ_comp}
The oscillatory integral 
$\mathcal{I}_{X_{\bm{l};\bm{w},\bm{\lambda}}}(x)$  in  (\ref{J_comp_osc})  satisfies the GKZ equation (\ref{GKZ_CPN}).
\end{prop}
\begin{proof}

Consider the GKZ equation (\ref{IGKZW_CPN}) for the oscillatory integral $\mathcal{I}_{\mathbb{C}\textbf{P}^{N-1}_{\bm{w}}}(x)$ for the mirror Landau-Ginzburg model of  the projective space $X=\mathbb{C}\textbf{P}^{N-1}_{\bm{w}}$ denoted by
\begin{align}
0&=\widehat{A}_{\mathbb{C}\textbf{P}^{N-1}_{\bm{w}}}\left(\widehat{x},\widehat{y}\right)\mathcal{I}_{\mathbb{C}\textbf{P}^{N-1}_{\bm{w}}}(x)
=\prod_{i=0}^{N-1}\left(\widehat{y}-w_i\right)\mathcal{I}_{\mathbb{C}\textbf{P}^{N-1}_{\bm{w}}}(x)-\widehat{x}\mathcal{I}_{\mathbb{C}\textbf{P}^{N-1}_{\bm{w}}}(x),
\nonumber
\end{align}
where $\widehat{x}$ (resp. $\widehat{y}$) acts on $\mathcal{I}_{\mathbb{C}\textbf{P}^{N-1}_{\bm{w}}}(x)$ as $x$ (resp. $\hbar x d/dx$). 
Perform the Laplace transformation of this differential equation:
\begin{align}
0&=\int_0^{\infty}dv_1\cdots\int_0^{\infty}dv_n\, \mathrm{e}^{-\frac{\sum_{i=1}^n(v_i+\lambda_i\log v_i)}{\hbar}}
\widehat{A}_{\mathbb{C}\textbf{P}^{N-1}_{\bm{w}}}\left(v_1^{l_1}\cdots v_n^{l_n}\widehat{x},\widehat{y}\right)\mathcal{I}_{\mathbb{C}\textbf{P}^{N-1}_{\bm{w}}}(v_1^{l_1}\cdots v_n^{l_n}x)
\nonumber \\
&=\int_0^{\infty}dv_1\cdots\int_0^{\infty}dv_n\, \mathrm{e}^{-\frac{\sum_{a=1}^n(v_a+\lambda_a\log v_a)}{\hbar}}
\prod_{i=0}^{N-1}\left(\hbar x\frac{d}{dx}-w_i\right)
\mathcal{I}_{\mathbb{C}\textbf{P}^{N-1}_{\bm{w}}}(v_1^{l_1}\cdots v_n^{l_n}x)
\nonumber \\
&\quad
-\int_0^{\infty}dv_1\cdots\int_0^{\infty}dv_n\, \mathrm{e}^{-\frac{\sum_{a=1}^n(v_a+\lambda_a\log v_a)}{\hbar}}
v_1^{l_1}\cdots v_n^{l_n}x\mathcal{I}_{\mathbb{C}\textbf{P}^{N-1}_{\bm{w}}}(v_1^{l_1}\cdots v_n^{l_n}x)
\nonumber \\
&=\prod_{i=0}^{N-1}\left(\hbar x\frac{d}{dx}-w_i\right)
\mathcal{I}_{X_{\bm{l};\bm{w},\bm{\lambda}}}(x)
\nonumber \\
&\quad
-x\int_0^{\infty}dv_1\cdots\int_0^{\infty}dv_n\, \mathrm{e}^{-\frac{\sum_{a=1}^n(v_a+\lambda_a\log v_a)}{\hbar}}
v_1^{l_1}\cdots v_n^{l_n}
\mathcal{I}_{\mathbb{C}\textbf{P}^{N-1}_{\bm{w}}}(v_1^{l_1}\cdots v_n^{l_n}x).
\nonumber
\end{align}
To manipulate further we will use the following integration by parts repeatedly for each $v_i$'s:
\begin{align}
0&=\int_0^{\infty}dv_a\,\hbar\frac{d}{dv_a}\left(v_a^{m}\mathrm{e}^{-\frac{v_a+\lambda_a\log v_a}{\hbar}}
\mathcal{I}_{\mathbb{C}\textbf{P}^{N-1}_{\bm{w}}}(v_1^{l_1}\cdots v_n^{l_n}x)
\right)
\nonumber \\
&=\int_0^{\infty}dv_a\left(-v_a-\lambda_a+m\hbar\right)v_a^{m-1}\mathrm{e}^{-\frac{v_a+\lambda_a\log v_a}{\hbar}}
\mathcal{I}_{\mathbb{C}\textbf{P}^{N-1}_{\bm{w}}}(v_1^{l_1}\cdots v_n^{l_n}x)
\nonumber \\
&\quad+\int_0^{\infty}dv_av_a^{m}\mathrm{e}^{-\frac{v_a+\lambda_a\log v_a}{\hbar}}
\hbar l_a\frac{x}{v_a}\frac{d}{dx}\mathcal{I}_{\mathbb{C}\textbf{P}^{N-1}_{\bm{w}}}(v_1^{l_1}\cdots v_n^{l_n}x)
\nonumber \\
&=-\int_0^{\infty}dv_a\,v_a^{m}\mathrm{e}^{-\frac{v_a+\lambda_a\log v_a}{\hbar}}
\mathcal{I}_{\mathbb{C}\textbf{P}^{N-1}_{\bm{w}}}(v_1^{l_1}\cdots v_n^{l_n}x)
\nonumber \\
&\quad+\left(l_a\hbar x\frac{d}{dx}-\lambda_a+m\hbar\right)\int_0^{\infty}dv_a\,v_a^{m-1}\mathrm{e}^{-\frac{v_a+\lambda_a\log v_a}{\hbar}}
\mathcal{I}_{\mathbb{C}\textbf{P}^{N-1}_{\bm{w}}}(v_1^{l_1}\cdots v_n^{l_n}x).
\nonumber
\end{align}
Then one finds the GKZ equation (\ref{GKZ_comp}) for the $J$-function $J_{X_{\bm{l};\bm{w},\bm{\lambda}}}(x)$:
\begin{align}
\left[
\prod_{i=0}^{N-1}\left(\hbar x\frac{d}{dx}-w_i\right)-x\prod_{a=1}^{n}\prod_{m=1}^{l_a}\left(l_a\hbar x\frac{d}{dx}-\lambda_a+m\hbar \right)\right]\mathcal{I}_{X_{\bm{l};\bm{w},\bm{\lambda}}}(x)=0.
\end{align}
\end{proof}

\subsection{Proof of Proposition \ref{prop:behavior-of-saddle-point-approximation}} \label{appendix:asymptotic-behavior-of -coefficients}
Here we investigate the behavior of coefficients when $x \to \infty$ in the saddle point approximation \eqref{eq:saddle-expansion-oscillatory-integral} of the oscillatory integral for 
\[
W_X = \sum_{i=1}^{N-1} (u_i + w_i \log u_i) - 
\sum_{a=1}^{n} (v_a + \lambda_a \log v_a) + 
\frac{v_1 \cdots v_n}{u_1 \cdots u_{N-1}} x 
+ w_0 \log \left( \frac{v_1 \cdots v_n}{u_1 \cdots u_{N-1}} x \right),
\]
which is mirror to $X = X_{\bm{w},\bm{\lambda}}$. 
In this subsection we write $W = W_X$ for simplicity.

\smallskip
To prove (i) in Proposition \ref{prop:behavior-of-saddle-point-approximation}, it is enough to find an asymptotic behavior of second derivatives of $W$:
\begin{eqnarray*}
\frac{\partial^{2} W}{\partial u_{i}^2} 
& = &  
- \frac{w_i - w_0}{u_i^2} + \frac{2}{u_i^2} \, \frac{v_1 \cdots v_n}{u_1 \cdots u_{N-1}} x \\
\frac{\partial^{2} W}{\partial u_{i} \partial u_j} 
& = &  
\frac{1}{u_i u_j} \, \frac{v_1 \cdots v_n}{u_1 \cdots u_{N-1}} x \quad (i \ne j)\\
\frac{\partial^{2} W}{\partial u_{i} \partial v_a} 
& = &  
- \frac{1}{u_i v_a} \, \frac{v_1 \cdots v_n}{u_1 \cdots u_{N-1}} x  \\
\frac{\partial^{2} W}{\partial v_{a} \partial v_b} 
& = &  
\frac{1}{v_a v_b} \, \frac{v_1 \cdots v_n}{u_1 \cdots u_{N-1}} x \quad (a \ne b)\\
\frac{\partial^{2} W}{\partial v_{a}^2} 
& = &  
\frac{\lambda_a - w_0}{v_a^2}. 
\end{eqnarray*}
At a critical point $({\bm u}^{\rm (c)},{\bm v}^{\rm (c)})$, we can use
\[
\left( \frac{v_1 \cdots v_n}{u_1 \cdots u_{N-1}} x \right) 
\biggl|_{({\bm u},{\bm v}) = ({\bm u}^{\rm (c)},{\bm v}^{\rm (c)})} 
= u_i^{\rm (c)} + w_i - w_0
= v_a^{\rm (c)} + \lambda_a - w_0.
\]
This behaves as $O(x^{\frac{1}{N-n}})$ in the case of \eqref{eq:critical-bahavior-1}, and as $\lambda_b - w_0 + O(x^{-1})$ in the case of \eqref{eq:critical-bahavior-2}.
Therefore, we can find the behavior 
\begin{eqnarray*}
{\rm Hess}({\bm u}^{\rm (c)},{\bm v}^{\rm (c)}) = 
\begin{cases} 
O(x^{-\frac{N+n-1}{N-n}}) & \text{in the case of \eqref{eq:critical-bahavior-1}}, \\[+.5em]
O(x^{2}) & \text{in the case of \eqref{eq:critical-bahavior-2}}
\end{cases}
\end{eqnarray*}
of the Hessian when $x \to \infty$. 
The claim of (i) in Proposition \ref{prop:behavior-of-saddle-point-approximation} follows immediately from this computation.

\smallskip
Let us prove (ii) in Proposition \ref{prop:behavior-of-saddle-point-approximation}. We take a coordinate at a critical point $({\bm u}^{\rm (c)}, {\bm v}^{\rm (c)})$ of $W$:
\[
{\bm \xi} = (\xi_1, \dots, \xi_{N+n-1}) 
= (u_1 - u_1^{\rm (c)},\dots,u_{N-1} - u_{N-1}^{\rm (c)}, v_1 - v_1^{\rm (c)}, \dots, v_n - v_n^{\rm (c)}),
\]
and consider the Taylor expansion of $W$ at the critical point ${\bm \xi}^{\rm (c)} = {\bm 0}$:
\begin{multline}
W({\bm \xi};x) = W({\bm \xi}^{\rm (c)};x) 
+ \frac{1}{2!}\sum_{i,j} 
W_{ij}({\bm \xi}^{\rm (c)};x) \, \xi_i  \xi_j  
+ \sum_{m \ge 3} \frac{1}{m!} \sum_{i_1, \dots, i_m} 
W_{i_1 \dots i_m}({\bm \xi}^{\rm (c)};x) \,
\xi_{i_1}  \cdots \xi_{i_m},
\end{multline}
where $W_{i_1 \dots i_m} = (\partial^m W)/(\partial \xi_{i_1} \cdots \partial \xi_{i_m})$. Since we have chosen generic $w_i$ and $\lambda_a$ so that the critical points are non-degenerate, the Hesse matrix $H$ has non-zero determinant at ${\bm \xi}^{\rm (c)}$. We further take a linear transformation 
$\xi_i = (- \hbar)^{1/2}\sum_j C_{i}^j(x) s_j$ of the coordinates which transforms the quadratic part of $W$ as 
\begin{equation} \label{eq:holomorphic-Morse-lemma}
\frac{1}{2 \hbar}\sum_{i,j} 
W_{ij}({\bm \xi}^{\rm (c)};x)\, \xi_i \xi_j = - \frac{1}{2} \sum_{i=1}^{N+n-1} s_i^2.
\end{equation}
We can find these coefficients $C_{i}^{j}(x)$ by applying the simultaneous completing the square to the quadratic form in the left hand side of \eqref{eq:holomorphic-Morse-lemma}. Eventually we can find the behavior of the coefficients $C_{i}^{j}(x)$ for large $x$ as follows:
\begin{itemize}
\item 
Let us consider the case when the critical point $({\bm u}^{\rm (c)},{\bm v}^{\rm (c)})$ behaves as \eqref{eq:critical-bahavior-1} in Lemma \ref{lemm:asymptotic-critical-pt}. Since $W_{ij}({\bm \xi}^{\rm (c)};x)$ behaves as $O(x^{- \frac{1}{N-n}})$ in this case, we can show that
\begin{equation} \label{eq:behavior-C-type1}
C_{i}^{j}(x) = O(x^{\frac{1}{2(N-n)}})
\end{equation}
holds for all $i,j=1,\dots, N+n-1$, when $x \to \infty$. 
\item
Let us consider the case when the critical point $({\bm u}^{\rm (c)},{\bm v}^{\rm (c)})$ behaves as \eqref{eq:critical-bahavior-2} in Lemma \ref{lemm:asymptotic-critical-pt}. Let $i_b \in \{N,\dots,N+n-1 \}$ be the label of $v_b$; that is, $\xi_{i_b} = v_b - v_b^{\rm (c)} $. Since we can arrange the quadratic part of $W$ as
\begin{multline*}
W_{i_b i_b}({\bm \xi}^{\rm (c)};x) \xi_{i_b}^2
+ 2\sum_{i \ne i_b} 
W_{i \, i_b}({\bm \xi}^{\rm (c)};x) \xi_{i_b} \xi_i 
+ \sum_{i,j \ne i_b} 
W_{ij}({\bm \xi}^{\rm (c)};x) \xi_{i} \xi_j  \\
= 
W_{i_b i_b}({\bm \xi}^{\rm (c)};x) \left( \xi_{i_b} + \sum_{i \ne i_b} \frac{W_{i\, i_b}({\bm \xi}^{\rm (c)};x)}{W_{i_b i_b}({\bm \xi}^{\rm (c)};x)} \xi_i  \right)^2 \\
- \frac{1}{W_{i_b i_b}({\bm \xi}^{\rm (c)};x)} \left(\sum_{i \ne i_b} W_{i\, i_b}({\bm \xi}^{\rm (c)};x) \xi_i\right)^2  + \sum_{i,j \ne i_b} W_{ij}({\bm \xi}^{\rm (c)};x) \xi_{i} \xi_j.
\end{multline*}
Thus we can choose 
\[
s_{i_b} = W_{i_b i_b}^{1/2}({\bm \xi}^{\rm (c)};x) \left( \xi_{i_b} + \sum_{i \ne i_b} \frac{W_{i\, i_b}({\bm \xi}^{\rm (c)};x)}{W_{i_b i_b}({\bm \xi}^{\rm (c)};x)} \xi_i \right)
\]
as one of new coordinates, and other $s_i$'s are written in terms of $\xi_i$'s except for $\xi_{i_b}$. Since $W_{i_b i_b}({\bm \xi}^{\rm (c)};x) = O(x^2)$, $ W_{i\, i_b}({\bm \xi}^{\rm (c)};x) = O(x)$ if $i \ne i_b$ and $W_{ij}({\bm \xi}^{\rm (c)};x)=O(1)$ for $i,j \ne i_b$, we can conclude 
\begin{eqnarray} \label{eq:behavior-C-type2}
C_{i}^{j}(x) = \begin{cases}
 O(x^{-1})  & \text{if $i = i_b$} \\
 O(1) & \text{otherwise}
\end{cases}
\end{eqnarray}
hold when $x \to \infty$.
\end{itemize}

Let us proceed the computation of saddle point expansion. The above change of the coordinate yields 
\[
d\xi_1 \cdots d\xi_{N+n-1} = \frac{(-\hbar)^{\frac{N+n-1}{2}}}{\sqrt{{\rm Hess}({\bm u}^{\rm (c)},{\bm v}^{\rm (c)})}} \, ds_1\cdots ds_{N+n-1}.
\]
Then, by the standard argument of the saddle point method, the asymptotic expansion of the oscillatory integral is computed by term-wise integration:
\begin{multline} \label{eq:saddle-point-formula-pre}
{\mathcal I}_{X}(x) \sim \exp\left( \frac{1}{\hbar} W({\bm \xi}^{\rm (c)};x) \right) \, \frac{(-\hbar)^{\frac{N+n-1}{2}}}{u^{\rm (c)}_1 \cdots u^{\rm (c)}_{N-1} \,\sqrt{{\rm Hess}({\bm \xi}^{\rm (c)})}}  \\ \times 
\left(1 + \sum_{m = 1}^{\infty} \hbar^{\frac{m}{2}} \sum_{i_1, \dots, i_{m}} f_{i_1 \dots i_m} \int_{{\mathbb R}^{N+n-1}} ds_1 \cdots ds_{N+n-1} \exp\left( - \frac{1}{2} \sum_{i=1}^{N+n-1} s_i^2 \right)  s_{i_1} \cdots s_{i_m} \right), 
\end{multline}
where $f_{i_1 \dots i_m}$ is the Taylor coefficient given by  
\begin{multline*}
\frac{u^{\rm (c)}_1 \cdots u^{\rm (c)}_{N-1}}{u_1 \cdots u_{N-1}}\exp\left( \frac{1}{\hbar} \sum_{m \ge 3} \frac{1}{m!} \sum_{i_1, \dots, i_m} W_{i_1 \dots i_m}({\bm \xi}^{\rm (c)};x) \xi_{i_1}  \cdots \xi_{i_m} \right) \Biggl|_{\xi_i = (- \hbar)^{1/2}\sum_j C_{i}^j(x) s_j} \\ 
= 
1 + \sum_{m = 1}^{\infty} \hbar^{\frac{m}{2}} \sum_{i_1, \dots, i_{m}} f_{i_1 \dots i_m} s_{i_1} \cdots s_{i_m}.
\end{multline*}

\begin{itemize}
\item 
If $({\bm u}^{\rm (c)},{\bm v}^{\rm (c)})$ behaves as \eqref{eq:critical-bahavior-1}, then $W_{i_1 \dots i_m}({\bm \xi}^{\rm (c)};x) = O(x^{- \frac{m-1}{N-n}})$, and hence, 
\[
W_{i_1 \dots i_m}({\bm \xi}^{\rm (c)};x) \, C_{i_1}^{j_1} \cdots C_{i_m}^{j_m} = 
O(x^{- \frac{m-2}{2(N-n)}}) \quad
\text{for any $j_1, \dots, j_m$}.
\]
For $m \ge 3$, this tends to $0$ when $x \to \infty$.

\item 
If $({\bm u}^{\rm (c)},{\bm v}^{\rm (c)})$ behaves as \eqref{eq:critical-bahavior-2}, then $W_{i_1 \dots i_m}({\bm \xi}^{\rm (c)};x) = O(x^{\ell_b})$, where $\ell_b$ is the number of $i_b$ in the indices $i_1, \dots, i_m$. Therefore, 
\[
W_{i_1 \dots i_m}({\bm \xi}^{\rm (c)};x) \, C_{i_1}^{j_1} \cdots C_{i_m}^{j_m} = 
O(1) \quad \text{for any $j_1, \dots, j_m$}.
\]
\end{itemize}
We can also verify that
\[
\frac{u_i^{\rm (c)}}{u_i} = \left( 1 + (-\hbar)^{1/2} \sum_{i} \frac{C_i^{j}}{u_i^{\rm (c)}} s_j \right)^{-1},
\]
and the coefficient satisfies 
\begin{eqnarray*}
\frac{C_i^{j}}{u_i^{\rm (c)}} = 
\begin{cases}
O(x^{- \frac{1}{2(N-n)}}) &  \text{for the case of \eqref{eq:critical-bahavior-1}} \\
O(1) &  \text{for the case of \eqref{eq:critical-bahavior-2}} 
\end{cases}
\end{eqnarray*}
when $x \to \infty$. Therefore, we can conclude that 
\begin{eqnarray} \label{eq:limit-bahavior-saddle-expansion-pre}
f_{i_1 \dots i_m} = 
\begin{cases} 
O(x^{- \frac{1}{2(N-n)}}) & \text{for the case of \eqref{eq:critical-bahavior-1}} \\
O(1) & \text{for the case of \eqref{eq:critical-bahavior-2}}
\end{cases}
\end{eqnarray}
for $m \ge 1$. After evaluating the Gaussian integrals in \eqref{eq:saddle-point-formula-pre} by using 
\begin{eqnarray*}
\int_{\mathbb R} ds_i\,  e^{- \frac{1}{2} s_i^2} \, s_i^{k}  = 
\begin{cases} 
0 & \text{if $k$ is odd}, \\[+.2em]
\displaystyle 
\sqrt{2\pi} \, {(k-1)!!}
& \text{if $k$ is even},
\end{cases}
\end{eqnarray*}
we obtain the saddle point approximation \eqref{eq:saddle-expansion-oscillatory-integral}. In particular, \eqref{eq:limit-bahavior-saddle-expansion-pre} proves  the claim (ii) in Proposition \ref{prop:behavior-of-saddle-point-approximation}.

\section{Computational results by iteration and topological recursion}
\label{section:computational-results}

In this appendix we will firstly give some explicit computational results of the WKB solutions to the GKZ equations. In Appendix \ref{app_sub:top_rec_p1} we will explicitly perform the WKB reconstruction (\ref{wave_function}) for the equivariant ${\IC}\textbf{P}^1$ model, and see agreements with the results in Appendix \ref{app_sub:wkb_J}.

\subsection{Some iterative computations for the GKZ equation}\label{app_sub:wkb_J}

Assume the saddle point approximation of the oscillatory integral
\begin{align}
\mathcal{I}_X(x)\sim\exp\left(\sum_{m=0}^{\infty}\hbar^{m-1}S_m(x)\right),
\nonumber
\end{align}
one finds a set of the first order differential equations for $S_m$'s by expanding the GKZ equation around $\hbar=0$. 

\vspace{0.2cm}
\noindent{\underline{$\mathbb{C}\textbf{P}^{N-1}$ model}}\\
\vspace{0.1cm}\noindent
The GKZ equation for the (non-equivariant) $\mathbb{C}\textbf{P}^{N-1}$ model is
\begin{align}
\left[\left(\hbar x\frac{d}{dx}\right)^{N}-x\right]\mathcal{I}_{\mathbb{C}\textbf{P}^{N-1}}(x)=0.
\label{app_ex1}
\end{align}
Some computational results of $S_m$'s for $N=2, 3, 4$ are listed in table \ref{tab:CPN}.
\begin{table}[t]
\begin{tabular}{|c|c|c|c|}
\hline
& $N=2$ & $N=3$ & $N=4$\\
\hline
$S_0(x)$ & $2x^{1/2}$ & $3x^{1/3}$ & $4x^{1/4}$ \\
\hline
$S_1(x)$ & $\log x^{-1/4}$ & $\log x^{-1/3}$ & $\log x^{-3/8}$ \\
\hline
$S_2(x)$ & $x^{-1/2}/16$ & $x^{-1/3}/9$ & $5x^{-1/4}/32$ \\
\hline 
$S_3(x)$ & $x^{-1}/64$ & $x^{-2/3}/54$ & $5x^{-1/2}/256$ \\
\hline
$S_4(x)$ & $25x^{-3/2}/3072$ & $x^{-1}/243$ & $17x^{-3/4}/24576$ \\
\hline
\end{tabular}
\caption{\label{tab:CPN}$S_m(x)$ for the $\mathbb{C}\textbf{P}^{N-1}$ models ($N=2,3,4$) obtained from the GKZ equation (\ref{app_ex1}).}
\end{table}

\noindent{\underline{Equivariant $\mathbb{C}\textbf{P}^{1}$ model}}\\
\vspace{0.1cm}\noindent
The GKZ equation for the equivariant $\mathbb{C}\textbf{P}^{1}$ model is
\begin{align}
\left[\left(\hbar x\frac{d}{dx}-w_0\right)\left(\hbar x\frac{d}{dx}-w_1\right)-x\right]\mathcal{I}_{\mathbb{C}\textbf{P}^{1}_{\bm{w}}}(x)=0.
\label{app_ex2}
\end{align}
For this model there are two solutions which have the formal power series expansion:
\begin{align}
\mathcal{I}^{(\pm)}_{\mathbb{C}\textbf{P}^{1}_{\bm{w}}}(x)\sim
\exp\left(
\sum_{m=0}^{\infty}\hbar^{m-1}S_m^{(\pm)}(x)
\right).
\nonumber
\end{align}
Computational results of $S^{(\pm)}_m$ for $m=0,1,2,3,4$ are listed in Table \ref{tab:eqv_CP1} modulo constant shifts.
\begin{table}[t]
\begin{tabular}{|c|l|}
\hline
$m$ & $S^{(\pm)}_m(x)$ \\
\hline
$0$ & $\pm\sqrt{4x+(w_0-w_1)^2}$ \\
& 
$+w_0\log\left(-w_0+w_1\pm\sqrt{4x+(w_0-w_1)^2}\right)+w_1\log\left(w_0-w_1\pm\sqrt{4x+(w_0-w_1)^2}\right)$
\\
\hline
$1$ & $-\frac{1}{4}\log\left(\frac{4x+(w_0-w_1)^2}{4}\right)$ \\
\hline
$2$ & $\pm\left(6x-(w_0-w_1)^2\right)/\left(12\left(4x+(w_0-w_1)^2\right)^{3/2}\right)$ \\
\hline
$3$ & $x\left(x-(w_0-w_1)^2\right)/\left(4x+(w_0-w_1)^2\right)^{3}$ \\
\hline
$4$ & $\pm\left(1500 x^3 - 3654 x^2 (w_0 - w_1)^2 + 378 x (w_0 - w_1)^4 + (w_0 - w_1)^6\right)/\left(360\left(4x+(w_0-w_1)^2\right)^{9/2}\right)$ \\
\hline
\end{tabular}
\caption{\label{tab:eqv_CP1}$S_m^{(\pm)}(x)$ for the equivariant $\mathbb{C}\textbf{P}^{1}$ model obtained from the GKZ equation (\ref{app_ex2}).}
\end{table}

\noindent{\underline{Degree 1 hypersurface in $\mathbb{C}\textbf{P}^{1}$}}\\
\vspace{0.1cm}\noindent
The GKZ equation for the degree 1 hypersurface 
$X_{\bm{w},\lambda}=X_{l=1;\bm{w},\lambda}$ in $\mathbb{C}\textbf{P}^{1}$ is
\begin{align}
\left[\left(\hbar x\frac{d}{dx}-w_0\right)\left(\hbar x\frac{d}{dx}-w_1\right)-x\left(\hbar x\frac{d}{dx}-\lambda+\hbar\right)\right]\mathcal{I}_{X_{\bm{w},\lambda}}(x)=0.
\label{app_ex3}
\end{align}
For this model we also find two solutions which have the formal power series expansion:
\begin{align}
\mathcal{I}^{(\pm)}_{X_{\bm{w},\lambda}}(x)
\sim
\exp\left(
\sum_{m=0}^{\infty}\hbar^{m-1}S_m^{(\pm)}(x)
\right).
\nonumber
\end{align}
Computational results of $S^{(\pm)}_m$ for $m=0,1,2,3$ are listed in Table \ref{tab:eqv_deg1} modulo constant shifts.
\begin{table}[t]
\begin{tabular}{|c|l|}
\hline
$m$ & $S^{(\pm)}_m(x)$ \\
\hline
$0$ & $\Biggl(x\pm\sqrt{x^2+2(w_0+w_1-2\lambda)x+(w_0-w_1)^2}\pm (w_0-w_1)\log x+(w_0+w_1)\log x$ \\
& $\pm(w_0+w_1-2\lambda)\log\left[(x+w_0+w_1-2\lambda+\sqrt{x^2+2(w_0+w_1-2\lambda)x+(w_0-w_1)^2})\right]$ \\
& $\mp(w_0-w_1)\log\Bigl[(w_0+w_1-2\lambda)x+(w_0-w_1)^2$ \\
& $\hspace{2.5cm}+(w_0-w_1)\sqrt{x^2+2(w_0+w_1-2\lambda)x+(w_0-w_1)^2}\Bigr]\Biggr)/2$ \\
\hline
$1$ & $-\frac{1}{4}\log\left(x^2+2(w_0+w_1-2\lambda)x+(w_0-w_1)^2\right)$ \\
&
$\pm\frac{1}{2}\log\left(x+w_0+w_1-2\lambda+\sqrt{x^2+2(w_0+w_1-2\lambda)x+(w_0-w_1)^2}\right]\mp \log\sqrt{2}$ \\
\hline
$2$ & $-x/(2(x^2+2(w_0+w_1-2\lambda)x+(w_0-w_1)^2))$\\
&
$\mp 5((w_0+w_1-2\lambda)x+(w_0-w_1)^2))/(12(x^2+2(w_0+w_1-2\lambda)x+(w_0-w_1)^2)^{3/2})$
\\
&$\mp((w_0  + w_1 -2\lambda)x+w_0^2 - 10 w_0 w_1 + w_1^2 + 8 (w_0+w_1)\lambda-8\lambda^2)$
\\
&
$\quad /\Bigl(24(w_0-\lambda)(w_1-\lambda)\sqrt{x^2+2(w_0+w_1-2\lambda)x+(w_0-w_1)^2}\Bigr)$ \\
\hline
$3$ &
$x\Bigl(3x^2+(w_0+w_1-2\lambda)x-2(w_0-w_1)^2\Bigr)$
\\
&
$\times\Bigl(x+w_0+w_1-2\lambda \mp\sqrt{x^2+2(w_0+w_1-2\lambda)x+(w_0-w_1)^2}\Bigr)$ 
\\
& $\quad /\Bigl(4(x^2+2(w_0+w_1-2\lambda)x+(w_0-w_1)^2)^3\Bigr)$ \\
\hline
\end{tabular}  
\caption{\label{tab:eqv_deg1}$S_m^{(\pm)}(x)$ for degree 1 hypersurface in $\mathbb{C}\textbf{P}^{1}$ obtained from the GKZ equation (\ref{app_ex3}).}
\end{table}
Especially, focusing on the $S^{(\pm)}_1(x)$ we find the following expansion around $x=\infty$:
\begin{align}
\begin{split}
\mathrm{e}^{S^{(+)}_1(x)}=1+\frac{(w_0-\lambda)(w_1-\lambda)}{2x^2}+O(x^{-3}),
\\
\mathrm{e}^{S^{(-)}_1(x)}=\frac{1}{x}-\frac{w_0+w_1-2\lambda}{x^2}+O(x^{-3}),
\nonumber
\end{split}
\end{align}
and these asymptotic expansions are consistent with Proposition \ref{prop:behavior-of-saddle-point-approximation} (i).


\subsection{Topological recursion for the equivariant $\mathbb{C}\textbf{P}^{1}$ model}\label{app_sub:top_rec_p1}

In the following, for the equivariant $\mathbb{C}\textbf{P}^{1}$ model, we will explicitly recover the computational result in Table \ref{tab:eqv_CP1} by applying the topological recursion (\ref{top_recursion}) in \cite{Eynard:2007kz}. The GKZ curve  
\begin{align}
\Sigma_{\mathbb{C}\textbf{P}^{1}_{\bm{w}}}=\Big\{\; (x,y)\in \mathbb{C}^*\times \mathbb{C}\; \Big|\; (y-w_0)(y-w_1)-x=0\; \Big\}
\nonumber
\end{align}
is parametrized by a local coordinate $z$ as follows:
\begin{align}
x(z)=z^2-\Lambda,\qquad
y(z)=z+\frac{1}{2}(w_0+w_1),\qquad
\Lambda=\frac{1}{4}(w_0-w_1)^2.
\nonumber
\end{align}
The spectral curve $\Sigma_{\mathbb{C}\textbf{P}^{1}_{\bm{w}}}$ has only one simple ramification point at $z=0$ in this local coordinate. Starting from
\begin{align}
\omega_{1}^{(0)}(z)=0,\qquad
\omega_{2}^{(0)}(z_1,z_2)=B(z_1,z_2)=\frac{dz_1dz_2}{(z_1-z_2)^2},
\nonumber
\end{align}
the differentials $\omega_{n}^{(g)}$ for $(g,n) \ne (0,1), (0,2)$ are defined by 
the topological recursion (\ref{top_recursion})
\begin{align}
\begin{split}
&
\omega_{n+1}^{(g)}(z,\bm{z}_N)=
\mathop{\mathrm{Res}}_{w=0}\,
\frac{\int_{-w}^{w}B(\cdot, z)}{2\left(y(w)-y(-w)\right)dx(w)/x(w)}\bigg(\omega_{n+2}^{(g-1)}(w,-w,\bm{z}_N)  \\
&\hspace{9em}
+\sum_{\ell=0}^{g}\sum_{\emptyset=J\subseteq N}\omega_{|J|+1}^{(g-\ell)}(w,\bm{z}_J)\omega_{|N|-|J|+1}^{(\ell)}(-w,\bm{z}_{N \backslash J}) \bigg),
\nonumber
\end{split}
\end{align}
where $N=\{1,2,\ldots,n\}\supset J=\{i_1,i_2,\ldots,i_j\}$, and $N\backslash J=\{i_{j+1},i_{j+2},\ldots,i_n\}$. 
Integrating these multi-differentials, one finds the free energies
\begin{align}
\begin{split}
&
F_{1}^{(0)}(x)=\int^z_{z_*} y(z')\frac{dx(z')}{x(z')},\qquad 
F_{2}^{(0)}(x)=\int^z_{z_*}\int^z_{z_*}\left(B(z_1',z_2')-\frac{dx(z_1')dx(z_2')}{(x(z_1')-x(z_2'))^2}\right),\\
&F_{n}^{(g)}(x)=\int^z_{z_*} \cdots \int^z_{z_*} \omega^{(g)}_{n}(z_1',\ldots,z_n'),\quad
(g,n) \ne (0,1), (0,2),
\end{split}
\label{open_free_energy}
\end{align}
where $z_*$ denotes a reference point.

The WKB reconstruction (\ref{wave_function}) of wave function is defined by $F_{n}^{(g)}(x)$'s as
\begin{align}
\psi_{\mathbb{C}\textbf{P}^{1}_{\bm{w}}}(x)=\exp\left(\sum_{g=0,n=1}^{\infty}\frac{\hbar^{2g-2+n}}{n!}F^{(g)}_{n}(x)\right).
\label{wkb_p1_ex}
\end{align}
We fix the reference point by $z_*=\infty$ which corresponds to $x(z_*)=\infty$. Here note that $F_{1}^{(0)}$ and $F_{2}^{(0)}$ need to be regularized by certain constant shifts so as to depend on $z_*$. Some explicit computational results of the free energies $F_n^{(g)}(x)$ are listed in Table \ref{tab:recursion_CP1}.
\begin{table}[t]
\begin{tabular}{|l|}
\hline
$F_n^{(g)}(x)$ \\
\hline \hline
$F_2^{(0)}(x)=-\frac12 \log\left[(x+\Lambda)/(x(z_*)+\Lambda)\right]$ \\
\hline
$F_3^{(0)}(x)=\mp\Lambda/(2(x+\Lambda)^{3/2})$ \\
$F_1^{(1)}(x)=\pm(3x+2\Lambda)/(48(x+\Lambda)^{3/2})$\\
\hline
$F_4^{(0)}(x)=\Lambda(-3x+\Lambda)/(4(x+\Lambda)^{3})$ \\
$F_2^{(1)}(x)=(3x^2-2\Lambda^2-6x\Lambda)/(96((x+\Lambda)^3)$ \\
\hline
$F_5^{(0)}(x)=\mp\Lambda(2\Lambda^2-21x\Lambda+12 x^2)/(8(x+\Lambda)^{9/2})$\\
$F_3^{(1)}(x)=\pm(4\Lambda^3+18x\Lambda^2-63\Lambda x^2+6x^3)/(192(x+\Lambda)^{9/2})$\\
$F_1^{(2)}(x)=\mp(186\Lambda x^2+72 x\Lambda^2+16\Lambda^3-45x^3)/(15360(x+\Lambda)^{9/2})$ \\
\hline
\end{tabular}
\caption{\label{tab:recursion_CP1}Free energies for the equivariant $\mathbb{C}\textbf{P}^{1}$ model. Here we have two types free energies corresponding to the branches $z=\pm \sqrt{x+\Lambda}$.}
\end{table}

The wave function (\ref{wkb_p1_ex}) is reorganized by
\begin{align}
\psi_{\mathbb{C}\textbf{P}^{1}_{\bm{w}}}(x)=\exp\left(\sum_{m=0}^{\infty}\hbar^{m-1}F_m(x)\right),\qquad
F_m(x)=\sum_{\substack{g\ge 0,\;n\ge 1,\\2g+n-1=m}}\frac{1}{n!}F_n^{(g)}(x),
\nonumber
\end{align}
where corresponding to two branches $z=\pm \sqrt{x+\Lambda}$ we find two types of free energies $F_m(x)=F_m^{(\pm)}(x)$. Since the leading term $S_0(x)$ of the asymptotic expansion  of the $J$-function obeys
\begin{align}
x\frac{dS_0(x)}{dx}=y(x),
\nonumber
\end{align}
the free energy $F_1^{(0)}(x)$ in (\ref{open_free_energy}) agrees with  $S_0^{(\pm)}(x)$ up to a constant shift. Using the computational results in Table \ref{tab:recursion_CP1}, $F_m^{(\pm)}(x)$'s ($m\ge 1$) are computed immediately and summarized in Table \ref{tab:Fm_eqv_CP1}. 
\begin{table}[t]
\begin{tabular}{|c|l|}
\hline
$m$ & $F_m^{(\pm)}(x)$ \\
\hline
$1$ & $-\frac{1}{4} \log\left[(x+\Lambda)/(x(z_*)+\Lambda)\right]$ \\
\hline
$2$ & $\pm\left(3x-2\Lambda\right)/\left(48\left(x+\Lambda\right)^{3/2}\right)$ \\
\hline
$3$ &$x(x-4\Lambda)/\left(64\left(x+\Lambda\right)^{3}\right)$\\
\hline
$4$ & $\pm\left(375x^3 - 3654 x^2 \Lambda + 1512 x \Lambda^2 +16\Lambda^3\right)/\left(46080\left(x+\Lambda\right)^{9/2}\right)$ \\
\hline
\end{tabular}
\caption{\label{tab:Fm_eqv_CP1}$F_m(x)=F_m^{(\pm)}(x)$ for the equivariant $\mathbb{C}\textbf{P}^{1}$ model.}
\end{table}
Comparing the computational results in Tables \ref{tab:eqv_CP1} and \ref{tab:Fm_eqv_CP1},
one finds the agreement 
\begin{align}
F_m^{(\pm)}(x)=S_m^{(\pm)}(x),
\nonumber
\end{align}
for $m=1,2,3,4$ up to a constant shift of $F_1^{(\pm)}(x)$.

\end{document}